\pdfoutput=1
\documentclass[11pt]{amsart}
\usepackage{amsaddr}
\usepackage{fullpage}
\usepackage[sort&compress,square,numbers]{natbib}
\usepackage[hang]{subfigure}
\usepackage{rotating}

\usepackage{multirow}

\usepackage{algorithm}
\usepackage{algpseudocode}
\usepackage{listings}
\usepackage{empheq}
\usepackage{multicol}

\usepackage{xspace}

\newcommand{\typeleaf}[1]{\type{#1} = \leaf}

\newcommand{\typeline}[1]{\type{#1} = \linenode}

\newcommand{\leaf}{\textsc{leaf}}

\newcommand{\linenode}{\textsc{line}}

\newcommand{\type}[1]{\textsc{type}(#1)}

\newcommand{\round}[1]{\textsc{round}(#1)}

\newcommand{\roundis}[2]{\round{#1} = #2}

\newcommand{\parent}[1]{\textsc{parent}(#1)}

\newcommand{\parentis}[2]{\parent{#1} = #2}

\newcommand{\lltree}{\textsc{Line-Leaf Tree}\xspace}

\newcommand{\transition}{\textsc{Transition}\xspace}
\newcommand{\here}{\textsc{here}\xspace}

\newcommand{\nil}{\textsc{Nil}\xspace}

\newcommand{\upc}{\textsc{Up Correct}\xspace}
\newcommand{\downc}{\textsc{Down Correct}\xspace}

\usepackage{pdfsync}

\newtheorem{theorem}{Theorem}
\newtheorem{corollary}[theorem]{Corollary}
\newtheorem{lemma}[theorem]{Lemma}
\newtheorem{property}[theorem]{Property}
\newtheorem{claim}[theorem]{Claim}

\theoremstyle{definition}
\newtheorem{definition}[theorem]{Definition}

\begin{document}

\title{Searching in Dynamic Tree-Like Partial Orders}

\author{Brent Heeringa}
\address{Department of Computer Science \\
Williams College}
\email{heeringa@cs.williams.edu}
\thanks{Brent Heeringa is supported by NSF grant IIS-08125414.}
\author{Marius C\u{a}t\u{a}lin Iordan}
\address{Department of Computer Science \\
Stanford University }
\email{mci@cs.stanford.edu}
\thanks{Marius C\u{a}t\u{a}lin Iordan is supported by the William R. Hewlett Stanford Graduate Fellowship.}
\author{Louis Theran}
\address{Department  of Mathematics. \\
Temple University}
\email{theran@temple.edu}
\thanks{Louis Theran is supported by CDI-I grant DMR 0835586 to Igor Rivin and M. M. J. Treacy.}

\begin{abstract}
We give the first data structure for the problem of maintaining a dynamic set of $n$ elements drawn from a partially ordered universe described by a tree.  We define the \lltree, a linear-sized data structure that supports the operations: {\em insert; delete; test membership; and predecessor}.  The performance of our data structure is within an $O(\log w)$-factor of optimal. Here $w \leq n$ is the width of the partial-order---a natural obstacle in searching a partial order.
\end{abstract}

\maketitle

\clearpage

\section{Introduction}

A fundamental problem in data structures is maintaining an ordered set $S$ of $n$ items drawn from a universe $\mathcal{U}$ of size $M \gg n$.  For a totally ordered $\mathcal{U}$, the dictionary operations: \emph{insert}; \emph{delete}; \emph{test membership}; and \emph{predecessor} are all supported in $O(\log n)$ time and $O(n)$ space in the comparison model via balanced binary search trees.  Here we consider the relaxed problem where $\mathcal{U}$ is partially ordered and give the first data structure for maintaining a dynamic partially ordered set drawn from a universe that can be described by a tree.

As a motivating example, consider an email user that has stockpiled years of messages into a series of hierarchical folders.  When searching for an old message, filing away a new message, or removing an impertinent message, the user must navigate the hierarchy.  Suppose the goal is to minimize, in the worst-case, the number of folders the user must consider in order to find the correct location in which to retrieve, save, or delete the message.  Unless the directory structure is completely balanced, an optimal search does not necessarily start at the top---it might be better to start farther down the hierarchy if the majority of messages lie in a sub-folder.  If we model the hierarchy as a rooted, oriented tree and treat the question ``is message $x$ contained somewhere in folder $y$?'' as our comparison, then maintaing an optimal search strategy for the hierarchy is equivalent to maintaining a {\em dynamic} partially ordered set under insertions and deletions.

\subsubsection*{{\bf Related Work}}

The problem of searching in trees and partial orders has recently received considerable attention.  Motivating this research are practical problems in filesystem synchronization, software testing and information retrieval~\cite{ben-asher-etal:sjc1999}.  However, all of this work is on the {\em static} version of the problem.  In this case, the set $S$ is fixed and a search tree for $S$ does not support the insertion or deletion of elements.
For example, when $S$ is totally ordered, the optimal minimum-height solution is a standard binary
search tree.  In contrast to the totally ordered case, finding a minimum height static search tree for an arbitrary partial order is
NP-hard~\cite{carmo-etal:tcs2004}.  Because of this, most recent work has focused on partial orders that can be described by rooted, oriented trees.  These are called {\em tree-like} partial orders in the literature.  For tree-like partial orders, one can find a minimum height search tree in linear time~\cite{mozes-etal:soda2008,onak-parys:focs2006,dereniowski-dariusz:dam2008}.  In contrast, the weighted version of the tree-like problem (where the elements have weights and the goal is to minimize the {\em average} height of the search tree) is NP-hard~\cite{jacobs-etal:icalp2010} although there is a constant-factor approximation~\cite{laber-molinaro:icalp2008}.  Most of these results operate in the edge query model which we review in Sec.~\ref{sec:model}.

Daskalakis et al. have recently studied the problem of \emph{sorting} partial orders~\cite{daskalakis-etal:soda2009,daskalakis-etal:arxiv2007} and, in~\cite{daskalakis-etal:arxiv2007}, ask for analogues of balanced binary search trees for dynamic partially ordered sets.  We are the first to address this question.

\subsubsection*{{\bf Rotations do not preserve partial orders.}}

\begin{figure}[tb]
\centering\includegraphics[scale=0.35]{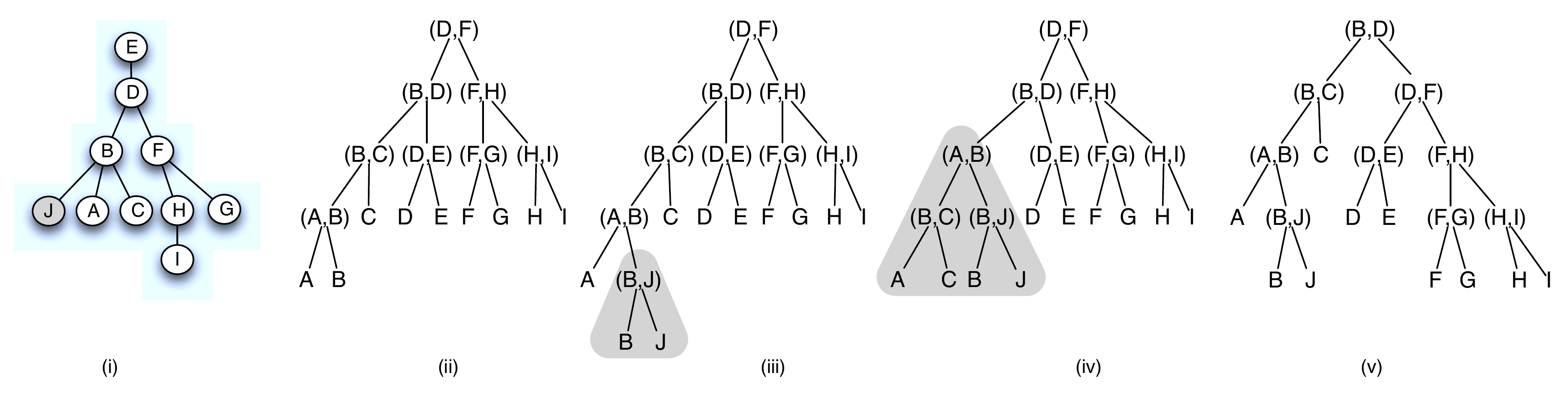}
\caption{(i) A partially ordered set $\{A, B, C, D, E, F, G, H, I, J\}$.  A downward path from node $X$ to node $Y$ implies $X \prec Y$.  Note that, for example, $E \prec F$ and $G$ and $I$ are incomparable.  (ii) An optimal search tree for the set $\{A, B, \ldots, I\}$.  For any query $(X,Y)$ an answer of $X$ means descend left and an answer of $Y$ means descend right.  (iii) After adding the element $J$, a standard search tree would add a new query $(B,J)$ below $(A,B)$ which creates an imbalance.  (iv) The search tree after a rotation; the subtree highlighted in grey is not a correct search tree for the partial order (i).  (v) An optimal search tree for the set $\{A, B, \ldots, J\}$.  \label{fig:rotations}}
\vspace{-7mm}
\end{figure}

Traditional data structures for dynamic ordered sets ({\em e.g.}, red black trees, AVL trees)  appear to rely on the total order of the data.  All these data structures use binary tree rotations
as the fundamental operations; applied in an unrestricted manner, rotations \emph{require} a totally ordered universe.  For example, consider Figure~\ref{fig:rotations} (ii) which gives an optimal search tree for the elements $\{A, B, \ldots, I\}$ depicted in the partial order of Figure~\ref{fig:rotations} (i).  If we insert node J (colored grey) then we must add a new test $(B,J)$ below $(A,B)$ which creates the sub-optimal search tree depicted in Figure~\ref{fig:rotations} (iii).  Using traditional rotations yields the search tree given in Figure~\ref{fig:rotations} (iv) which does not respect the partial order; the leaf marked $C$ should appear under the right child of test $(A,B)$.  Figure~\ref{fig:rotations} (v) denotes a correct optimal search for the set $\{A, B, \ldots, J\}$.
The key observation is that, if we imagine the leaves of a binary search tree for a total order partitioning the real line,
rotations preserve the order of the leaves, but not any kind of subtree relations on them.  As a consequence, blindly applying rotations to a search tree for the static problem does not yield a viable dynamic data structure.   To sidestep this problem, we will, in essence, decompose the tree-like partial order into totally ordered chains and totally incomparable stars.

\subsubsection*{{\bf Techniques and Contributions}}

We define the \lltree, the first data structure that supports the fundamental dictionary operations for a \emph{dynamic} set $S \subseteq \mathcal{U}$ of $n$ elements drawn from a universe equipped with a \emph{partial order} $\preceq$ described by a rooted, oriented tree.

Our dynamic data structure is based on a static construction algorithm that takes as input the {\em Hasse diagram} induced by $\preceq$ on $S$ and in $O(n)$ time and space produces a \lltree for $S$.  The Hasse diagram $H_{S}$ for $S$ is the directed graph that has as its vertices the elements of $S$ and a directed edge from $x$ to $y$ if and only if $x \prec y$ and no $z$ exists such that $x \prec z \prec y$.  We build the \lltree inductively via a natural contraction process which starts with $H_{S}$ and, ignoring the edge orientations, repeatedly performs the following two steps until there is a single node:
\begin{enumerate}
\item Contract paths of degree-two nodes into balanced binary search trees (which we can binary search efficiently); and
\item Contract leaves into linear search structures associated with their parents (which are natural search structures since the children of an interior node are mutually incomparable).
\end{enumerate}
One of these steps always applies in our setting since $H_{S}$ is a rooted, oriented tree.  We give an example of each step of the construction in Figure~\ref{fig:bstlst}.   We show that the contraction process yields a search tree that is provably within an $O(\log w)$-factor of the minimum-height search tree for $S$.  The parameter $w$ is the {\em width} of $S$---the size of the largest subset of mutually incomparable elements of $S$---which represents a natural obstacle when searching a partial order.  Our construction algorithm and analysis appear in Section~\ref{sec:construction}.

The intuition behind the proof of the approximation ratio is that an optimal search tree for any minor of $H_{S}$ gives a lower bound on an optimal search tree for $H_{S}$.  Since optimal search trees are easy to describe for paths of degree-two nodes as well as for stars, the approximation ratio follows by bounding the number of rounds in the contraction process.  We also show that our analysis is tight.

\begin{figure}[tb]
\centering\includegraphics[scale=0.17]{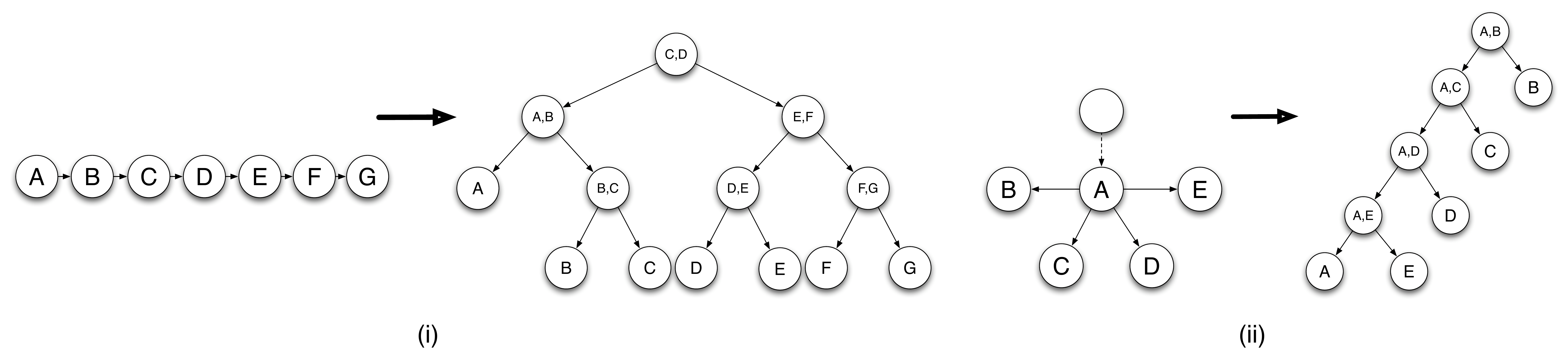}
\caption{Examples of (i) a line contraction where we build a balanced binary search tree from a path and (ii) a leaf contraction where we build a linear search tree from the leaves of a node. \label{fig:bstlst}}
\end{figure}

To make the $\lltree$ \emph{fully dynamic}, in Section~\ref{sec:operations} we give procedures  to update it under insertions and deletions.  All the operations, take $O(\log w) \cdot OPT$ comparisons and RAM operations where $OPT$ is the height of a minimum-height search tree for $S$. Additionally, \emph{insertion} requires only $O(h)$ comparisons, where $h$ is the height of the \lltree being updated.  (The non-restructuring operations \emph{test membership} and \emph{predecessor} also require at most $O(h)$ comparisons since the \lltree is a search tree).  Because $w$ is a property of $S$, in the dynamic setting it changes under insertions and deletions.  However, the \lltree maintains the $O(\log w)\cdot OPT$ height bound \emph{at all times}.  This means it is well-defined to speak of the $O(\log w)\cdot OPT$ upper bound without mentioning $S$.

The insertion and deletion algorithms maintain the invariant that the updated {\sc Line-Leaf} {\sc Tree} is structurally equivalent to the one that we would have produced had the static construction algorithm been applied to the updated set $S$.  In fact, the heart of insertion and deletion is {\em correcting} the contraction process to maintain this invariant.  The key structural property of a \lltree---one that is not shared by constructions for optimal search trees in the static setting---is that its sub-structures essentially represent either paths or stars in $S$, allowing for updates that make only local changes to each component search structure.  The $O(\log w)$-factor is the price we pay for the additional flexibility.  The dynamic operations, while conceptually simple, are surprisingly delicate.  As with many data structures, our proofs perform a case analysis which mimics the underlying algorithmic definitions of {\sc Insert} and {\sc Delete} respectively.

In Section~\ref{sec:empirical} we provide empirical results on both random and real-world data that show the \lltree is strongly competitive with the static optimal search tree.

\section{Models and Definitions} \label{sec:model}

Let $\mathcal{U}$ be a finite set of $M$ elements and let $\preceq$ be a partial order, so the pair $(\mathcal{U}, \preceq)$ forms a \emph{partially ordered set}.  We assume the answers to $\preceq$-queries are provided by an oracle. (Daskalakis, et al.~\cite{daskalakis-etal:soda2009} provide a space-efficient data structure to answer $\preceq$-queries in $O(1)$ time.)

In keeping with previous work, we say that $\mathcal{U}$ is {\em tree-like} if $H_{\mathcal{U}}$ forms a rooted, oriented tree.  Throughout the rest of this paper, we assume that $\mathcal{U}$ is tree-like and refer to the vertices of $H_{\mathcal{U}}$ and the elements of $\mathcal{U}$ interchangeably. For convenience, we add a dummy minimal element $\nu$ to $\mathcal{U}$.  Since any search tree for a set $S \subseteq \mathcal{U}$ embeds with one extra comparison into a corresponding search tree for $S \cup \{\nu\}$, we assume from now on that $\nu$ is always present in $S$.

Given these assumptions it is easy to see that tree-like partial orders have the following properties:

\begin{property} \label{prop:tree-like}
Any subset $S$ of a tree-like partially ordered universe $\mathcal{U}$ is also tree-like.
\end{property}

\begin{property} \label{prop:pred}
Every non-root element in a tree-like partially ordered set $S \subseteq \mathcal{U}$ has exactly one predecessor in $H_{S}$.
\end{property}

\begin{figure}[tb]
\begin{center}
\includegraphics[scale=.5]{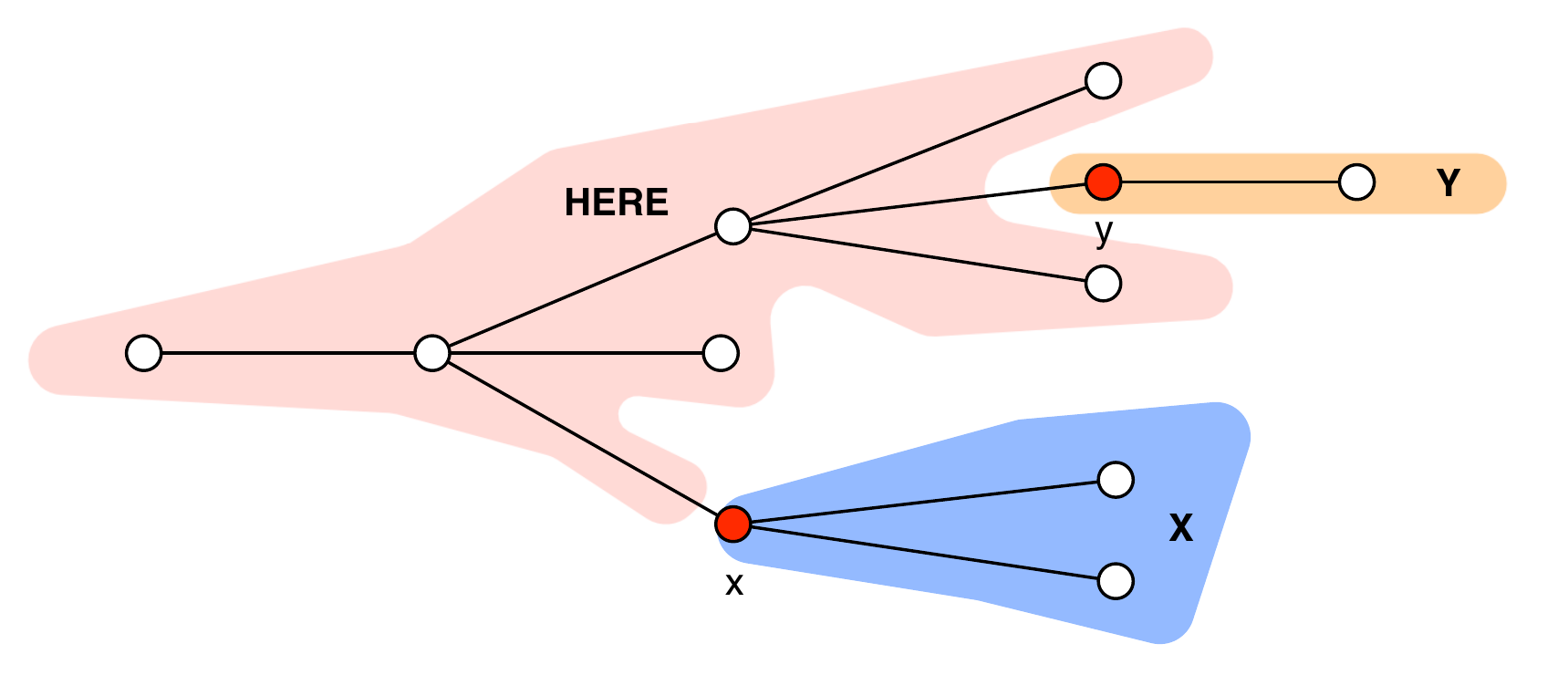}
\caption{\label{fig:dynamic-edge-queries} Given two nodes $x$ and $y$ in $S$ and a third node $u \in \mathcal{U}$, a dynamic edge query on $(x,y)$ with respect to $u$ can answer (i) {\sc y}, in which case $u$ falls somewhere in the shaded
area labelled Y; (ii) {\sc x}, in which case $u$ falls somewhere in the shaded area labelled X; or (iii) \here, in which case $u$ falls somewhere in the shaded area labelled HERE.  Notice that if $(x,y)$ forms an actual edge then the query reduces to a standard edge query}.
\end{center}
\end{figure}

Let $T_S$ be the undirected (but still rooted and oriented) Hasse diagram for $S$.

We extend edge queries to {\em dynamic edge queries} by allowing queries on arbitrary pairs of nodes in $T_{S}$ instead of just edges in $T_{S}$.

\begin{definition}[\textbf{Dynamic Edge-queries}]  \label{def:edge-query}
Let $u$ be an element in $\mathcal{U}$ and $x$ and $y$ be nodes in $T_{S}$.  Let $S' = S \cup \{ u \}$ and consider the edges $(x,x')$ and $(y,y')$ bookending the unique path from $x$ to $y$ in $T_{S'}$.  Define $T_{S'}^x$, $T_{S'}^y$ and $T_{S'}^{\here}$ to be the three connected components of $T_{S'} \setminus \{(x,x'), (y,y')\}$ containing $x$, $y$, and neither $x$ nor $y$, respectively.  A {\em dynamic edge query} on $(x,y)$ with respect to $u$ has one of the following three answers:
\begin{enumerate}
\item {\sc x}:   if $u \in T_{S'}^x$ ($u$ equals or is closer to $x$)
\item {\sc y}:   if $u \in T_{S'}^y$ ($u$ equals or is closer to $y$)
\item \here:    if $u \in T_{S'}^{\here}$ ($u$ falls between, but is not equal to either, $x$ or $y$)
\end{enumerate}
\end{definition}

Figure~\ref{fig:dynamic-edge-queries} gives an example of a dynamic edge query.   Any dynamic edge query can be simulated by $O(1)$ standard comparisons when $H_{S}$ is tree-like.  This is not the case for more general orientations of $H_S$ and an additional data structure is required to implement either our algorithms or algorithms of~\cite{mozes-etal:soda2008,onak-parys:focs2006}.  Thus, for a tree-like $S$, the height of an optimal search tree in the dynamic edge query model and the height of an optimal decision tree for $S$ in the comparison model are always within a small constant factor of each other.  For the rest of the paper, we will often drop {\em dynamic} and refer to {\em dynamic edge queries} simply as {\em edge queries}.

\section{\lltree Construction and Analysis} \label{sec:construction}

We build a $\lltree$ $\mathcal{T}$ inductively via a contraction process on $T_{S}$.   Each contraction step builds a {\em component search structure} of the $\lltree$. These component search structures are either linear search trees or balanced binary search trees.  A linear search tree $LST(x)$ is a sequence of dynamic edge queries, all of the form $(x,y)$ where $y \in S$, that ends with the node $x$.  A balanced binary search tree $BST(x,y)$ for a path of contiguous degree-2 nodes between, but not including, $x$ and $y$ is a tree that {\em binary searches} the path using edge queries.

Let $T_0 = T_{S}$.  If the contraction process takes $m$ iterations total, then the final result is a single node which we label $\mathcal{T}=T_{2m}$.    In general, let $T_{2i-1}$ be the partial order tree after the line contraction of iteration $i$ and $T_{2i}$ be the partial order tree after the leaf contraction of iteration $i$ where $i \geq 1$.  We now show how to construct a \lltree for a fixed tree-like set $S$.

\begin{description}

\item [Base Cases]  Associate an empty balanced binary search tree $BST(x,y)$ with every actual edge $(x,y)$ in $T_0$.   Associate a linear search tree $LST(x)$ with every node $x$ in $T_0$.  Initially, $LST(x)$ contains just the node itself.

\item [Line Contraction] Consider the line contraction step of iteration $i \geq 1$: If $x_2, \dots, x_{t-1}$ is a path of contiguous degree-2 nodes in $T_{2(i-1)}$ bounded on each side by non-degree-2 nodes $x_1$ and $x_t$ respectively, we contract this path into a balanced binary search tree $BST(x_{1},x_{t})$ over the nodes $x_2, \dots, x_{t-1}$. The result of the path contraction is an edge labeled $(x_1, x_t)$. This edge yields a dynamic edge query.

\item[Leaf Contraction] Consider the leaf contraction step of iteration $i \geq 1$: If $y_1, \dots, y_{t}$ are all degree-1 nodes in $T_{2i-1}$ adjacent to a node $x$ in $T_{2i-1}$, we contract them into the linear search tree $LST(x)$ associated with $x$.  Each node $y_{j}$ contracted into $x$ adds a dynamic edge query $(x,y_{j})$ to $LST(x)$.  If nodes were already contracted into $LST(x)$ from a previous iteration, we add the new edge queries to the front (top) of the LST.
\end{description}

\noindent After $m$ iterations we are left with $\mathcal{T}=T_{2m}$ which is a single node.  This node is the root of the \lltree.

\subsection{Example Construction}

\begin{figure}[tb]
\begin{center}
\includegraphics[scale=0.3]{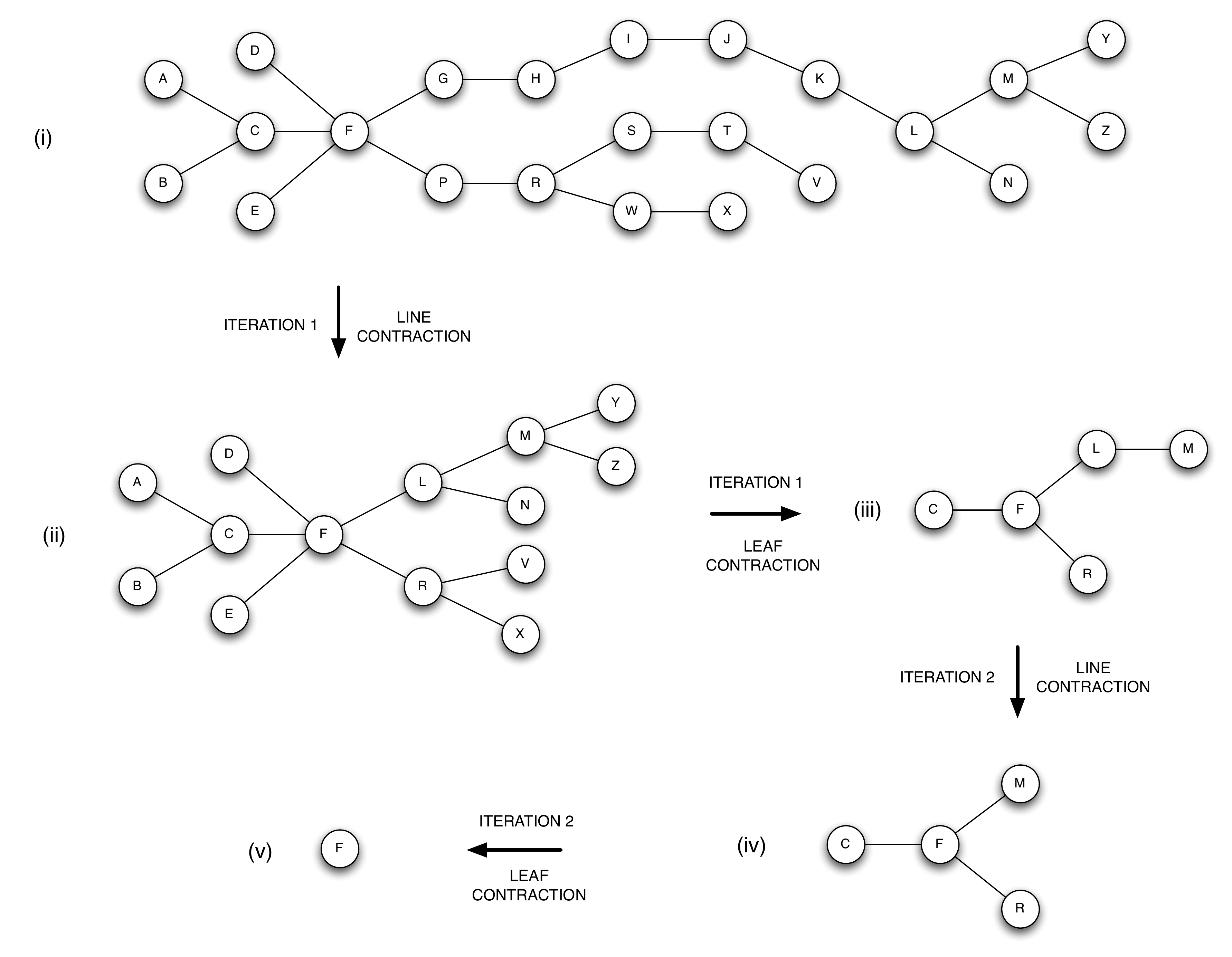}
\end{center}
\caption{The set $S$ under several iterations of the construction algorithm. \label{fig:construction-process}}
\end{figure}

\begin{figure}[tb]
\begin{center}
\includegraphics[scale=0.35]{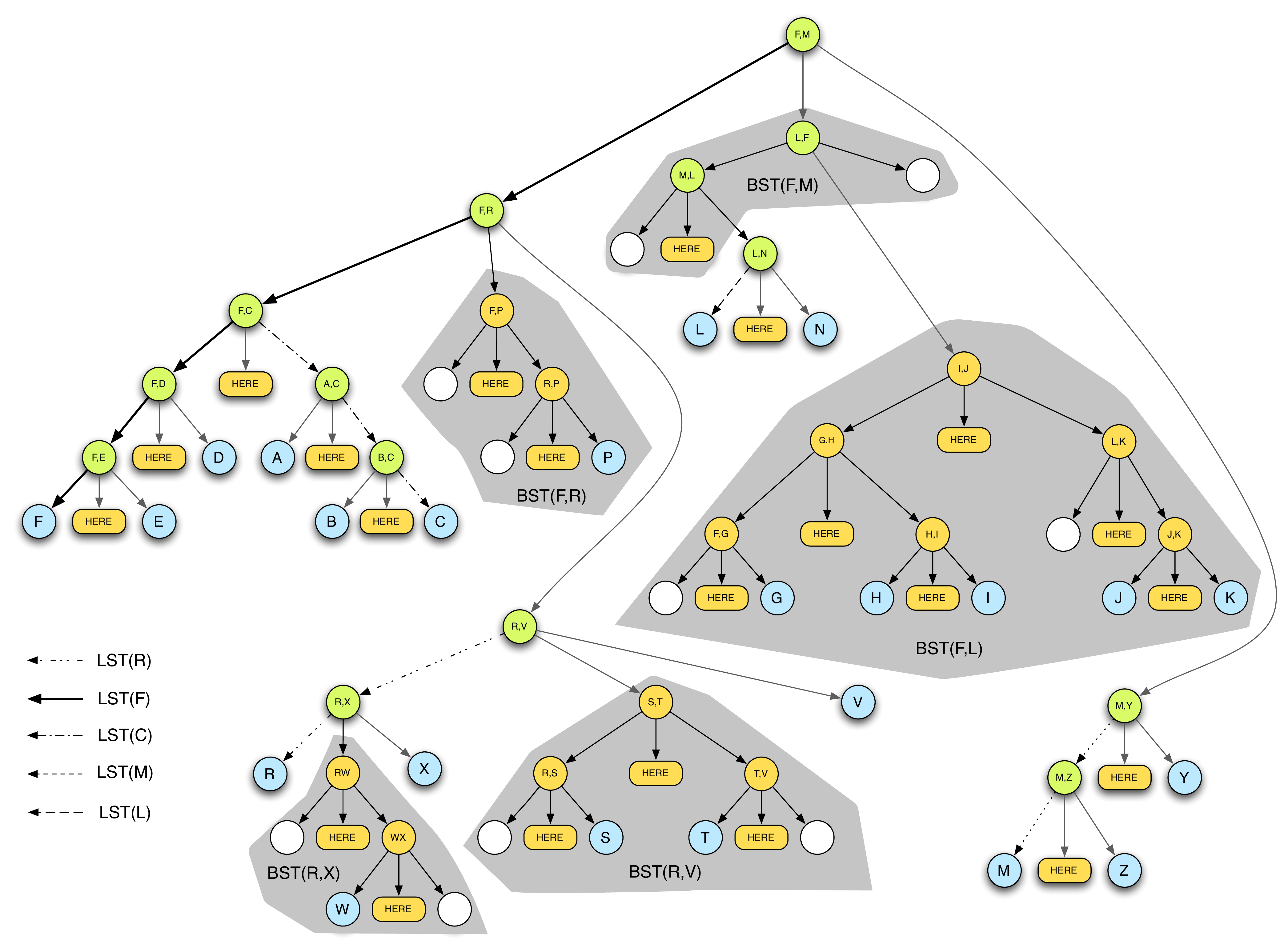}
\caption{\lltree for the set $S$ given in Figure~\ref{fig:construction-process}. The empty nodes indicate impossible answers. \label{fig:lltree}}
\end{center}
\end{figure}

Here we provide an example $\lltree$ construction for a partial order on a set $S$ with 23 elements.  Figure~\ref{fig:construction-process} shows $T_{S}$ after each step of each round of the contraction process.  Figure~\ref{fig:lltree} shows the final $\lltree$.

Suppose $S$ has the tree structure illustrated in Figure~\ref{fig:construction-process} (i). We associate an empty balanced binary search tree (BST) with every edge in $T_0 = T_{S}$ and a linear search tree (LST) comprised of only the node itself with every node in $T_0$. The first path contraction creates BSTs for the chains $\{ G, H, I, J, K \}$, $\{ P \}$, $\{ S, T \}$, $\{ W \}$, and associates them with the edges $(F, L)$, $(F, R)$, $(R, V)$, $(R, X)$, respectively. We obtain the tree in Figure~\ref{fig:construction-process} (ii).

The first iteration ends with a leaf contraction step that adds collections of leaves $\{ A, B \}$, $\{ D, E \}$, $\{ N \}$, $\{ V, X \}$, $\{ Y, Z \}$ to the LSTs of elements $C$, $F$, $L$, $R$, $M$, respectively.  This yields the tree in Figure~\ref{fig:construction-process} (iii).

At this point, the next path contraction creates a BST for the single-element chain $\{ L \}$ and associates this BST with edge $(F,M)$. Finally, the second leaf contraction reduces the tree to a single node by contracting the final leaves $C$, $M$, $R$ into the LST of node $F$ as shown in Figure~\ref{fig:construction-process} (v).  This ends the construction process.

Notice that in Figure~\ref{fig:lltree} some answers to edge queries are left empty.  We call these {\em impossible} answers.  This happens because the answer \here to an edge query $(x,y)$ implies that the node $u$ we seek is not equal to either $x$ or $y$, but rather lies between them. However, if there is at least one node on the path between $x$ and $y$, we need to ask the questions of the edges adjacent to nodes $x$ and $y$ on that path in order to determine whether $u$ should be placed between two elements.  Such a question cannot answer $x$ or $y$ since the \here answer eliminated this possibility. Thus these choices are impossible.

\subsection{Searching a \lltree}

Searching a $\lltree$ for an element $u$ is tantamount to searching the component search structures.  A search begins with $LST(x)$ where $x$ is the root of $\mathcal{T}$.   Searching $LST(x)$ with respect to $u$ serially questions the edge queries in the sequence.  Starting with the first edge query, if $(x,y)$ answers {\sc x} then we move onto the next query $(x,z)$ in the sequence.  If the query answers \here then we proceed by searching for $u$ in $BST(x,y)$.  If it answers {\sc y}, then we proceed by searching for $u$ in $LST(y)$.  If there are no more edge queries left in $LST(x)$, then we return the actual element $x$.   When searching $BST(x,y)$, if we ever receive a \here response to the edge query $(a,b)$, we proceed by searching for $u$ in $BST(a,b)$.  That is, we leave the current BST and search in a new BST.  If the binary search concludes with a node $x$, then we proceed by searching $LST(x)$.  Searching an empty BST returns \nil.

\subsection{Implementation Details}

The $\lltree$ is an index into $H_{S}$ but not a replacement for $H_{S}$.  That is, we maintain a separate DAG data structure for $H_{S}$ across insertions and deletions into $S$.  This allows us, for example, to easily identify the predecessor and successors of a node $x \in S$ once we've used the \lltree to find $x$ in $H_{S}$.  The edges of $H_{S}$ also play an essential role in the implementation of the \lltree.  Namely, an edge query $(x,y)$ is actually two pointers: $\lambda_{1}(x,y)$ which points to the edge $(x,a)$ and $\lambda_{2}(x,y)$ which points to the edge $(b,y)$.  Here $(x,a)$ and $(b,y)$ are the actual edges bookending the undirected path between $x$ and $y$ in $T_{S}$.  This allows us to take an actual edge $(x,a)$ in memory, rename $x$ to $w$, and indirectly update all edge queries $(x,z)$ to $(w,z)$ in constant time.  Here the path from $z$ to $x$ runs through $a$.  Note that we are not touching the pointers involved in each edge query $(x,z)$, but rather, the actual edge in memory to which the edge query is pointing.

Edge queries are created through line contractions so when we create the binary search tree $BST(x,y)$ for the path $x,a,\ldots,b,y$, we let $\lambda_{1}(x,y) = \lambda_{1}(x,a)$ and $\lambda_{2}(x,y) = \lambda_{2}(b,y)$.  We assume that every edge query $(x,y)$ corresponding to an actual edge $(x',y')$ has $\lambda_{1}(x,y)=\lambda_{2}(x,y)=(x',y')$.

\subsection{Node Properties} \label{sec:node-props}

We associate two properties with each node in $S$.  The {\em round} of a node $x$ is the iteration $i$ where $x$ was contracted into either an LST or a BST. We say \roundis{$x$}{$i$}.
The {\em type} of a node represents the step where the node was contracted. If node $x$ was {\it line contracted}, we say \typeline{$x$}, otherwise we say \typeleaf{$x$}.

In addition to {\sc round} and {\sc type}, we assume that both the linear and binary search structures provide a {\sc parent} method that operates in time proportional to the height of the respective data structure and yields either a node (in the case of a leaf contraction) or an edge query (in the case of a line contraction).  More specifically, if node $x$ is leaf contracted into $LST(a)$ then \parentis{$x$}{$a$}.  If node $x$ is line contracted into $BST(a,b)$ then \parentis{$x$}{$(a,b)$}.  We emphasize that the {\sc parent} operation here refers to the \lltree and not $T_{S}$.  Collectively, the {\sc round}, {\sc type}, and {\sc parent} of a node help us recreate the contraction process when inserting or removing a node from $S$.

\subsection{Approximation Ratio}

The following theorem gives the main properties of the static construction.

\begin{theorem} \label{thm:height}
The worst-case height of a \lltree $\mathcal{T}$ built from a tree-like $S$ is $\Theta(\log w) \cdot OPT$ where $w$ is the width of $S$ and $OPT$ is the height of an optimal search tree for $S$.  In addition, given $H_{S}$, $\mathcal{T}$ can be built in $O(n)$ time and space.
\end{theorem}

\begin{proof}
We begin with some lower bounds on $OPT$.

\begin{claim}
$OPT \geq \max \{ \Delta(S), \log n, \log D, \log w \}$ where $\Delta(S)$ is the maximum degree of a node in $T_{S}$, $n$ is the size of $S$, $D$ is the diameter of $T_{S}$ and $w$ is the width of $S$.
\end{claim}

\begin{proof}
Let $x$ be a node of highest degree $\Delta(S)$ in $T_{S}$. Then, to find $x$ in the $T_{S}$ we require at least $\Delta(S)$ queries, one for each edge adjacent to $x$~\cite{laber-nogueira:endm2001}. This implies $OPT \geq \Delta(S)$. Also, since querying any edge reduces the problem space left to search by at most a half, we have $OPT \geq \log n$.  Because $n$ is an upper bound on both the width $w$ of $S$ and $D$, the diameter of $T_S$ we obtain
the final two lower bounds. 
\end{proof}

Recall that the width $w$ of $S$ is the number of leaves in $T_S$.  Each round in the contraction process reduces the number of remaining leaves by at least half:  round $i$ starts with a tree $T_{2i}$ on $n_i$ nodes with $w_i$ leaves. A line-contraction produces a tree $T_{2i+1}$, still with $w_i$ leaves.  Because $T_{2i+1}$ is full, the number of nodes neighboring a leaf is at most $w_i/2$.  Round $i$ completes with a leaf contraction that removes all $w_i$ leaves, producing $T_{2i+2}$.  As every leaf in $T_{2i+2}$ corresponds to an internal node of $T_{2i+1}$ adjacent to a leaf,  $T_{2i+2}$ has at most $w_i/2$ leaves.  It follows that the number of rounds is at most $\log w$.  The length of any root-to-leaf path is bounded in terms of the number of rounds.

\begin{lemma} \label{lemma:path}
On any root-to-leaf path in the \lltree there is at most one BST and one LST for each iteration $i$ of the construction algorithm.
\end{lemma}

\begin{proof}
On a root-to-leaf path, the \lltree contains LST and BST data structures in decreasing order of the iteration $i$ since the data structure is built incrementally from the bottom up. Suppose we are currently in $LST(a)$. The search structures immediately accessible from this point (aside from ourselves) are:
\begin{itemize}
\item $LST(b)$ for all queries $(a,b) \in LST(a)$
\item $BST(a, c)$ for all queries $(a,c) \in LST(a)$
\end{itemize}
If $(a,b) \in LST(a)$, then \typeleaf{$b$} and so \round{$b$} $<$ \round{$a$} by construction. If $d$ is a node in $BST(a, c)$, then \round{$d$} $\leq$ \round{$c$} $<$ \round{$a$} since $d$ was line contracted before $c$ was leaf contracted into $LST(a)$.  Now suppose we are currently in $BST(a,b)$. All nodes $c$ contracted into this BST have equal {\sc round} $j$ by construction. The next accessible search structures are:
\begin{itemize}
\item $BST(d,e)$ for all edge queries $(d,e) \in BST(a,b)$
\item $LST(c)$ for each leaf $c$ of $BST(a,b)$ (this $LST$ may consist of only node $c$)
\end{itemize}
If $f$ is a node in $BST(d,e)$, then \round{$f$} $< j$ since $f$ was line contracted before all nodes in $BST(a,b)$ (otherwise, $f$ would be in $BST(a,b)$). If $c$ is a node in $BST(a,b)$ then \round{$c$} = $j$.

Finally, consider a root-to-leaf path. Suppose at some point we are in $LST(a)$ and the next search structure we enter is $LST(b)$. It follows from above arguments that \round{$a$} is strictly smaller than \round{$b$}. Suppose at some point we are in $BST(c,d)$ and the next structure on the path is $BST(e,f)$. Then for all nodes $g$ line contracted into $BST(c,d)$ and all nodes $h$ line contracted into $BST(e,f)$, we have \round{$g$} strictly smaller than \round{$h$} and this concludes our proof. 
\end{proof}

For each LST we perform at most $\Delta(S)$ queries. In each BST we ask at most $O(\log D)$ questions.  By the previous lemma, since we search at most one BST and one LST for each iteration $i$ of the contraction process and since there at most $\log w$ iterations, it follows that the height of the \lltree is bounded above by:
\(
(\Delta(S) + O(\log D)) \log w = O(\log w) \cdot OPT
\).

\begin{figure}[t!]
\centering\includegraphics[scale=0.4]{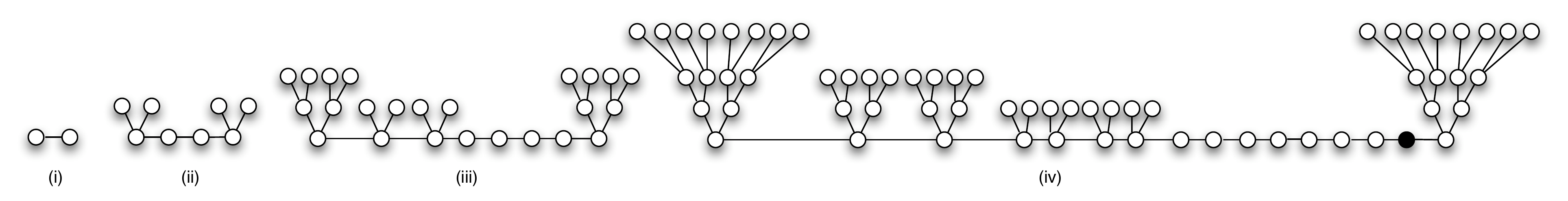}
\caption{ A tight example of our analysis:  starting with a single node (i) we grow the partial order tree (ii)--(iv) by adding nodes both horizontally and vertically.  The contraction process of our \lltree construction algorithm reverses the process that grows the tree. \label{fig:tight-analysis}}
\end{figure}

Now we show that in the worst case, the height of $\mathcal{T}$ is at least $\Omega(\log w) \cdot OPT$. Consider growing a partial-order tree $\mathcal{T}$ both vertically and horizontally according the process depicted in Figure~\ref{fig:tight-analysis}:  call a node {\em free} if it has no edge moving in the vertical direction.  If Figure~\ref{fig:tight-analysis} (i) depicts the tree after iteration 1, Figure~\ref{fig:tight-analysis} (ii) depicts the tree after iteration 2, and so on, then in iteration $k>1$ we add $(k+1)2^{k-1}$ nodes to the tree from iteration $k-1$ according the following rules:
\begin{itemize}
\item add 2 children to each of the $k2^{k-2}$ free nodes in the vertical direction.
\item add $2^{k-1}$ new free nodes just to the left of the rightmost node in the horizontal direction (these new nodes collectively form the $k^{th}$ {\em base tree}).
\end{itemize}
Thus, after $k$ iterations there are $N=\sum_{j=1}^{k} (j+1)2^{j-1} = \Theta(k2^{k})$ nodes.  Since $c \log N -  \log\log N < k < \log N $ for any $k>1$ and constant $c>1$ we have $k$ is $\Theta(\log N)$.  Note also that the width $w$ of $\mathcal{T}$ is $\Theta(N)$.  An optimal search tree for $\mathcal{T}$ uses $k$ edge queries to narrow the search down to one of the $k+1$ base trees and then uses an additional $O(k)$ queries to binary search that base tree.  This binary search is possible because in a tree with constant maximum degree, there is always an edge that cuts the tree into pieces of size at least $\frac{1}{3}n$.  Thus $OPT=\Theta(k)$.  However, the contraction process on $\mathcal{T}$ that inductively defines the \lltree results in a sequence of minors that essentially reverses the process of growing $\mathcal{T}$.  For example, line- and leaf-contracting Figure~\ref{fig:tight-analysis} (iv) yields Figure~\ref{fig:tight-analysis} (iii).  Thus, in the unfortunate case that the node we desire is the node just to the right of the rightmost node on the horizontal line (i.e. the black node in the figure), the \lltree must binary search the horizontal components of each of the $k$ base trees.  In other words, it must perform $\sum_{i=1}^{k} O(\log(2^{i}))=O(k^{2})$ edge queries.  Thus the height of the $\lltree$ is at least within a factor of $\Theta(k)=\Theta(\log N)=\Theta(\log w)$ of the height of the optimal static search tree.

We know prove the time and space bounds.  Consider the line contraction step at iteration $i$:   we traverse $T_{2(i-1)}$, labeling paths of contiguous degree-2 nodes and then traverse the tree again and form balanced BSTs over all the paths.  Since constructing balanced BSTs is a linear time operation, we can perform a complete line contraction step in time proportional to the size of size of $T_{2(i-1)}$.  Now consider the leaf contraction step at iteration $i$:  We add each leaf in $T_{2i-1}$ to the LST corresponding to its remaining neighbor.  This operation is also linear in the size of $T_{2i-1}$.  Since we know the size of $T_{2i}$ is halved after each iteration, starting with $n$ nodes in $T_0$, the total number of operations performed is $\sum_{i=0}^{\log n} O(\frac{n}{2^{i}})=O(n)$.

Given that the construction takes at most $O(n)$ time, the resulting data structure occupies at most $O(n)$ space.
\end{proof}

Theorem~\ref{thm:height} assumes that the Hasse diagram for $S$ is already constructed.  If this is not the case, then one must {\em sort} the partial order, which, for general partial orders requires $\Omega(n(\log n + w))$ comparisons~\cite{daskalakis-etal:soda2009,daskalakis-etal:arxiv2007}.  We are unaware of any work that directly addresses the sorting question for tree-like partial orders.

\section{Operations} \label{sec:operations}

\subsection{Test Membership} \label{sec:test-membership}
To test whether an element $A \in \mathcal{U}$ appears in $\mathcal{T}$, we search for $A$ in $LST(x)$ where $x$ is the root of $\mathcal{T}$.  The search ends when we reach a terminal node.  The only terminal nodes in the \lltree are either leaves representing the elements of $S$ or \nil (which are empty BSTs).  So, if we find $A$ in $\mathcal{T}$ then {\sc test membership} returns {\sc True}, otherwise it returns {\sc False}. Given that {\sc test membership} follows a root-to-leaf path in $\mathcal{T}$, the previous discussion constitutes a proof of the following theorem.

\begin{theorem} \label{theorem:test-membership}
{\sc Test Membership} takes $O(h)$ time.
\end{theorem}

\subsection{Predecessor} \label{sec:test-membership}

Property~\ref{prop:tree-like} guarantees that each node $A \in \mathcal{U}$ has exactly one predecessor in $S$.  Finding the predecessor of $A$ in $S$ is similar to {\sc test membership}.  We search $\mathcal{T}$ until we find either $A$ or \nil.  Traditionally if $A$ appears in a set then it is its own predecessor, so, in the first case we simply return $A$.  In the latter case, $A$ is not in $\mathcal{T}$ and \nil corresponds to an empty binary search tree $BST(y,z)$ for the actual edge $(y,z)$ where, say, $y \prec z$.  We know that $A$ falls between $y$ and $z$ (and potentially between $y$ and some other nodes) so $y$ is the predecessor of $A$.  We return $y$.  Given that {\sc predecessor} also follows a root-to-leaf path in $\mathcal{T}$, the previous discussion yields a proof of the following theorem.

\begin{theorem} \label{theorem:predecessor}
{\sc Predecessor} takes $O(h)$ time.
\end{theorem}

\subsection{Insert} \label{sec:insert}

Let $A$ be the node we wish to insert in $\mathcal{T}$ and let $S'=S \cup \{A\}$.  Our goal is transform $\mathcal{T}$ into $\mathcal{T'}$ where $\mathcal{T'}$ is the \lltree produced through the contraction process when started on $T_{S'}$.  We divide {\sc insert} into three corrective steps:  {\em local correction}, {\em down correction}, and {\em up correction} which we describe below.  Local correction repairs the contraction process for elements that appear near $A$ during the contraction process.   Down correction repairs $\mathcal{T}$ for nodes with round at most $\round{A}$.  Up correction repairs $\mathcal{T}$ for nodes with round at least $\round{A}$.

We begin with some notation.  Let $X$ be a node such that $LST(X)$ has $t$ edge queries $(X,Y_{1}) \ldots (X,Y_{t})$ sorted in descending order by \round{$Y_{i}$}.  In other words, $Y_{1}$ is the last node leaf-contracted into $LST(X)$, $Y_{t}$ is the first node leaf-contracted into $X$ and $Y_{i}$ is the $(t-i+1)^{th}$ node contracted into $LST(X)$.  Define $\rho_{i}(X)=Y_{i}$ and $\mu_{i}(X)=\round{Y_{i}}$.  That is, $\rho_{i}(X)$ yields the $(t-i+1)^{th}$ node contracted into $LST(X)$ and $\mu_{i}(X)$ yields the {\sc round} of $\rho_{i}(X)$. If $i>t$ then $\mu_{i}(X)=0$.

The following lemma relates the {\em type} of a node to the {\em rounds} of the nodes contracted into it.

\begin{lemma} \label{lemma:value-lemma}
Let $B$ be a node in a \lltree such that \roundis{$B$}{$k$}.
\begin{enumerate}
\item If \parentis{$B$}{{\sc null}} then either $\mu_1(B) = k-1 > k-2 = \mu_2(B) = \mu_3(B) \geq \mu_4(B)$ or $\mu_1(B) = \mu_2(B) = \mu_3(B) = k-1 \geq \mu_4(B)$.
\item If $\typeleaf{B}$ then $\mu_1(B) = \mu_2(B) = k-1 \geq \mu_3(B)$.
\item If $\typeline{B}$ then $\mu_1(B) = k-1 \geq \mu_2(B)$.
\end{enumerate}
\end{lemma}

\begin{proof}
The proof follows from the contraction process. If~\parentis{$B$}{{\sc null}} then $B$ is a full node at iteration $k-1$ and is the sole remaining node at iteration $k$, or $B$ has degree 1 at iteration $k-1$ and we arbitrarily made it root.  If $\typeleaf{B}$, then $B$ was not contracted at iteration $k-1$, it was a full node. But $B$ is leaf contracted at iteration $k$, thus it has degree 1. Therefore, at least two nodes were leaf contracted into $LST(B)$ at iteration $k-1$.  If $\typeline{B}$, then at iteration $k-1$, $B$ was a full node. But $B$ is line contracted at iteration $k$, thus it has degree 2. Therefore, at least one node was leaf contracted into $LST(B)$ at iteration $k-1$.
\end{proof}

\subsubsection{Local Correction} \label{sec:localcorrection}

We start by finding the predecessor of $A$ in $T_{S}$.  Call this node $B$.  We refer to $B$ as the {\em insertion point}.  $A$ potentially falls between $B$ and any number of its $children(B)$.  That is, $A$ may replace $B$ as the parent of a set of nodes $D \subseteq children(B)$.  We emphasize that the {\em parent} and {\em child} relationship here is over $H_{S}$ and not the \lltree $\mathcal{T}$.  We use $D$ to identify two other sets of nodes $C$ and $L$.  The set $C$ represents nodes that, in $T_{S}$, were leaf-contracted into $B$ in the direction of some edge $(B,D_{j})$ where $D_{j} \in D$.   The set $L$ represents nodes that were involved in the contraction process of $B$ itself.  Depending on {\sc type}$(B)$ the composition of $L$ falls into one of the following two cases:
\begin{enumerate}
\item if \typeline{$B$} then let \parentis{$B$}{$(E,F)$}.  Let $D_{E}$ and $D_{F}$ be the two neighbors of $B$ on the path from $E$ to $F$.  If $D_{E}$ and $D_{F}$ are in $D$ then $L=\{E, F\}$.  If only $D_{E}$ is in $D$, then $L=\{E\}$.  If only $D_{F}$ is in $D$, then $L=\{F\}$. Otherwise, $L=\emptyset$.
\item If \typeleaf{$B$} then let \parentis{$B$}{$E$}.   Let $D_{E}$ be the neighbor of $B$ on the path $B \ldots E$.  Let $L=\{E\}$ if $D_{E}$ is in $D$  and let $L=\emptyset$ otherwise.
\end{enumerate}
We call nodes appearing in either $C$ or $L$ {\em stolen} nodes.

\begin{lemma} \label{lemma:identify}
Identifying $B$, $L$ and $C$ takes at most $O(h)$ time.
\end{lemma}

\begin{proof}
By Theorem~\ref{theorem:test-membership} we can identify $B$ in $O(h)$ time.  Now we can use $H_{S}$ to identify the successors of $B$  which we can use to form $D$. Using a single {\em parent} operation (which is clearly bounded above by $h$), we can find either $LST(E)$ where \parentis{$B$}{$E$} or $BST(E,F)$ where \parentis{$B$}{$(E,F)$}.  We can  use the pointers offered by, in the first case, the dynamic edge query $(B,E)$ to identify $D_{E}$ and, in the second case, the dynamic edge queries $(B,E)$ and $(B,F)$ to identify $D_{E}$ and $D_{F}$.  With these nodes in hand, we can easily form $L$ by checking, in constant time, if, in the first case, $D_{E}$ is in $D$ and, in the second case, if $D_{E}$ and $D_{F}$ are in $D$.   Now we analyze the formation of the set $C$.  For each edge $(B,Y)$ in $LST(B)$, we use $\lambda_{1}(B,Y)=(B,Z)$ to identify the neighbor $Z$ of $B$ along the path $B \ldots Y$.  If $Z \in D$ then add $Z$ to $C$.  Since the height of $LST(B)$ is bounded above by $h$, we have the desired result.

\end{proof}

If $C$ and $L$ are both empty, then $A$ appears as a leaf in $T_{S'}$ and $\roundis{A}{1}$.  In this case, we only need to correct $\mathcal{T}$ upward since the addition of $A$ does not affect nodes contracted in earlier rounds, so we call \upc with $A$ and $B$.  However, if $C$ or $L$ is non-empty, then $A$ is an interior node in $T_{S'}$ and $A$ essentially acts as $B$ to the stolen nodes in $C$.  Thus, for every edge query $(B,C_{i})$ where $C_{i} \in C$, we remove $(B,C_{i})$ from $LST(B)$ and insert it into $LST(A)$.  In addition, we create a new edge $(B,A)$ and add it to $H_{S}$ which yields $H_{S'}$.

\begin{lemma} \label{lemma:steal}
The edge query removals from $LST(B)$ and their insertion into $LST(A)$ collectively take time proportional to the height of $LST(B)$.
\end{lemma}

\begin{proof}
We can traverse $LST(B)$, remove the edge queries involving nodes in  $C$, and insert them in $LST(A)$ in time proportional to the height of $LST(B)$ since LSTs are just linked lists.  For each stolen edge query $(B,C_{i})$ we need only replace $B$ with $A$ in the actual edge $\lambda_{1}(B,C_{i})=(B,X)$ so that it becomes $(A,X)$.  These pointer updates are bounded above by the height of $LST(B)$, so the lemma follows.

\end{proof}

\noindent Lemma~\ref{lemma:identify} and Lemma~\ref{lemma:steal} imply the following corollary.
\begin{corollary} \label{cor:pre}
Local correction takes $O(h)$ time.
\end{corollary}

\begin{table}[tbh]
\begin{center}
\caption{The \transition function which serves as a helper to {\sc insert}.  \label{tab:transition}}
\scriptsize
\renewcommand{\tabcolsep}{2mm}
\renewcommand{\arraystretch}{1.4}
\begin{tabular}{|c|l|l|}
\hline
\multirow{7}{*}{\begin{sideways} \transition($P,Q$) \end{sideways}}  & \multicolumn{1}{c|}{{\bf Updated Properties}} & \multicolumn{1}{c|}{{\bf Data Structure Updates}} \\ \cline{2-3} \cline{2-3}
& {\sc round}$(P) = \mu_2(P) + 1$ & \\ \cline{2-3}
& \multicolumn{2}{c|}{{\bf if} $\mu_1(P) = \mu_2(P)$ {\bf then} {\bf Up Correct} $P$ at insertion point $Q$} \\ \cline{2-3}
&  \multicolumn{2}{c|}{{\bf else} let $M = \rho_1(P)$} \\ \cline{2-3}
& & remove edge $(P,M)$ from $LST(P)$ \\
& \parent{$M$} $\gets$ $Q$ & $BST(Q,M) \gets$ \downc $BST(Q,P)$ and $BST(P,M)$\\
& & create edge $(Q,M)$ from $BST(Q,M)$ and insert it into $LST(Q)$ \\ \hline
\end{tabular}
\end{center}
\end{table}

\begin{table}
\begin{center}
\caption{The {\sc Insert} algorithm:  after locally correcting $A$ around its predecessor $B$, {\sc insert} uses \transition to either \downc or \upc.  Below, $k$ denotes {\sc round}$(B)$ before insertion and $BST(A,B)$ is an empty BST for the new edge $(A,B)$.
\label{tab:insert}}
\Tiny
\renewcommand{\tabcolsep}{2mm}
\renewcommand{\arraystretch}{1.4}
\begin{tabular}{| l| l|}
\hline

\multicolumn{2}{|l|}{{\sc insert}($A,B$)} \\ \hline \hline

\multicolumn{2}{|l|}{Apply Local Correction which yields candidate versions of $LST(A)$ and $LST(B)$ and a new edge $(A,B)$.} \\ \hline \hline

\multicolumn{1}{|l|}{{\bf Updated Properties}} & \multicolumn{1}{c|}{{\bf Data Structure Updates}} \\ \hline \hline

\multicolumn{2}{|l|}{{\bf Case 1}: \typeleaf{$B$} and \parentis{$B$}{\sc null}} \\ \hline \hline

\multicolumn{2}{|c|}{{\bf if} (1) $\mu_1(B) = \mu_2(B) = k-1$, or (2) $\mu_1(B) = k-1$ and $\mu_2(A) < k-2$, or} \\
\multicolumn{2}{|c|}{(3) $\mu_1(B) = \mu_2(B) = k-2$ and $\mu_2(A) < k-1$ {\bf then}} \\ \hline
{\bf if} $\mu_1(A) = \mu_2(A) = \mu_1(B) = \mu_2(B)$ {\bf then} & \multirow{2}{*}{\transition$(A,B)$} \\
{\sc round}$(B)$ $\gets$ {\sc round}$(B)+1$ & \\ \hline
\multicolumn{2}{|c|}{{\bf else}} \\ \hline
{\sc round}$(A)$ $\gets$ $\mu_1(A) + 1$ & \multirow{2}{*}{$A$ becomes new root of the \lltree} \\
\type{$A$} $\gets$ \leaf & \multirow{2}{*}{\transition$(B,A)$} \\
\parent{$A$} $\gets$ {\sc null} & \\ \hline \hline

\multicolumn{2}{|l|}{{\bf Case 2}: \typeleaf{$B$}, \parentis{$B$}{$E$}, and $L=\emptyset$} \\ \hline \hline

\multicolumn{2}{|c|}{{\bf if} $\mu_1(A) = \mu_2(A) = k-1$ {\bf then}} \\ \hline
{\sc round}$(A)$ $\gets$ $\mu_1(A) + 1$ & remove edge $(B,E)$ from $LST(E)$ \\
\type{$A$} $\gets$ \leaf & $BST(A,E) \gets $ \downc $BST(A,B)$ and $BST(B,E)$\\
\parent{$A$} $\gets$ $E$ & insert edge $(A,E)$ into $LST(E)$ \\ \hline
\multicolumn{2}{|c|}{{\bf else} \transition$(A,B)$} \\ \hline \hline

\multicolumn{2}{|l|}{{\bf Case 3}: \typeleaf{$B$}, \parentis{$B$}{$E$}, and $L=\{E\}$} \\ \hline \hline

{\sc round}$(A)$ $\gets$ $\mu_1(A) + 1$ & replace $B$ with $A$ in $BST(B,E)$ which becomes $BST(A,E)$ \\ \hline
\multicolumn{2}{|c|}{{\bf if} $\mu_2(B) < k-1$ {\bf then}} \\ \hline
\multirow{2}{*}{\type{$A$} $\gets$ \leaf} & remove edge $(B,E)$ from $LST(E)$\\
\multirow{2}{*}{\parent{$A$} $\gets$ $E$} & create edge $(A,E)$ and insert into $LST(E)$  \\
& \transition$(B,A)$ \\ \hline
\multicolumn{2}{|c|}{{\bf else} $BST(B,E) \gets $ \downc $BST(B,A)$ and $BST(A,E)$} \\
\hline \hline

\multicolumn{2}{|l|}{{\bf Cases 4-5}: \typeline{$B$} and \parentis{$B$}{$(E,F)$}} \\ \hline \hline

\multicolumn{2}{|c|}{Let $(N,B)$, $(B,M)$ be edges adjacent to $B$ in $BST(E,F)$ } \\
\multicolumn{2}{|c|}{in the directions of $E$ and $F$ respectively.  W.l.o.g. $A$ falls between $(B,M)$} \\ \hline \hline

\multicolumn{2}{|l|}{{\bf Case 4}: $L=\{E\}$ or $L=\{F\}$} \\ \hline \hline

{\sc round}$(A)$ $\gets$ $\mu_1(A) + 1$ & \multirow{2}{*}{remove edge $(B,M)$ from $BST(E,F)$} \\
{\sc round}$(B)$ $\gets$ $\mu_1(B) + 1$ & \multirow{2}{*}{replace $B$ with $A$ in $BST(B,M)$ which becomes $BST(A,M)$} \\
\type{$A$} $\gets$ {\sc line} & \\ \hline
\multicolumn{2}{|c|}{{\bf if} $\mu_1(B) > \mu_1(A)$ {\bf then}} \\ \hline
&  $BST(B,M) \gets $ \downc $BST(B,A)$ and $BST(A,M)$ \\
&  insert edge $(B,M)$ (with $BST(B,M)$) back in $BST(E,F)$  \\ \hline
\multicolumn{2}{|c|}{{\bf else if} $\mu_1(B) = \mu_1(A)$ {\bf then}} \\ \hline
\parent{$A$} $\gets$ $(E,F)$ & insert edges $(B,A)$ and $(A,M)$ into $BST(E,F)$ \\ \hline
\multicolumn{2}{|c|}{{\bf else} $\mu_1(B) < \mu_1(A)$} \\ \hline
\multirow{4}{*}{\parent{$A$} $\gets$ $(E,F)$}  & insert edge $(A,M)$ into $BST(E,F)$ \\
& remove edge $(N,B)$ from $BST(E,F)$ \\
& $BST(N,A) \gets$ \downc  $BST(N,B)$ and $BST(B,A)$ \\
& create edge $(N,A)$ (with $BST(N,A)$) and insert into $BST(E,F)$ \\ \hline \hline

\multicolumn{2}{|l|}{{\bf Case 5}: $L=\{E,F\}$ or $L = \emptyset$} \\ \hline \hline

\multicolumn{2}{|c|}{{\bf if} $L = \{ E, F \}$ {\bf then}} \\ \hline
\type{$A$} $\gets$ {\sc line} & \multirow{2}{*}{replace $B$ with $A$ in $BST(E,F)$} \\
{\sc round}$(A)$ $\gets$ {\sc round}$(B)$ & \multirow{2}{*}{\transition$(B,A)$} \\
\parent{$A$} $\gets$ $(E,F)$ & \\  \hline
\multicolumn{2}{|c|}{{\bf else} $L =  \emptyset $: \transition$(A,B)$} \\

\hline
\end{tabular}
\end{center}
\end{table}

Local correction leaves us with candidate versions of $LST(A)$ and $LST(B)$ as well as a set of nodes $L$.  The edges in $LST(A)$ and $LST(B)$ remain in their respective LSTs with one small exception:  Stealing edge queries from $B$ and inserting them into $A$ may cause one of $A$ or $B$ to no longer adhere to Lemma~\ref{lemma:value-lemma} and we may need to continue correcting the \lltree upward or downward.

The {\sc insert} procedure uses a helper function, \transition, to identify these situations and transition into either \downc or \upc:  given two nodes $P$ and $Q$ where $\round{P} \leq \round{Q}$, it determines if $P$ was line contracted between $\rho_{1}(P)$ and $Q$ at some earlier round.  If this is the case, then the contraction process has been repaired except for node $P$ which may be out of place on the line from $\rho_{1}(P)$ to $Q$ so it calls \downc on $(\rho_{1}(P),P)$ and $(P, Q)$ to finish the contraction process.  Otherwise, we have patched the contraction process for all rounds up to $\round{P}$ so we \upc $P$ and $Q$ to complete the repair.  A formal description of \transition appears in Table~\ref{tab:transition} and a formal description of {\sc insert} appears in Table~\ref{tab:insert}.

\begin{table}[tb]
\begin{center}
\footnotesize
\renewcommand{\tabcolsep}{2mm}
\renewcommand{\arraystretch}{1.4}
\begin{tabular}{| c| c| c| c|}
\hline
{\bf Case} & \type{$B$} & $\parent{B}$ & {\bf L} \\
\hline
\hline
1 & 	\multirow{3}{*}{\sc leaf} &	\multicolumn{1}{c|}{\sc null} &	\multirow{2}{*}{$L=\emptyset$} \\ \cline{1-1} \cline{3-3}
2 &	&							\multirow{2}{*}{$E$}  &	 \\ \cline{1-1} \cline{4-4}
3 &	&							&									$L=\{E\}$\\ \hline
4 &	\multirow{2}{*}{\sc line} &	\multirow{2}{*}{$(E,F)$} &		$L=\{E\}$ or $L=\{F\}$		\\ \cline{1-1} \cline{4-4}
5 &	&							&									$L = \{E, F\}$ or $L = \emptyset$ \\
\hline
\end{tabular}
\vspace{2mm}
\caption{An exhaustive list of cases for {\sc insert}. \label{tbl:center-correction}}
\end{center}
\end{table}

\begin{theorem} \label{thm:insert}
{\sc Insert} takes $O(h)$ time.
\end{theorem}

\begin{proof}
The heart of our proof is showing that {\sc insert} arrives at a scenario where \transition can be called.  We show this by exhaustively examining how {\sc insert} deals with all possible {\sc round}, {\sc type}, and {\sc parent} values of $B$ as well as the contents of $L$ after executing local correction.  To help this verification, we summarize the list of cases in Table~\ref{tbl:center-correction}.  In all the cases below $\roundis{B}{k}$.

\begin{description}

\item [Case 1: \parentis{$B$}{{\sc null}}]  From Lemma~\ref{lemma:value-lemma} we know that before insertion, either $\mu_1(B) = \mu_2(B) = \mu_3(B) = k-1$ or $\mu_1(B) = k-1 > k-2 = \mu_2(B) = \mu_3(B)$.  Consider the first case where, before insertion, $\mu_1(B) = \mu_2(B) = \mu_3(B) = k-1$.  We have two possibilities after inserting $A$.
\begin{itemize}
\item [(a)] Suppose that after insertion, $\mu_1(B) = \mu_2(B) = k-1$. From Lemma~\ref{lemma:value-lemma}, this implies that $\max \{ \mu_3(B), \mu_1(A) \} = k-1$ and $B$ has degree at least 3 at the beginning of iteration $k-1$. Similarly, if $\mu_1(A) = \mu_2(A) = k-1$, then $A$ also has degree 3 at the beginning of iteration $k-1$. After iteration $k-1$, either $B$ survives alone, or $B$ and $A$ each survive with degree 1. We keep $B$ as the root of the \lltree; if $A$ survives together with $B$, we increment \round{$B$} and this completes the correction.
\item [(b)] Suppose that after insertion, $\mu_2(B) \leq \mu_1(B) < k-1$. This implies that $\mu_1(A) = \mu_2(A) = k-1$. Thus, analogous to above, $A$ survives alone after iteration $k-1$. $A$ becomes the new root of the \lltree.  We can now apply the \transition function with $P=B$ and $Q=A$ to correct $B$ on the path between $A$ and $\rho_{1}(B)$.  This completes the correction.\\
\end{itemize}

Now, suppose that before insertion $\mu_1(B) = k-1 > k-2 = \mu_2(B) = \mu_3(B)$. \\

\begin{itemize}
\item [(a)] Suppose that after insertion $\mu_1(B) = k-1$ where $M=\rho_{1}(B)$.  This implies $k-2 \geq \mu_1(A) \geq \mu_2(A)$. Thus, after $k-2$ iterations, $T_{S'}$ is either a line with endpoints $M$ and $B$ ($A$ was contracted earlier), or a line with endpoints $M$ and $A$ ($B$ may be on the chain of degree 2 nodes connecting $M$ and $A$ or may have been contracted earlier). If $\mu_1(A) = \mu_2(A) = k-2$, then $A$ survives until iteration $k-1$. Even if $B$ survives as well ($\mu_2(B) = k-2$), it is line contracted into $BST(A,M)$. Thus, w.l.o.g. $A$ becomes the new root node. We then apply \transition with $P=B$ and $Q=A$ to correct the path between $A$ and $\rho_{1}(B)$.  This completes the correction. If $\mu_2(A) < k-2$, then $A$ does not survive until iteration $k-1$. This means that $\mu_2(B) = k-2$ and $B$ survives together with $M$.  Without any loss of generality, we keep $B$ as root. We then apply \transition with $P=A$ and $Q=B$ to correct $A$ on the path between $B$ and $\rho_{1}(A)$.  This completes the correction.
\item [(b)] Suppose that after insertion $\mu_1(A) = k-1$. The situation is symmetric to case (a) above: if $\mu_1(B) = \mu_2(B) = k-2$, then $B$ survives until iteration $k-1$. Even if $A$ survives as well ($\mu_2(A)  = k-2$), it is line contracted into $BST(B,M)$. Thus, w.l.o.g. $B$ stays the root node. We then apply \transition with $P=A$ and $Q=B$. If $\mu_2(B)  < k-2$, then $B$ does not survive until iteration $k-1$. This means that $\mu_2(A)  = k-2$ and $A$ survives together with $M$.  Here we make $A$ the new root. We then apply \transition with $P=B$ and $Q=A$ to repair the path from $\rho_{1}(B)$ to $A$.  This completes the correction.\\
\end{itemize}

To review,  if either (1) $\mu_1(B) = \mu_2(B) = k-1$, or (2) $\mu_1(B) = k-1$ and $\mu_2(A) < k-2$, or (3) $\mu_1(B) = \mu_2(B) = k-2$ and $\mu_2(A) < k-1$, then $B$ remains the root of the \lltree after insertion. We then apply \transition with $P=A$ relative to $Q=B$. Otherwise, $A$ becomes the new root and we \transition with $P=B$ and $Q=A$.\\

For cases 2--3, \typeleaf{$B$} and \parentis{$B$}{$E$}. Before insertion, after iteration $k-1$, $B$ had degree 1 and was connected to $E$ through a (possibly empty) chain of degree 2 nodes. The chain was line contracted into $BST(B,E)$ and $B$ was leaf contracted into $E$.\\

\item[Case 2:  $L= \emptyset $]

After insertion, if $\mu_1(A) = \mu_2(A) = k-1$, then after iteration $k-1$, $A$ has degree 1 and is connected to $E$ through a (possibly empty) chain of degree 2 nodes that may contain $B$. Thus, the edge query $(A,E)$ replaces edge $(B,E)$ in $LST(E)$. Since {\sc round}$(B)  = \mu_1(B) + 1 \leq k$, $B$ is line contracted between $A$ and $E$ so we \downc with respect to $(A,B)$ and $(B,E)$ to determine $BST(A,E)$.   Otherwise, if $\mu_2(A)  < k-1$, then $A$ is contracted before iteration $k$ and $\mathcal{T}$ is identical to $\mathcal{T'}$ beginning with {\sc round} $k$.  We keep edge $(B,E)$ in $LST(E)$ and apply the \transition algorithm with $P=A$ and $Q=B$ which completes the correction. 

\item[Case 3: $L=\{ E \}$]

There are two subcases.  First, if $\mu_1(B) = \mu_2(B) = k-1$ after insertion, then after iteration $k-1$, $B$ has degree 1 and is connected to $E$ through a (possibly empty) chain of degree 2 nodes that may contain $A$.  We know $(B,E)$ remains an edge query in $LST(E)$, but it now needs to accommodate the addition of $A$ since $\round{A}= \mu_1(A) +1 \leq k$  Thus, we remove the edge query $(B,E)$ from $LST(E)$ and, as a preliminary step, replace $B$ with $A$ in $BST(B,E)$ to produce $BST(A,E)$.  Then we \downc with respect to $(A,B)$ and $(A,E)$ to determine the new $BST(B,E)$.  Finally, we insert $(B,E)$ back into $LST(E)$.

Second, if after insertion $\mu_2(B)  < k-1$, then $B$ is contracted before iteration $k$. After iteration $k-1$, $A$ has degree 1 and is connected to $E$ through a (possibly empty) chain of degree 2 nodes.  Thus, we remove $(B,E)$ from $LST(E)$, replace $B$ with $A$ in $BST(B,E)$ to yield $BST(A,E)$, and insert the new edge query $(A,E)$.  Now we're in a position to apply the \transition function with $P=B$ and $Q=A$ after which we've repaired the contraction process.\\

For cases 4--5, \typeline{$B$} and \parentis{$B$}{$(E,F)$}. Before insertion, after iteration $k-1$, $B$ had degree 2 and was part of a chain of degree 2 nodes connecting $E$ and $F$. $(E,F)$ became the {\sc parent} of $B$ and the chain together with $B$ was line contracted into $BST(E,F)$. Let $(B,N)$ and $(B,M)$ be the edges representing $B$ in $BST(E,F)$, where w.l.o.g. $A$ falls between $(B,M)$.\\

\item[Case 4: $L=\{E\}$ or $L=\{F\}$]

After insertion, {\sc round}$(B) = \mu_1(B) + 1$ and \roundis{$A$}{$\mu_1(A) + 1$}, since each of the two nodes is line contracted right after it has accumulated all of its leaves. We examine what happens to the \lltree after iteration $k-1$. If $\mu_1(B)  > \mu_1(A)$, then the current \lltree looks identical to the one prior to insertion since $A$ is line contracted at a prior iteration somewhere between $B$ and $M$. If $\mu_1(B) = \mu_1(A)$, then both $A$ and $B$ have degree 2 and are part of a chain of degree 2 nodes connecting $E$ and $F$. $A$ and $B$ are line contracted together into $BST(E,F)$. If $\mu_1(B) < \mu_1(A)$, then $A$ has replaced $B$ in the chain of degree 2 nodes connecting $E$ and $F$.  $A$ is line contracted into $BST(E,F)$.

In the data structure, we always replace $B$ with $A$ in $BST(B,M)$ to create $BST(A,M)$. If $\mu_1(B) > \mu_1(A)$, we \downc $BST(B,A)$ and $BST(A,M)$ which yields a correct version of $BST(B,M)$ which we insert back into $BST(E,F)$. If $\mu_1(B) = \mu_1(A)$, we insert edges $(A,B)$ and $(A,M)$ into $BST(E,F)$. If $\mu_1(B) < \mu_1(A)$, we repair the line between $A$ and $N$ so that $B$ settles in its proper place.  We remove $BST(N,B)$ from $BST(E,F)$ and \downc $BST(N,B)$ and $BST(B,A)$ to produce $BST(A,N)$ which we insert back into $BST(E,F)$.  This concludes the contraction process.

\item[Case 5: $L = \{E, F\}$ or $L = \emptyset$]

If $L = \{E, F\}$ then after iteration $k-1$, $A$ is connected to $E$ and $F$ by chains of degree 2 nodes (identical to the ones for $B$ pre-insertion). If $\mu_1(B) = \mu_2(B) = k-1$, then $B$ also survives as a neighbor of $A$.  Thus, $B$ is a leaf at some point in the contraction process so we \upc $B$ at insertion point $A$ (this happens via the call to \transition).   If $B$ does not survive, then $A$ is line contracted into $BST(E,F)$ analogously to how $B$ was pre-insertion. In the data structure, we replace $B$ with $A$ in $BST(M,B)$ and $BST(N,B)$ and call \transition to potentially repair the path from $\rho_{1}(B)$ to $A$.  This completes the correction.\\

\end{description}

If $L = \emptyset$, then after iteration $k-1$, $B$ is connected to $E$ and $F$ by chains of degree 2 nodes (identical to the ones pre-insertion). If $\mu_1(A) = \mu_2(A) = k-1$, then $A$ also survives as a neighbor of $B$. Thus, $A$ is a leaf at some point in the contraction process so we \upc $A$ at insertion point $B$.  If $A$ does not survive, then $B$ is line contracted into $BST(E,F)$ analogous to pre-insertion. In the data structure, we need only worry about correcting the path from $\rho_{1}(A)$ to $B$ which is done via \transition.  This completes the correction.

What's left to show is that {\sc insert} runs in time proportional to the height of the \lltree.  From Corollary~\ref{cor:pre} local correction operates in $O(h)$ time.  Furthermore, each case of {\sc insert} performs at most $O(1)$ BST, LST, \downc, and \upc operations -- each of which takes at most $O(h)$ time (See Lemma~\ref{lemma:down-correction} for \downc and Lemma~\ref{lemma:up-correction} for \upc).
\end{proof}

\subsubsection{Down Correction} \label{sec:down-correction}

\begin{table}[tb]
\begin{center}
\caption{\downc :  Given $BST(E,B)$ which was created in round $m$ ($m=0$ if $BST(E,B)$ is empty) and $BST(B,F)$ which was created in round $n$ ($n=0$ if $BST(B,F)$ is empty), return $BST(E,F)$ where $B$ appears in the correct position in $BST(E,F)$. \label{tab:down}}
\renewcommand{\tabcolsep}{2mm}
\renewcommand{\arraystretch}{1.4}
\tiny
\begin{tabular}{|l|l|}
\hline

\multicolumn{2}{|l|}{\downc($BST(E,B), BST(B,F)$)} \\ \hline \hline

\multicolumn{1}{|c|}{{\bf Updated Properties}} & \multicolumn{1}{c|}{{\bf Data Structure Updates}} \\ \hline \hline

\multicolumn{2}{|l|}{{\bf Cases 1--5}} \\ \hline \hline

\typeline{$B$} & \\ \hline \hline

\multicolumn{2}{|l|}{{\bf Case 1}: $m$ $<$ {\sc round}$(B)$ $=$ $n$} \\ \hline \hline

\multirow{2}{*}{\parentis{$B$}{$(E,F)$}} & insert  $(E,B)$ into $BST(B,F)$ which becomes $BST(E,F)$  \\
& return $BST(E,F)$ \\ \hline \hline

\multicolumn{2}{|l|}{{\bf Case 2}: $m,n$ $<$ {\sc round}$(B)$} \\ \hline \hline

\multirow{2}{*}{\parentis{$B$}{$(E,F)$}} & insert edges $(E,B)$ and $(B,F)$ into a new (empty) $BST(E,F)$  \\
& return $BST(E,F)$ \\ \hline \hline

\multicolumn{2}{|l|}{{\bf Case 3}: $m$ $=$ {\sc round}$(B)$ $=$ $n$} \\ \hline \hline

\multirow{2}{*}{\parentis{$B$}{$(E,F)$}} & merge $BST(E,B)$ and $BST(B,F)$ into $BST(E,F)$\\
&  return $BST(E,F)$ \\ \hline \hline

\multicolumn{2}{|l|}{{\bf Case 4}: $m$ $\leq$ {\sc round}$(B)$ $<$ $n$} \\ \hline \hline

\multicolumn{2}{|c|}{Let $(B,N)$ be the edge queries representing $B$ in $BST(B,F)$, where $N \neq F$} \\ \hline
\multirow{5}{*}{ \parent{$N$} $\gets$ $(E,F)$}& remove edge $(B,N)$ from $BST(B,F)$ which becomes becomes $BST(N,F)$ \\
& $BST(E,N) \gets$ \downc  $BST(E,B)$ and $BST(B,N)$ \\
& create edge $(E,N)$ with $BST(E,N)$ \\
& insert edge $(E,N)$ into $BST(N,F)$ which becomes $BST(E,F)$  \\
& return $BST(E,F)$ \\ \hline \hline

\multicolumn{2}{|l|}{{\bf Case 5}: $m,n$ $>$ {\sc round}$(B)$} \\ \hline \hline

\multicolumn{2}{|c|}{Let $(M,B)$ and $(B,N)$ be the edge queries} \\
\multicolumn{2}{|c|}{representing $B$ in $BST(E,B)$, $BST(B,F)$, respectively.} \\ \hline
& remove edge $(M,B)$ from $BST(E,B)$ which becomes $BST(E,M)$ \\
\parent{$M$} $\gets$ $(E,F)$ & remove edge $(B,N)$ from $BST(B,F)$ which becomes $BST(N,F)$ \\
& $BST(M,N) \gets$ \downc $BST(M,B)$ and $BST(B,N)$ \\ \hline

\multicolumn{2}{|c|}{{\bf if} $m > n$ (w.l.o.g) {\bf then}} \\ \hline

\multirow{4}{*}{\parent{$N$} $\gets$ $(M,F)$} & create edge $(M,N)$ with $BST(M,N)$ \\
& insert edge $(M,N)$ into $BST(N,F)$ which becomes $BST(M,F)$ \\
& insert edge $(M,F)$ into $BST(E,M)$ which becomes $BST(E,F)$ \\
& return $BST(E,F)$ \\ \hline
\multicolumn{2}{|c|}{{\bf else} $m = n$} \\ \hline
\multirow{2}{*}{\parent{$N$} $\gets$ $(E,F)$} & merge $BST(E,M)$, $BST(M,N)$, and $BST(N,F)$ into $BST(E,F)$ \\
& return $BST(E,F)$ \\

\hline
\end{tabular}
\end{center}
\end{table}

Down correction repairs the contraction process along a path in the partial order tree.  More specifically, down correction takes two binary search trees $BST(E,B)$ and $BST(B,F)$ where $\round{B}$ does not respect the contraction process and returns a third search tree $BST(E,F)$ where $B$ has been {\em floated} down to the BST created in same round as $B$.  In all cases, we know the edge $(E,F)$ appears at some point in the contraction process.  We assume that if $X$ is a node on the path from $E$ to $F$ then both $\round{X}$ and $LST(X)$ are well-formed and correct.  This includes includes $B$---it is simply out of place structurally with respect to the contraction process.  Moreover, and without loss of generality, we assume that $\round{B} \leq \round{E} \leq \round{F}$ and that if $\round{E}=\round{B}$ then $\typeleaf{E}$.  Down correction is used as a subroutine by both {\em insertion} and {\em deletion}.  A formal description of the algorithm is given in Table~\ref{tab:down}.

\begin{lemma} \label{lemma:down-correction}
Let $BST(E,B)$ and $BST(B,F)$ be binary search trees along a path from $E$ to $B$ to $F$ where $LST(X)$ and $\round{X}$ are correct for every node $X$ on the path from $E$ to $F$.  Furthermore let $\round{B} \leq \round{E} \leq \round{F}$ and $\round{E}=\round{B}$ only when $\typeleaf{E}$.  If $B$ does not respect the contraction process with respect to the path from $E$ to $F$ then \downc returns $BST(E,F)$ in $O(h)$ time where $B$ occupies the correct position.
\end{lemma}

\begin{proof}
We begin by showing that \downc successfully repairs the contraction process.  Suppose $BST(E,B)$ was created in round $m$ where $m=0$ if $BST(E,B)$ is empty and suppose $BST(B,F)$ was created in round $n$ where $n=0$ if $BST(B,F)$ is empty.  The proof is by structural induction.  There are 5 cases.  We begin with the base cases where both $m$ and $n$ do not exceed $\round{B}$.

\begin{description}
\item [Case 1: $m < \round{B} = n$]  $B$ is line contracted after nodes in $BST(E,B)$, but together with all the nodes in $BST(B,F)$. We add edge query $(E,B)$ to $BST(B,F)$ and the result becomes the new $BST(E,F)$.
\item [Case 2: $m,n < \round{B}$]  $B$ is the only node line contracted at {\sc round}$(B)$ between $E$ and $F$. This happens after all other nodes in $BST(E,B)$ and $BST(B,F)$ were already line contracted. We create a new $BST(E,F)$ and populate it with edge queries $(E,B)$ and $(B,F)$.
\item [Case 3: $m = \round{B} = n$] $B$ is line contracted together with all nodes in $BST(E,B)$ and all nodes in $BST(B,F)$. We merge $BST(B,E)$ and $BST(B,F)$ and the result becomes the new $BST(E,F)$.
\item [Case 4: $m \leq \round{B} < n$] $B$ is line contracted either after or at the same time as nodes in $BST(E,B)$, but before any of the nodes in  $BST(B,F)$. Let $(B,N)$ be the edge query representing $B$ in $BST(B,F)$, where $N$ cannot be $F$ since $BST(B,F)$ is not empty. Then, we know that $B$ is line contracted somewhere between $E$ and $N$ at a previous iteration.  Remove $(B,N)$ from $BST(B,F)$ (which becomes $BST(N,F)$).  Recursively down correcting on $BST(E,B)$ and $BST(B,N)$ yields a correct $BST(E,N)$ which we insert into $BST(N,F)$ to produce $BST(E,F)$.
\item [Case 5: $\round{B} < n,m$]  $B$ is line contracted before any of the nodes in either $BST(B,E)$ or $BST(B,F)$. Let $(B,M)$ and  $(B,N)$ be the edges representing $B$ in $BST(E,B)$ and $BST(B,F)$, respectively where $M$ cannot be $E$ and $N$ cannot be $F$, since the two $BST$s are not empty. We remove $(M,B)$ and $(B,N)$ from $BST(E,B)$ and  $BST(B,F)$, obtaining $BST(E,M), BST(N,F)$. Then, we know that $B$ is contracted somewhere between $M$ and $N$ at a previous iteration: inductively down correcting $BST(M,B)$ and $BST(B,N)$ yields $BST(M,B)$.  If w.l.o.g. $m > n$, we proceed similarly to Case 4: we insert $(M,N)$ into $BST(N,F)$ (creating $BST(M,F)$) and then insert $(M,F)$ to $BST(E,M)$. The result is the new $BST(E,F)$. Otherwise, if $m=n$, we proceed similarly to Case 3: we merge $BST(E,M)$, $BST(M,N)$, and $BST(N,F)$ to create the new $BST(E,F)$.
\end{description}

What's left to show is that \downc operates in at most $O(h)$ time.  Each recursive call (in cases 4 and 5) operates on, minimally, $BST(B,N)$ where $(B,N)$ is the edge bordering some path from $B$ to $F$ that was line contracted into $BST(B,F)$ at iteration $n$.  From Lemma~\ref{lemma:path}, $BST(B,N)$ was created at a previous iteration, so the algorithm halts after visiting at most $l$ BSTs.  Because $(B,N)$ is always a bordering edge, the sum of the heights of the $l$ BSTs is bounded above by $O(h)$.  Since we perform at most $O(1)$ BST operations on each of the $l$ BSTs, and these operations run in time proportional to the height, we have the desired bound.  
\end{proof}

\subsubsection{Up Correction} \label{sec:up-correction}

\begin{table}
\begin{center}
\caption{\upc : insert node $A$ at insertion point $B$ where $\round{A} \leq \round{B}$ and $LST(A)$ is fully-formed and correct with respect to $\mathcal{T'}$.  \label{tab:up}}
\renewcommand{\tabcolsep}{2mm}
\renewcommand{\arraystretch}{1.4}
\scriptsize
\begin{tabular}{|l|l|}
\hline

\multicolumn{2}{|l|}{\upc($A,B$)} \\ \hline \hline

\multicolumn{1}{|c|}{{\bf Updated Properties}} & \multicolumn{1}{c|}{{\bf Data Structure Updates}} \\ \hline \hline

\multicolumn{2}{|l|}{{\bf Cases 1--6}} \\ \hline \hline

\type{$A$} $\gets$ \leaf & \\ \hline \hline

\multicolumn{2}{|l|}{{\bf Case 1}: {\sc round}$(A)$ $<$ {\sc round}$(B)$, \parentis{$B$}{{\sc null}}, and $\mu_2(B) < \mu_1(B) = $ {\sc round}$(A) = k-1$} \\ \hline \hline

\multicolumn{2}{|c|}{Let $\rho_1(B) = M$} \\ \hline
\parent{$M$} $\gets$ {\sc null} & remove edge $(B,M)$ from $LST(B)$\\
\parent{$A$} $\gets$ $M$ & $BST(A,M) \gets$ \downc $BST(A,B)$ and $BST(B,M)$ \\
{\sc round}$(B)$ $\gets$ {\sc round}$(B)-1$ & create edge $(A,M)$ from $BST(A,M)$ and insert it into $LST(M)$ \\
{\sc round}$(M)$ $\gets$ {\sc round}$(M)+1$ & $M$ becomes new root of the \lltree \\ \hline \hline

\multicolumn{2}{|l|}{{\bf Case 2}: {\sc round}$(A)$ $<$ {\sc round}$(B)$ (and Case 1 does not apply)} \\ \hline \hline

\parent{$A$} $\gets$ $B$ & insert edge $(A,B)$ into $LST(B)$ \\ \hline \hline

\multicolumn{2}{|l|}{{\bf Case 3}: {\sc round}$(A)$ $=$ {\sc round}$(B)$, \typeleaf{$B$}, and \parentis{$B$}{$E$}} \\ \hline \hline

\multirow{3}{*}{\parent{$A$} $\gets$ $E$} & remove edge $(B,E)$ from $LST(E)$ \\
& $BST(A,E) \gets$ \downc $BST(A,B)$ and $BST(B,E)$ \\
& create edge $(A,E)$ from $BST(A,E)$ and insert it into $LST(E)$ \\ \hline \hline

\multicolumn{2}{|l|}{{\bf Cases 4--6}: {\sc round}$(A)$ $=$ {\sc round}$(B)$, \typeline{$B$}, and \parentis{$B$}{$(E,F)$}} \\ \hline \hline

\parent{$A$} $\gets$ $B$ & insert edge $(A,B)$ into $LST(B)$ \\
{\sc round}$(B)$ $\gets$ {\sc round}$(B)+1$ & split $BST(E,F)$ into $BST(E,B)$ and $BST(B,F)$ \\ \hline \hline

\multicolumn{2}{|l|}{{\bf Case 4}: {\sc round}$(B)+1$ $<$ {\sc round}$(E)$, {\sc round}$(F)$ or} \\
\multicolumn{2}{|c|}{{\sc round}$(B)+1$ $=$ {\sc round}$(E)$ $<$ {\sc round}$(F)$ and \typeleaf{$E$}} \\ \hline \hline

& create new $BST(E,F)$ \\
& insert  edges $(E,B)$ and $(B,F)$ into $BST(E,F)$ \\ \hline \hline

\multicolumn{2}{|l|}{{\bf Case 5}: {\sc round}$(B)+1$ $=$ {\sc round}$(E)$ and \typeline{$E$}} \\ \hline \hline

\multicolumn{2}{|c|}{Let \parentis{$E$}{$(G,H)$}, where $H$ may be $F$} \\ \hline
\multirow{2}{*}{\parent{$B$} $\gets$ $(G,H)$} & remove edge $(E,F)$ from $BST(G,H)$ \\
& insert edges $(E,B)$ and $(B,F)$ into $BST(G,H)$ \\ \hline \hline

\multicolumn{2}{|l|}{{\bf Case 6}: {\sc round}$(E)$ $<$ {\sc round}$(B)+1$ $\leq$ {\sc round}$(F)$} \\ \hline \hline

\multicolumn{2}{|c|}{{\sc type}$(E)$ must be {\sc leaf} and \parentis{$E$}{$F$}} \\ \hline
{\sc type}$(B)$ $\gets$ \leaf & remove edge $(E,F)$ from $LST(F)$ \\
\parent{$E$} $\gets$ $B$ & insert  edge $(B,E)$ into $LST(B)$ \\ \hline
\multicolumn{2}{|c|}{{\bf if} \parentis{$F$}{{\sc null}} and {\sc round}$(F)$ $=$ {\sc round}$(B)+1$ {\bf then}} \\ \hline
\parent{$B$} $\gets$ {\sc null} & \multirow{2}{*}{insert edge $(B,F)$ into $LST(B)$} \\
\parent{$F$} $\gets$ $B$ & \multirow{2}{*}{$B$ becomes new root of the \lltree} \\
\round{$F$} $\gets$ \round{$F$} $-1$ & \\ \hline
\multicolumn{2}{|c|}{{\bf else} \upc $B$ at insertion point $F$} \\

\hline
\end{tabular}
\end{center}
\end{table}

Suppose we are inserting $A$ into $\mathcal{T}$ at insertion point $B$ where $\round{A} \leq \round{B}$ and $\round{A}$ and $LST(A)$ are correct with respect to $\mathcal{T'}$.  Suppose further that we know $\typeleaf{A}$ and that $(A,B)$ appears as an edge during some iteration of the contraction process (initially, $A$ and $B$ are neighbors in $H_{S}$).  The addition of $A$ may change the contraction process with respect to $B$ (and these changes may propagate to later iterations of the contraction process).  Up correction repairs the contraction process in this situation.  A formal description of the \upc algorithm is given in Table~\ref{tab:up}.

\begin{lemma} \label{lemma:up-correction}
\upc repairs the contraction process in $O(h)$ time so that $\mathcal{T}=\mathcal{T'}$.
\end{lemma}

\begin{proof}
We begin by showing that \upc correctly repairs the contraction process.  The proof is by structural induction.  Let $A$ be the node we are inserting at insertion point $B$ where $\round{A} \leq \round{B}=k$.  We distinguish six cases which are detailed below. 

\begin{description}

\item [Case 1]: {\bf $\round{A} < \round{B}$, \parentis{$B$}{{\sc null}}, and $\mu_2(B) < \mu_1(B) = \round{A} = k-1$}  Since \parentis{$B$}{{\sc null}}, $B$ is the root of $\mathcal{T}$.  By Lemma~\ref{lemma:value-lemma}, $B$ have either (i) $\mu_1(B) = k-1 > k-2 = \mu_2(B) = \mu_3(B)$ or (ii) $\mu_1(B) = \mu_2(B) = \mu_3(B) = k-1$.  This case handles the scenario that inserting $A$ into $LST(B)$ (which is the default action when \round{$A$} $<$ \round{$B$}) leads to $\mu_1(B) = \mu_2(B) = k-1 > k-2 = \mu_3(B)$ and $B$ ceases to be the root.

Before insertion, at the beginning of iteration $k-1$, the partial order tree consisted of nodes $B$ and $\rho_1(B) = M$ connected by a (possibly empty) chain of degree-2 nodes. We line contracted the chain into $BST(B,M)$ and installed $B$ as the root node. During the insertion procedure so far, $M$ was neither stolen from $B$, nor removed from $LST(B)$; $M$ survives after iteration $k-2$.  Since {\sc round}$(A) = k-1$, $A$ also survives. Thus, at the beginning of iteration $k-1$ the poset consists of nodes $M$ and $A$ connected by a chain of degree-2 nodes containing $B$. Now, we should line contract this chain into $BST(A,M)$.  We leaf-contract $A$ into $M$ and arbitrarily assign $M$ to be the root.

To repair the data structure, we replace $B$ with $M$ as the root node (we decrease $\round{B}$ to $k-1$ and we increase \round{$M$} to $k$). We also remove edge $(B,M)$ from $LST(B)$ and we insert a newly created edge $(M,A)$ into $LST(M)$. To construct the new $BST(A,M)$ we appeal to down correction:  the contraction process is correct along the path from $A$ to $M$---only $\round{B}$ has changed.

\item [Case 2: \round{$A$} $<$ \round{$B$} and Case 1 does not apply]  Node $A$ is leaf contracted into $LST(B)$ before iteration $k$ and does not change $\mathcal{T}$ beyond this iteration. In the data structure, we insert $(A,B)$ into $LST(B)$.

\item [Case 3: $\round{A}=\round{B}$, \typeleaf{$B$}, and $\parentis{B}{E}$]  Before insertion, at the beginning of iteration $k$, $B$ had degree 1 and was connected to $E$ through a (possibly empty) chain of degree-2 nodes. The chain was line contracted into $BST(B,E)$ and $B$ was leaf contracted into $LST(E)$. After insertion, $B$ has degree 2 with neighbors $A$ on one side and the chain ending with $E$ on the other. The chain, together with $B$, is line contracted into $BST(A,E)$ and $A$ should be subsequently leaf contracted into $LST(E)$ instead of $B$.  To repair the data structure, we replace edge $(B,E)$ with edge $(A,E)$ in $LST(E)$ where $BST(A,E)$ comes from down correcting $BST(A,B)$ and $BST(B,E)$.
\end{description}

\noindent For all subsequent cases (4--6), \round{$A$} $=$ \round{$B$} and \typeline{$B$}. Prior to insertion, at the beginning of iteration $k$, $B$ had degree 2 and was part of a chain of degree 2 nodes connecting $E$ and $F$. The chain was line contracted into $BST(E,F)$. After insertion, $B$ has degree 3 with the chain accounting for 2 and $A$ accounting for 1. The two sides of the chain are now line contracted, independently, into $BST(E,B)$ and $BST(B,F)$, respectively.  $A$ is subsequently leaf contracted into $LST(B)$. Thus, $B$ survives an extra iteration and we examine the fate of $B$, $E$, and $F$ in the cases below.

To repair the data structure, we split the former $BST(E,F)$ into $BST(E,B)$ and $BST(B,F)$ which we associate with edges $(E,B)$ and $(B,F)$, respectively. We also insert edge $(A,B)$ into $LST(B)$. Since $B$ survives one extra iteration, {\sc round}$(B)$ increases by 1 after Up Correction. To avoid confusion, throughout the rest of this section we continue to let {\sc round}$(B) = k$ refer to the pre-insertion value unless otherwise specified. We also assume w.l.o.g. that min $\{${\sc round}$(E)$, {\sc round}$(F)\}$ $=$ {\sc round}$(E)$.

\begin{description}

\item[Case 4]  This case has the following two sub-cases:
\begin{description}
\item [(a) $\round{B}+1< \round{E}, \round{F}$.] After insertion, at the beginning of iteration $k+1$, $B$ has degree 2 and its only neighbors are $E$ and $F$. Thus, $B$ is the only node line contracted into $BST(E,F)$.
\item [(b) $\round{B}+1= \round{E} < \round{F} \mbox{ and }$ \typeleaf{$E$}.] Like  (a), at the beginning of iteration $k+1$, $B$ has degree 2 and its only neighbors are $E$ (with degree 1) and $F$. $B$ is the only node line contracted into $BST(E,F)$, with $E$ subsequently leaf contracted into $LST(F)$ in the same iteration.
\end{description}
In both cases, to repair the data structure, we create a new $BST(E,F)$ and insert the edges queries $(E,B)$ and $(B,F)$.  Then we replace the existing search tree associated with the edge query $(E,F)$ with our new $BST(E,F)$.

\item [Case 5: $\round{B}+1=\round{E}$ and \typeline{$E$}]  Before insertion, $B$ was line contracted into $BST(E,F)$ at iteration $k$ and $E$ was line contracted into some $BST(G,H)$ at iteration $k+1$. After insertion, at the beginning of iteration $k+1$, $E$ and $B$ are part of a chain of degree 2 nodes connecting $G$ and $H$.

If \round{$E$} $<$ \round{$F$}, then $H=F$. Thus, $B$ should be line contracted into $BST(G,F)$ together with $E$. Similarly, if \round{$E$} $=$ \round{$F$} and \typeleaf{$F$}, then $H=F$ and $B$ should again be line contracted into $BST(G,F)$ together with $E$. Finally, if \round{$E$} $=$ \round{$F$} and \typeline{$F$}, then $E$, $B$, and $F$ become part of the same chain of degree-2 nodes between $G$ and $H$, where \parentis{$E$}{\parentis{$F$}{$(G,H)$}}. $E$, $B$, and $F$ should all be line contracted into $BST(G,H)$.

In all cases, to repair the data structure, we remove $(E,F)$ from $BST(G,H)$ and we add edge queries $(E,B)$ and $(B,F)$ to $BST(G,H)$, where \parentis{$E$}{$(G,H)$} before insertion.

\item [Case 6: $\round{E} < \round{B}+ 1 \leq \round{F}$]  \parentis{$B$}{$(E,F)$} implies $k \leq$ \round{$E$}, with equality implying \typeleaf{$E$}. Since \round{$E$} $< k+1$ we conclude $k =$ \round{$E$} and \typeleaf{$E$}. Before insertion, at the beginning of iteration $k$, $E$ had degree 1 and $B$ was part of a chain of degree 2 nodes connecting $E$ and $F$. The chain was line contracted into $BST(E,F)$ and then $E$ was leaf contracted into $LST(F)$. After insertion, $B$ has degree 3, connected to each of $E$ (degree 1), $A$ (degree 1), and $F$ by (possibly empty) chains of degree 2 nodes. Thus, $E$ and $A$ should now be leaf contracted into $LST(B)$. To repair the data structure, we remove $(E,F)$ from $LST(F)$ and we add $(E,B)$ to $LST(B)$.

In the event that \parentis{$F$}{{\sc null}} and {\sc round}$(F) = k+1$, we must reconfigure the top of the \lltree. We examine $LST(F)$, which has not yet been modified by the current call to \upc.   Without loss of generality, $\mu_1(F) = k$ and $\rho_1(F) = E$. If $\mu_2(F) < k$, then Lemma~\ref{lemma:value-lemma} implies $\mu_2(F) = \mu_3(F) = k-1$. Thus, at the beginning of iteration $k$, $B$ has degree 3 as described above and all of $E$, $A$, and $F$ have degree 1. They should be leaf contracted into $LST(B)$, leaving $B$ as the new root and sole node to survive until iteration $k+1$. However, if $\mu_2(F) = k$, then Lemma~\ref{lemma:value-lemma} implies $\mu_2(F) = \mu_3(F) = k$. Thus, at the beginning of iteration $k$, $B$ has degree 3 as described above, but $F$ has degree at least $3$ as well. After another leaf contraction, only $B$ and $F$ remain. For consistency, we choose $B$ as the root (replacing $F$). To repair the data structure, we replace $F$ with $B$ as the root node (we decrease \round{$F$} to $k$). We also insert $(B,F)$ into $LST(B)$.

Otherwise (\parent{$F$} $\neq$ {\sc null} or {\sc round}$(F) \neq k+1$), so we can apply \upc recursively to determine how $B$ and $F$ interact. Thus, the fate of edge query $(B,F)$ is determined by up-correcting $B$ at insertion point $F$.
\end{description}

What's left to show is that \upc operates in $O(h)$ time.  In all of the non-recursive cases (1-5, parts of 6) \upc performs at most $O(1)$ BST or LST operations.  Furthermore, all calls to \downc occur in non-recursive cases, so, by Lemma~\ref{lemma:down-correction} all the non-recursive cases meet the desired bound.  Now we address the recursive part of case 6.  Consider the path in the \lltree from the root of $\mathcal{T}$ to $B$.  This path includes a sequence of queries in $LST(F)$ down to $(F,E)$, a sequence of queries in $BST(E,F)$ ending at $B$.   We associate $O(1)$ tokens with each edge query on this path and $O(1)$ tokens with each node appearing in an edge query on this path so that the total number of tokens allocated is at most $O(h)$.  We use these tokens to pay for the operations collectively performed by {\em all the recursive calls}.  First, note that because the round of the insertion point always increases, we never consider a particular LST or BST more than twice on any complete execution of \upc.  Prior to the recursive call, we insert $(A,B)$ and $(B,E)$ into $LST(B)$ which takes $2$ tokens away from $B$ since insertion occurs at the head of the LST.   We remove edge $(E,F)$ from $LST(F)$ which takes $1$ token away from each edge query in $LST(F)$ leading down to $(E,F)$.  Finally, the splitting $BST(E,F)$ takes time proportional to the height of $BST(E,F)$, so we can pay for this operation using the tokens allocated to the edge queries appearing in $BST(E,F)$.  Since we only consider each of these data structures twice on any complete execution of \upc, we never run out of tokens.  Thus, we have the desired bound.  
\end{proof}

\subsection{Delete} \label{sec:delete}

In this section we describe an efficient method to restructure the \lltree when a node $A$ is deleted from $\mathcal{T}$.  We begin with a definition.

\begin{definition}
Let $B$ be a node in the \lltree such that \roundis{$B$}{$k$}.  $B$ is {\em fragile} if any of the following hold:
\begin{enumerate}
\item \parentis{$B$}{{\sc null}} and $\mu_1(B) = k-1 > k-2 = \mu_2(B) = \mu_3(B) \geq  \mu_4(B)$ or $\mu_1(B) = \mu_2(B) = \mu_3(B) = k-1 > \mu_4(B)$.
\item $\typeleaf{B}$ and $\mu_1(B) = \mu_2(B) = k-1 > \mu_3(B)$.
\item $\typeline{B}$ and $\mu_1(B) = k-1 > \mu_2(B)$.
\end{enumerate}
\end{definition}

A fragile node is one that barely adheres to Lemma~\ref{lemma:value-lemma}.  Fragile nodes play an important role in deletion because they are not robust to changes:  removing an edge query $(A,B)$ from $LST(B)$ where $B$ is a fragile node and $\round{A}=\mu_{1}(B)$ (or, in the case that \parentis{$B$}{{\sc null}} and $\mu_{1}(B) > \mu_2(B) = \mu_3(B) > \mu_4(B)$,  when $\round{A}=\mu_{3}(B)$) invalidate the contraction process.  This motivates the notion of instability:  a node $B$ becomes {\em unstable} if and only if we remove some edge query $(A,B)$ from $LST(B)$, or change the round of some node $A$ such that $(A,B)$ appears in $LST(B)$, and Lemma~\ref{lemma:value-lemma} no longer holds.  A node is {\em stable} if it adheres to Lemma~\ref{lemma:value-lemma}.  Table~\ref{tab:stabilize} gives a formal description of the {\sc stabilize} procedure which repairs the contraction process in a \lltree $\mathcal{T}$ when a single node $B$ becomes unstable.

\begin{lemma} \label{lemma:stabilize}
{\sc stabilize} correctly repairs the contraction process when a single node becomes unstable in $O(\log w) \cdot OPT$ time so that $\mathcal{T}=\mathcal{T'}$.
\end{lemma}

\begin{proof}
We begin by proving correctness by structural induction.  Suppose $B$ is unstable.

\begin{table}[tbh]
\begin{center}
\caption{{\sc stabilize} : Given a \lltree $\mathcal{T}$ with a single unstable node $B$, repair $\mathcal{T}$ so that every node is stable. \label{tab:stabilize}}
\scriptsize
\renewcommand{\tabcolsep}{2mm}
\renewcommand{\arraystretch}{1.4}
\begin{tabular}{|l|l|}
\hline

\multicolumn{2}{|l|}{{\sc stabilize}$(B)$ } \\ \hline \hline

\multicolumn{1}{|c|}{{\bf Updated Properties}} & \multicolumn{1}{c|}{{\bf Data Structure Updates}} \\ \hline \hline

\multicolumn{2}{|l|}{{\bf Case 1:} \parentis{$B$}{{\sc null}}, and (1) $\mu_1(B) = \mu_2(B) = \mu_3(B) = k-2$, or} \\
\multicolumn{2}{|l|}{(2) $\mu_1(B) = k-1 > \mu_2(B) = k-2 > \mu_3(B) = k-3$, or} \\
\multicolumn{2}{|l|}{(3) $\mu_1(B) = \mu_2(B) = k-1 > k-2 = \mu_3(B)$} \\ \hline \hline

\round{$B$} $\gets$ $\mu_3(B) + 1$ & \\ \hline
\multicolumn{2}{|c|}{{\bf if} $\mu_1(B) = k-1$ {\bf then}} \\ \hline
\multicolumn{2}{|c|}{Let $\rho_1(B) = M$ and $\rho_2(B) = N$} \\ \hline
\multirow{3}{*}{\round{$M$} $\gets$ $\mu_2(B) + 1$} & remove edges $(B,M)$ and $(B,N)$ from $LST(B)$ \\
& $BST(M,B) \gets$ \downc $(N,B)$ and $(B,M)$ \\
\parent{$M$} $\gets$ {\sc null}& create edge $(M,N)$ from $BST(M,N)$ and insert it into $LST(M)$ \\

& $M$ becomes new root of the \lltree \\ \hline \hline

\multicolumn{2}{|l|}{{\bf Case 2:} \typeleaf{$B$}, \parentis{$B$}{$E$}, and $\mu_1(B) = k-1 > k-2 = \mu_2(B)$} \\ \hline \hline

\multirow{2}{*}{\round{$B$} $\gets$ $\mu_2(B) + 1$} & remove edge $(B,E)$ from $LST(E)$ \\
& remove edge $(M,B)$ from $LST(B)$ \\ \hline
\multicolumn{2}{|c|}{Let $\rho_1(B) = M$ and let $(B,N)$ be the edge query representing $B$ in $BST(B,E)$} \\ \hline
\multicolumn{2}{|c|}{{\bf if} $N \neq E$ and \roundis{$N$}{$k$} {\bf then}} \\ \hline
\multirow{2}{*}{\type{$N$} $\gets$ {\sc leaf}} & remove edge $(B,N)$ from $BST(B,E)$ which becomes $BST(N,E)$ \\
\multirow{2}{*}{\parent{$N$} $\gets$ $E$} & create edge $(N,E)$ from $BST(N,E)$ and insert it into $LST(E)$ \\
\multirow{2}{*}{\parent{$M$} $\gets$ $N$} & $BST(M,N) \gets$ \downc $BST(M,B)$ and $BST(B,N)$ \\
& create edge $(M,N)$ from $BST(M,N)$ and insert it into $LST(N)$ \\ \hline

\multicolumn{2}{|c|}{{\bf else}} \\ \hline
\multirow{2}{*}{\parent{$M$} $\gets$ $E$} & $BST(M,E) \gets$ \downc $BST(M,B)$ and $BST(B,E)$ \\
& create edge $(M,E)$ from $BST(M,E)$ and insert it into $LST(E)$ \\ \hline
\multicolumn{2}{|c|}{{\bf if} $E$ is {\em unstable} {\bf then} {\sc stabilize}($E$)} \\ \hline \hline

\multicolumn{2}{|l|}{{\bf Case 3:} \typeline{$B$}, \parentis{$B$}{$(E,F)$}, and $\mu_1(B) = k-2$} \\ \hline \hline

\multirow{3}{*}{\round{$B$} $\gets$ $\mu_1(B) + 1$} & split $BST(E,F)$ into $BST(E,B)$ and $BST(B,F)$ \\
& $BST(E,F) \gets$ \downc $BST(E,B)$ and $BST(B,F)$ \\
& attach new $BST(E,F)$ to edge $(E,F)$ \\

\hline
\end{tabular}
\end{center}
\end{table}

\subsubsection*{Case 1: \parentis{$B$}{{\sc null}}}

(1) Suppose $\mu_1(B) = k-1 > k-2 = \mu_2(B) = \mu_3(B)$. Before the change, $M = \rho_1(B)$ and $B$ were the only two nodes that survived until the final leaf contraction at {\sc round} $k-1$. After deletion, we may have the following two anomalies, which may appear when $B$ is a fragile node:

(a) Suppose $\mu_1(B) = \mu_2(B) = \mu_3(B) = k-2$. This happens if $M$ changes.  Now $B$ alone survives to iteration $k-1$. We set {\sc round}$(B) = \mu_3(B) + 1 = k-1$ and keep $B$ as the root node.

(b) Suppose $\mu_1(B) = k-1 > k-2 = \mu_2(B) > k-3 = \mu_3(B) = \mu_4(B)$. This happens if the changed node is $\rho_2(B)$ or $\rho_{3}(B)$.   Call this node $N$. Now, at the leaf contraction step of iteration $k-2$, $M$ is a full node and all its children have degree 1 (including $B$). Thus, we set \roundis{$B$}{$\mu_3(B) + 1 = k - 2$}, \roundis{$M$}{$\mu_2(B) + 1 = k-1$}, and install $M$ as the new root. To repair the data structure, we remove edges $(B,N)$ and $(M,B)$ from $LST(B)$ and we create edge $(M,N)$ which we insert into $LST(M)$ (the new root) after we \downc $(M,B)$ and $(B,N)$.

(2)  Suppose $\mu_1(B) = \mu_2(B) = \mu_3(B) = k-1$. After deletion, we may have only one anomaly which may appear when $B$ is a fragile node: $\mu_1(B) = \mu_2(B) = k-1 > k-2 = \mu_3(B)$. This happens if the changed node is $\rho_3(B)$. Let $\rho_2(B) = N$. After deletion, at the beginning of iteration $k-1$, $M$ and $N$ have degree 1 and are connected by a chain of degree-2 nodes containing $B$. We line contract $B$ into $BST(M,N)$ and choose $M$ as root. Thus we set \roundis{$B$}{$\mu_3(B) + 1 = k - 1$}, \roundis{$M$}{$\mu_2(B) + 1 = k$}. To repair the data structure, we remove edges $(B,N)$ and $(M,B)$ from $LST(B)$ and we create edge $(M,N)$ which we insert into $LST(M)$ (the new root) after we \downc $(M,B)$ and $(B,N)$.

\subsubsection*{Case 2: \parentis{$B$}{$E$}}

Let (B,N) be the edge query representing $B$ in $BST(B,E)$. If $N \neq E$ (i.e. $BST(B,E)$ was not empty) and $\roundis{N}{k}$, then after deletion, at the beginning of iteration $k-1$, node $\rho_1(B)  = M$ has degree 1 and is connected to node $N$ through a chain of degree-2 nodes containing $B$. We line contract $B$ into $BST(M,N)$ and leaf contract $M$ into $LST(N)$. Subsequently, at the beginning of iteration $k$, $N$ has degree 1 and the path $N..E$ is a (possibly void) chain of degree 2 nodes. We line contract this path into $BST(N,E)$ and leaf contract $N$ into $LST(E)$. To repair the data structure, we remove edge $(B,N)$ from $BST(B,E)$ which yields $BST(N,E)$. Then, we set {\sc round}$(B) = \mu_2(B) + 1 = k-1$ and let $BST(M,N)$ be the result of \downc $BST(M,B)$ and $BST(N,B)$.  We create a new edge $(M,N)$ with $BST(M,N)$ attached and insert it into $LST(N)$. Finally, we create a new edge $(N,E)$ with $BST(N,E)$ attached and we insert it into $LST(E)$. Since $N$ effectively replaces $B$ (with the same {\sc round}, {\sc type}, and {\sc parent}) in $LST(E)$, $E$ cannot become unstable and the correction process is complete.

Otherwise, after the change to $LST(B)$, at the beginning of iteration $k-1$, node $\rho_1(B)  = M$ has degree 1, and is connected to node $E$ through a chain of degree-2 nodes containing $B$. We line contract $B$ into $BST(M,E)$ and leaf contract $M$ into $LST(E)$. To repair the data structure, we set {\sc round}$(B) = \mu_2(B) + 1 = k-1$ and let $BST(M,E)$ be the result of \downc $BST(M,B)$ and $BST(B,E)$.  We create a new edge $(M,E)$ with $BST(M,E)$ attached and replace edge query $(B,E)$ in $LST(E)$ with $(M,E)$.  Since $E$ has now lost a node of {\sc round} $k$ (namely $B$) from its $LST$, it may be unstable.  If this is the case, then recursively stabilizing $E$ finishes the correction process.

\subsubsection*{Case 3: \parentis{$B$}{$(E,F)$}}

After the change to $LST(B)$, at the beginning of iteration $k-1$, node $B$ has degree 2. We know $B$ is line contracted on the path $E$ to $F$. To repair the data structure, we split $BST(E,F)$ into $BST(E,B)$ and $BST(B,F)$ and let the new $BST(E,F)$ be the result of \downc $BST(E,B)$ and $(B,F)$.  where {\sc round}$(B) = \mu_1(B) + 1 = k-1$.

What's left to show is that we can perform these operations in $O(\log w) \cdot OPT$ time.  Each recursive call to {\sc stabilize} performs at most $O(1)$ BST operations and at most 1 call to \downc.  The call to \downc always happens with $BST(X,Y)$ and $BST(Y,Z)$ where $\roundis{Y}{k-1}$, the round of $BST(X,Y)$ is $k-1$ and the round of $BST(Y,Z)$ is at most $k$.  Thus, \downc will recursive at most once before hitting a base case.  This means we perform at most $O(1)$ BST operations for each \downc call.  Since there are at most $O(\log w)$ recursive calls and each BST operation takes at most $O(OPT)$ time, we have the desired bound.  
\end{proof}

\begin{table}[tb]
\begin{center}
\scriptsize
\caption{{\sc delete:} given a node $A$, remove $A$ from $\mathcal{T}$ and repair the contraction process so that $\mathcal{T}=\mathcal{T'}$ \label{tab:delete}}
\renewcommand{\tabcolsep}{2mm}
\renewcommand{\arraystretch}{1.4}
\begin{tabular}{|l|l|}
\hline
\multicolumn{2}{|l|}{{\sc delete}$(A)$} \\ \hline
\multicolumn{1}{|c|}{{\bf Updated Properties}} & \multicolumn{1}{c|}{{\bf Data Structure Updates}} \\ \hline \hline

& remove edge $(A,B)$ from $H_{S'}$ \\
& replace $A$ with $B$ nominally \\
& merge $LST(A \leftarrow B)$ and $LST(B)$ \\ \hline
\multicolumn{2}{|c|}{{\bf if} \round{$B$} $<$ \round{$A$} or} \\
\multicolumn{2}{|c|}{\round{$B$} $=$ \round{$A$} and \typeleaf{$A$}, \typeline{$B$} {\bf then}} \\ \hline
\type{$B$} $\gets$ \type{$A$} & \\
\round{$B$} $\gets$ \round{$A$} & insert $LST(B)$ at location $A$ \\
\parent{$B$} $\gets$ \parent{$A$} & \\ \hline
\multicolumn{2}{|c|}{{\bf else} keep $LST(B)$ at location $B$} \\ \hline \hline

\multicolumn{2}{|c|}{{\bf finally if} $B$ is $unstable$ {\bf then} {\sc stabilize} $B$} \\

\hline
\end{tabular}
\end{center}
\end{table}

With {\sc stabilize} in hand, we can formally define {\sc delete} which appears in Table~\ref{tab:delete} and prove its correctness and time bound.

\begin{theorem} \label{thm:delete}
{\sc Delete} takes $O(\log w) \cdot OPT$ time.
\end{theorem}

\begin{proof}
Let $A$ be the node we wish to delete and let $B$ be its predecessor in $H_{S}$.  Let $S'=S \setminus \{A\}$.   In $T_{S'}$, all the successors of $A$ are successors of $B$.  Thus, as a starting point in deletion, we must
\begin{enumerate}
\item remove the edge $(A,B)$ from $H_{S}$;
\item replace $A$ with $B$ in all edges $(A,X)$ in $H_{S}$ where $X \neq B$; and
\item insert every edge query $(A,X)$ from $LST(A)$ into $LST(B)$ where $X \neq B$.
\end{enumerate}
If, before deletion, either (1) $\round{A} > \round{B}$, or (2) $\round{A} = \round{B}$ and $\typeleaf{A}, \typeline{B}$, then $B$ essentially replaces $A$ in the remaining rounds of the contraction process so $\type{B}=\type{A}$, $\round{B}=\round{A}$, and $\parent{B}=\parent{A}$.  However, if $\round{A}$ $\leq$ $\round{B}$, then $B$ lasts as long as $A$ in the contraction process so there is no need, initially, to update its properties.  As with insertion, these steps can be performed in $O(h)$ time.  Of course, deleting $A$ may cause $B$ to become unstable.  We analyze when this occurs and appeal to the {\sc stabilize} procedure for correctness.

\begin{description}
\item [$\round{B} < \round{A}$ or $\round{A} = \round{B}$] {\bf  and $\typeleaf{A}$ and  $\typeline{B}$:}  In this case, $B$ was either (a) line contracted between $A$ and some other node $M$ or (b) leaf contracted into $A$. Let \roundis{$A$}{$k$} before deletion.
\begin{description}
\item [(a)] Suppose $B$ was line contracted before deleting $A$.  We must remove $(B,A)$ from $BST(M,A)$ which, after deletion, becomes $BST(M,B)$.  Furthermore, after deletion, $LST(B)$ contains all the edge queries from $LST(A)$.  Since $\round{B} < \round{A}$ before deletion, $B$ has a higher {\sc round} after deletion.  In fact, $B$ survives exactly as long as $A$ did before deletion, effectively replacing $A$ in all iterations of the contraction algorithm.  Because we replaced $A$ with $B$ in all edges $(A,X)$ in $H_{S}$, the contraction process has been corrected and we are finished.
\item [(b)] Suppose that $B$ was leaf contracted into $A$. Then $LST(B)$ contains all the edge queries from $LST(A)$ except for $(A,B)$.  If $A$ was a fringe node before deletion and \roundis{$B$}{$k-1$} then $B$ will not be a full node at iteration $k-1$.  Instead, it will have degree 2 which violates Lemma~\ref{lemma:value-lemma}.  In this case, we must continue to correct the contraction process which we do through the {\sc patch} procedure.  Otherwise, Lemma~\ref{lemma:value-lemma} holds and the contraction process is repaired.
\end{description}

\item [$\round{B} > \round{A}$ or $\round{B} = \round{A}$ and $\typeline{A}$:] Here, $B$ absorbs $A$ in the search tree. $B$ was either (a) line contracted together with $A$ into some $BST$, (b) the {\sc parent} of $A$, or (c) either $E$ or $F$ in the event that \parentis{$A$}{$(E,F)$}.

(a) Suppose $B$ was line contracted together with $A$. Then $B$ holds all nodes formerly in either $LST(A)$ or $LST(B)$. By Lemma~\ref{lemma:value-lemma}, after deletion $\mu_1(B) = \mu_2(B) = k-1$, and so $B$ is line contracted as before. The contraction process doesn't change any further.

(b) Suppose $B$ was \parent{$A$}. There is a possibility that before insertion \roundis{$A$}{$k-1$} and $B$ was a fringe node. If this is the case, then $B$ will not be a full node at iteration $k-1$; instead, it will have degree 2. We address this situation in Anti-Up Correction. Otherwise, Lemma~\ref{lemma:value-lemma} holds for $B$ and $B$ effectively absorbs $A$.

(c) Suppose \parent{$A$} was $(B,F)$ w.l.o.g. Then we simply add nodes to $LST(B)$ of {\sc round} smaller than \round{$B$}, without removing any others. The contraction process does not change any further.

\end{description}

To summarize, if $B$ was contracted before $A$, then $B$ replaces $A$: we put the merged $LST(B)$ in the position formerly occupied by $A$ and $B$ inherits $A$'s attributes: {\sc type}, {\sc round}, and {\sc parent}. Otherwise, $B$ absorbs $A$ and we put the merged $LST(B)$ in the position formerly occupied by $B$, while $B$ keeps its own attributes (see Table~\ref{tab:delete} for a succinct description of the algorithm).

To prove the bound on the running time, we observe that {\sc delete} makes at most $O(1)$ LST operations, each of which is at most $O(h)$.  Thus, the call to {\sc stabilize} dominates the running time of {\sc delete}.  Therefore, by Lemma~\ref{lemma:stabilize} we have the desired bound. 
\end{proof}

Deletion is the only operation for which we do not have an $O(h)$ bound on the running time.  Here is the problem:  suppose we want to delete $A$ and the path in $\mathcal{T}$ runs only through LSTs.  If we recursively need to call {\sc stabilize} on nodes appearing these LSTs, then each call to \downc may operate on a BSTs which have no ancestor / descendent relationship in the tree---with insertion, this never happens because any time we manipulate a BST, it's on path from the root down to the predecessor of the node we wish to insert.

\section{Empirical Results} \label{sec:empirical}

\begin{figure}[tb]
\subfigure[]{\includegraphics[scale=0.38]{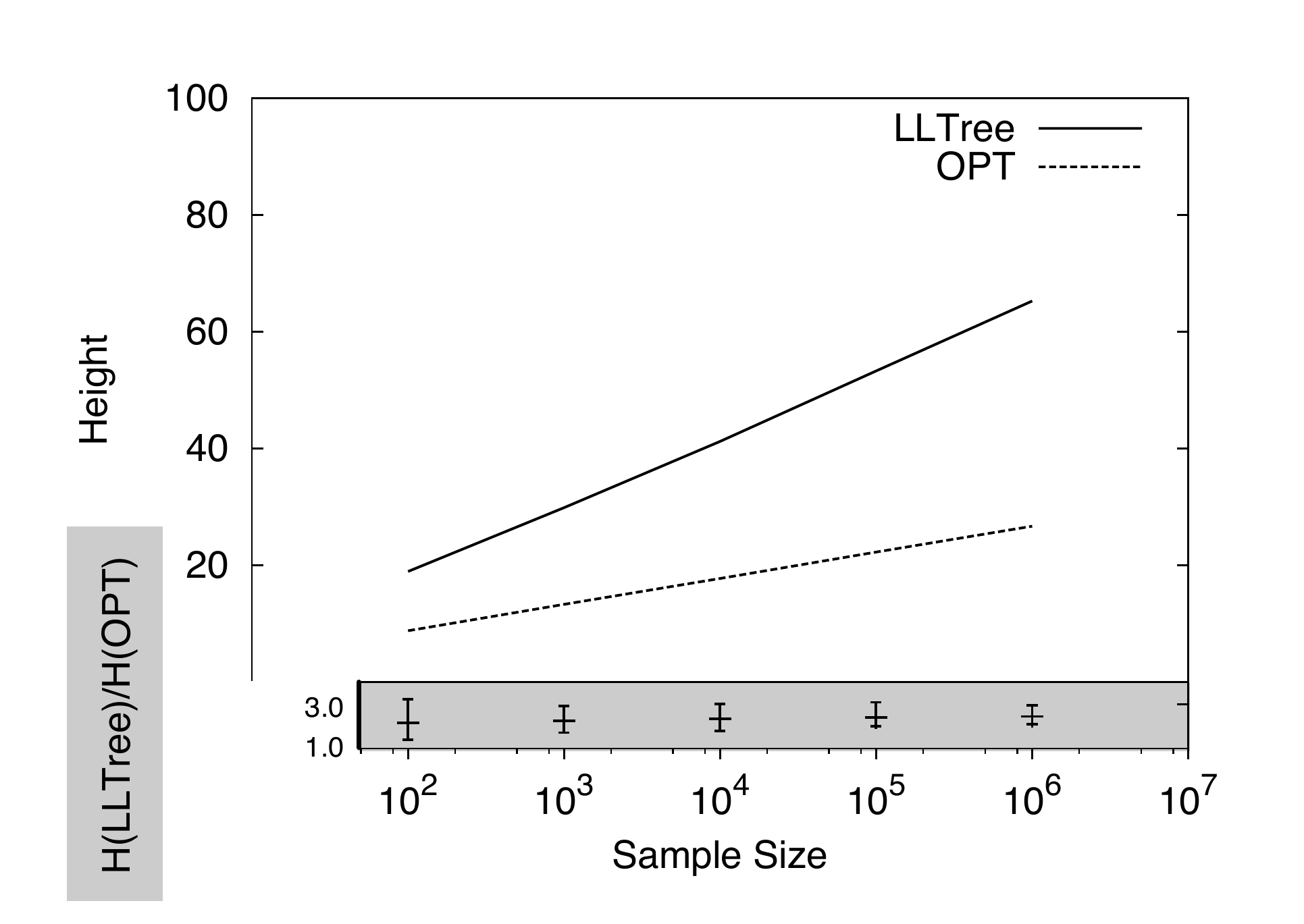}}
\subfigure[]{\includegraphics[scale=0.38]{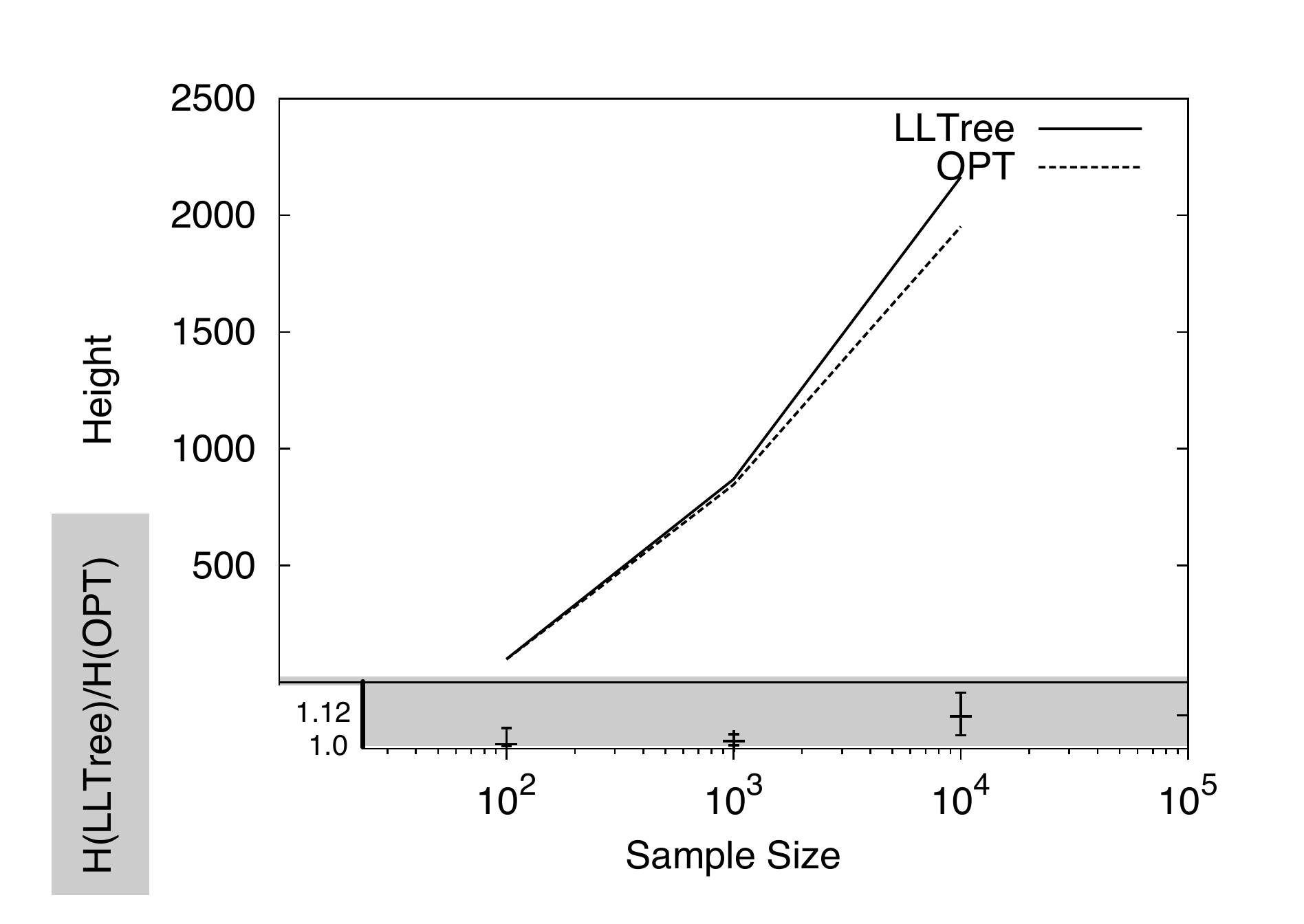}}
\caption{Results comparing the height of the \lltree to the optimal static search search tree on (a) random tree-like partial orders; and (b) a large portion of the UNIX filesystem.   The non-shaded areas show the average height of both the \lltree and optimal static algorithm.  The shaded area shows their ratio (as well as the min and max values over the 1000 iterations). \label{fig:exps}    }
\end{figure}

Here we show the results of two experiments which compare the height of a \lltree with the height of an optimal {\em static} search tree for a tree-like set $S$.  For these experiments, we consider the height of a search tree to be the maximum number of edge queries performed on any root-to-leaf path.  So any dynamic edge query in a \lltree counts as two edge queries in our experiments.

In the first experiment, we examine tree-like partial orders of
increasing size $n$.  For each $n$, we independently sample $1000$
partial-orders uniformly at random from all tree-like partial orders
with $n$ nodes.\footnote{In keeping with the uniform model for general
partial orders defined in~\cite{carmo-etal:tcs2004}, we assume
$\mathcal{P}(n)$ is the set of all rooted, labeled, oriented trees on
$1, \ldots, n$ such that every root-to-leaf path has labels that
increase.  The set $\mathcal{P}(n)$ is in one-to-one correspondence
with the set of {\em increasing trees} (these tree are also known as
{\em heap-ordered} and {\em recursive
trees})~\cite{meir-moon:cjm1978}.  The expected worst-case and average
height of a random increasing tree is $\Theta(\log
n)$~\cite{bergeron-etal:caap1992,drmota:aoc2009,grimmett:branching}.  This is in contrast
to general partial orders which, on average, have height 3.}  The
non-shaded area of Figure~\ref{fig:exps} (a) shows the heights of the
\lltree and the optimal static tree averaged over the samples.  The
important thing to note is that both appear to grow linearly in $\log n$.
We suspect that the differing slopes come mainly from the overheard of dynamic edge queries,
and we conjecture that the \lltree performs within a small constant factor of $OPT$ with high
probability in the uniform tree-like model. The shaded area of Figure~\ref{fig:exps} (a) shows the
average, minimum, and maximum approximation ratio over the samples.

Although the first experiment shows that the \lltree is competitive with the optimal static tree on {\em average} tree-like partial orders, it may be that, in practice, tree-like partial orders are distributed non-uniformly. Thus, for our second experiment, we took the {\tt /usr} directory of an Ubuntu 10.04 Linux distribution as our universe $\mathcal{U}$ and independently sampled 1000 sets of size $n=100$, $n=1000$, and $n=10000$ from $\mathcal{U}$ respectively.  The {\tt /usr} directory contains 23,328 nodes, of which 17,340 are leaves.  The largest directory is {\tt /usr/share/doc} which contains 1551 files.  The height of {\tt /usr} is 12.  We believe that this directory is somewhat representative of the use cases found in our motivation.  As with our first experiment, the shaded area in Figure~\ref{fig:exps} (b) shows the ratio of the height of the \lltree to the height of the optimal static search tree, averaged over all 1000 samples for each sample size.  The non-shaded area shows the actual heights averaged over the samples.  The \lltree is again very competitive with the optimal static search tree, performing at most a small constant factor more queries than the optimal search tree.

\subsubsection*{{\bf Acknowledgements}}

We would like to thank T. Andrew Lorenzen for his help in running the experiments discussed in Section~\ref{sec:empirical}.

\bibliographystyle{splncs}
\bibliography{../arxiv}

\newpage

\appendix

\section{Figures for Insertion and Deletion} \label{sec:figures}

Here we provide figures describing each case of the Insert, Up Correction, Down Correction, Stabilize, Transition, and Delete procedures described in the main text (Table~\ref{tab:figures}).

\begin{table}[h]
\begin{center}
\renewcommand{\tabcolsep}{2mm}
\renewcommand{\arraystretch}{1.4}
\begin{tabular}{|c|c|c|}
\hline

\multicolumn{1}{|c|}{{\bf Procedure}} & \multicolumn{1}{c|}{{\bf Case}} & \multicolumn{1}{c|}{{\bf Figure}} \\ \hline \hline

\multirow{6}{*}{Up Correction} & 1 & Fig.~\ref{fig:up1}  \\
& 2 & Fig.~\ref{fig:up2}  \\
& 3 & Fig.~\ref{fig:up3}  \\
& 4 & Fig.~\ref{fig:up4}  \\
& 5 & Fig.~\ref{fig:up5}  \\
& 6 & Fig.~\ref{fig:up6}  \\ \hline \hline

\multirow{5}{*}{Down Correction} & 1 & Fig.~\ref{fig:down1}  \\
& 2 & Fig.~\ref{fig:down2}  \\
& 3 & Fig.~\ref{fig:down3}  \\
& 4 & Fig.~\ref{fig:down4}  \\
& 5 & Fig.~\ref{fig:down5}  \\ \hline \hline

\multirow{1}{*}{Transition} & -- & Fig.~\ref{fig:transition1},~\ref{fig:transition2}  \\ \hline \hline

\multirow{5}{*}{Insert} & 1 & Fig.~\ref{fig:insert1a},~\ref{fig:insert1b}  \\
& 2 & Fig.~\ref{fig:insert2}  \\
& 3 & Fig.~\ref{fig:insert3}  \\
& 4 & Fig.~\ref{fig:insert4}  \\
& 5 & Fig.~\ref{fig:insert5}  \\ \hline \hline

\multirow{3}{*}{Stabilize} & 1 & Fig.~\ref{fig:stabilize1}  \\
& 2 & Fig.~\ref{fig:stabilize2}  \\
& 3 & Fig.~\ref{fig:stabilize3}  \\ \hline \hline

\multirow{1}{*}{Delete} & -- & Fig.~\ref{fig:delete1} \\

\hline
\end{tabular}
\end{center}
\vspace{2mm}
\caption{Index of figures describing each case of the dynamic operations on the \lltree. \label{tab:figures}}
\end{table}

\begin{figure}[tb]
\begin{center}
\includegraphics[width=\textwidth]{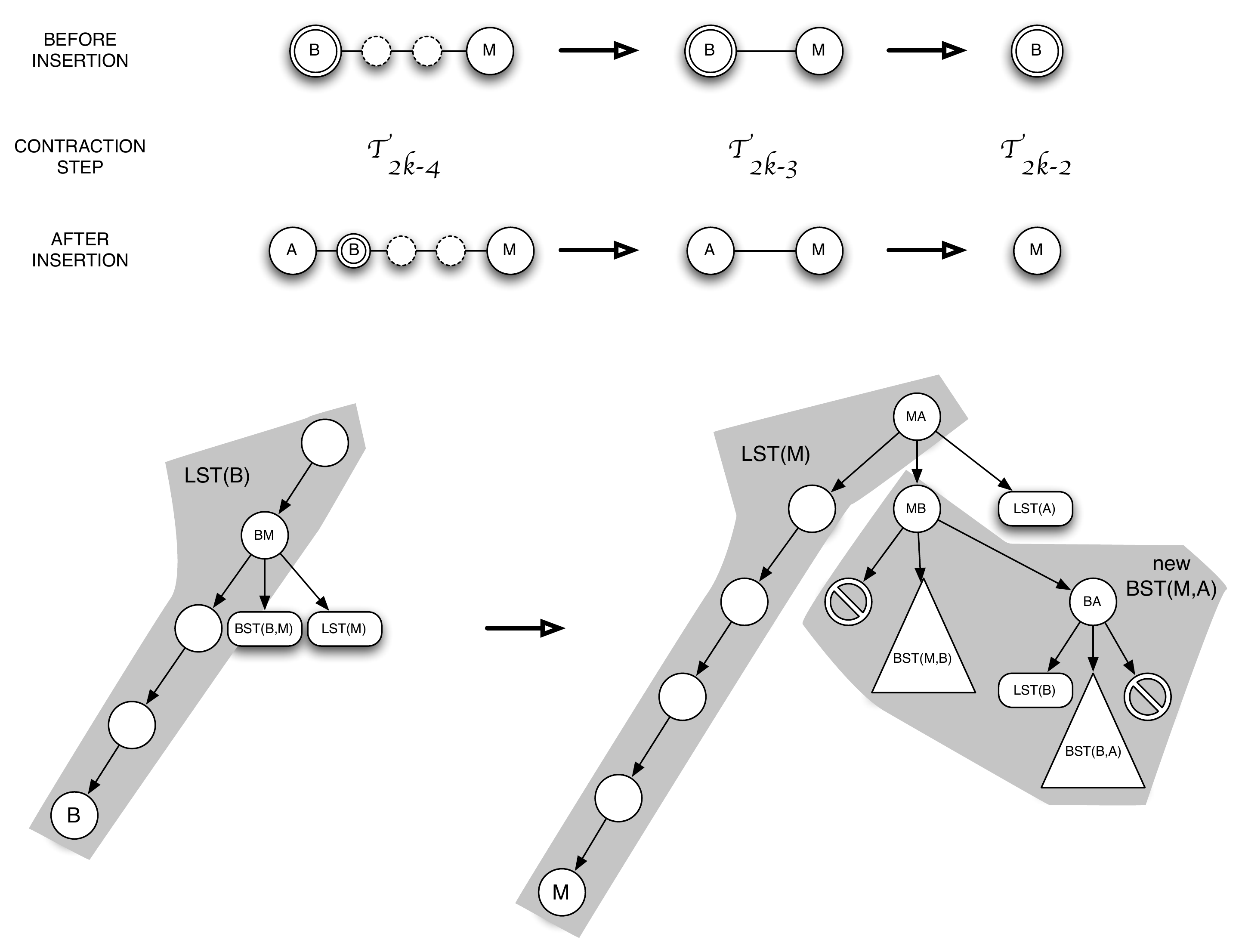}
\end{center}
\vspace{-0.2in}
\caption{Up Correction, Case 1.}
\label{fig:up1}
\end{figure}

\begin{figure}[tb]
\begin{center}
\includegraphics[scale=0.4]{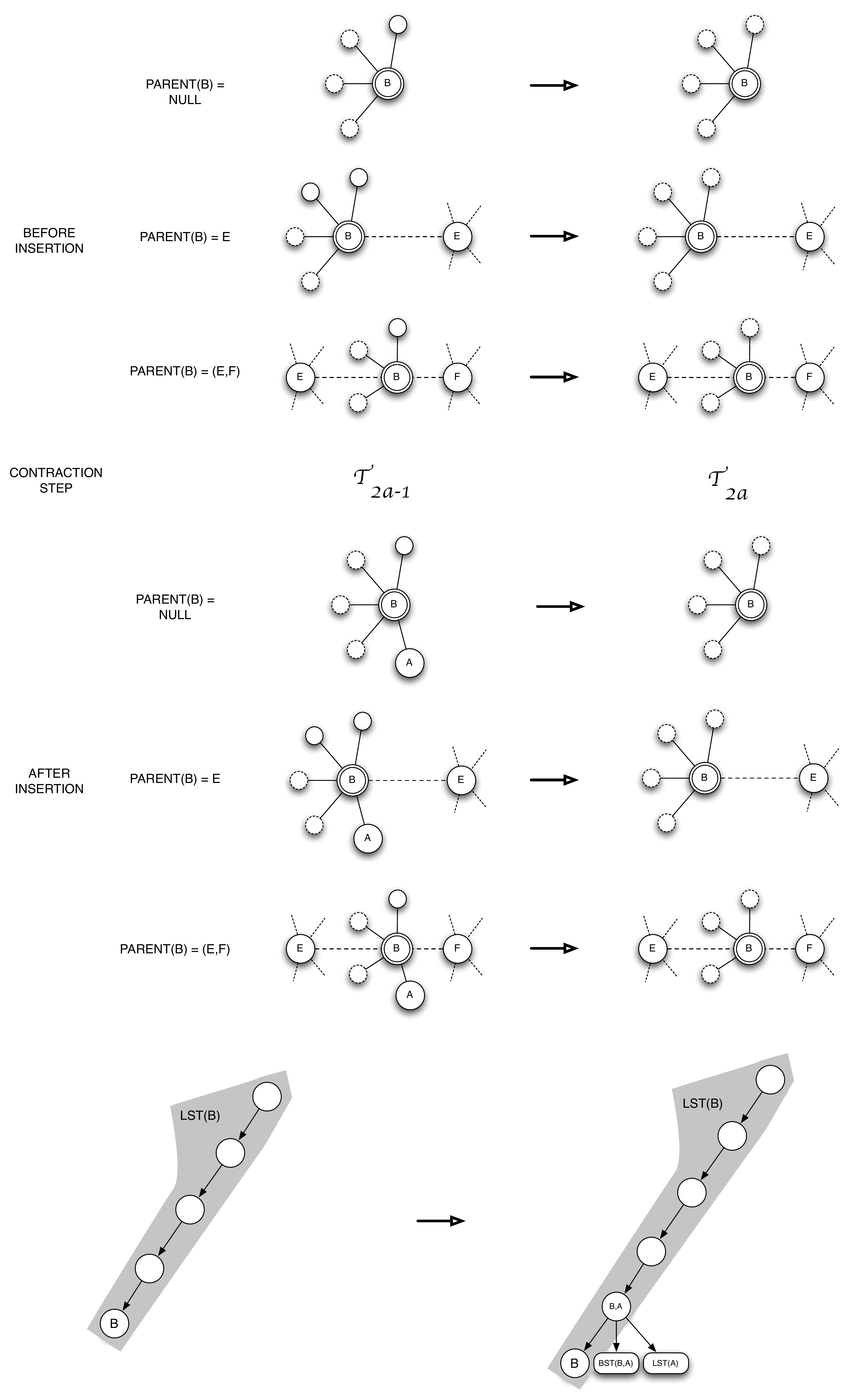}
\end{center}
\vspace{-0.2in}
\caption{Up Correction, Case 2.}
\label{fig:up2}
\end{figure}

\clearpage

\begin{figure}[tb]
\begin{center}
\includegraphics[width=\textwidth]{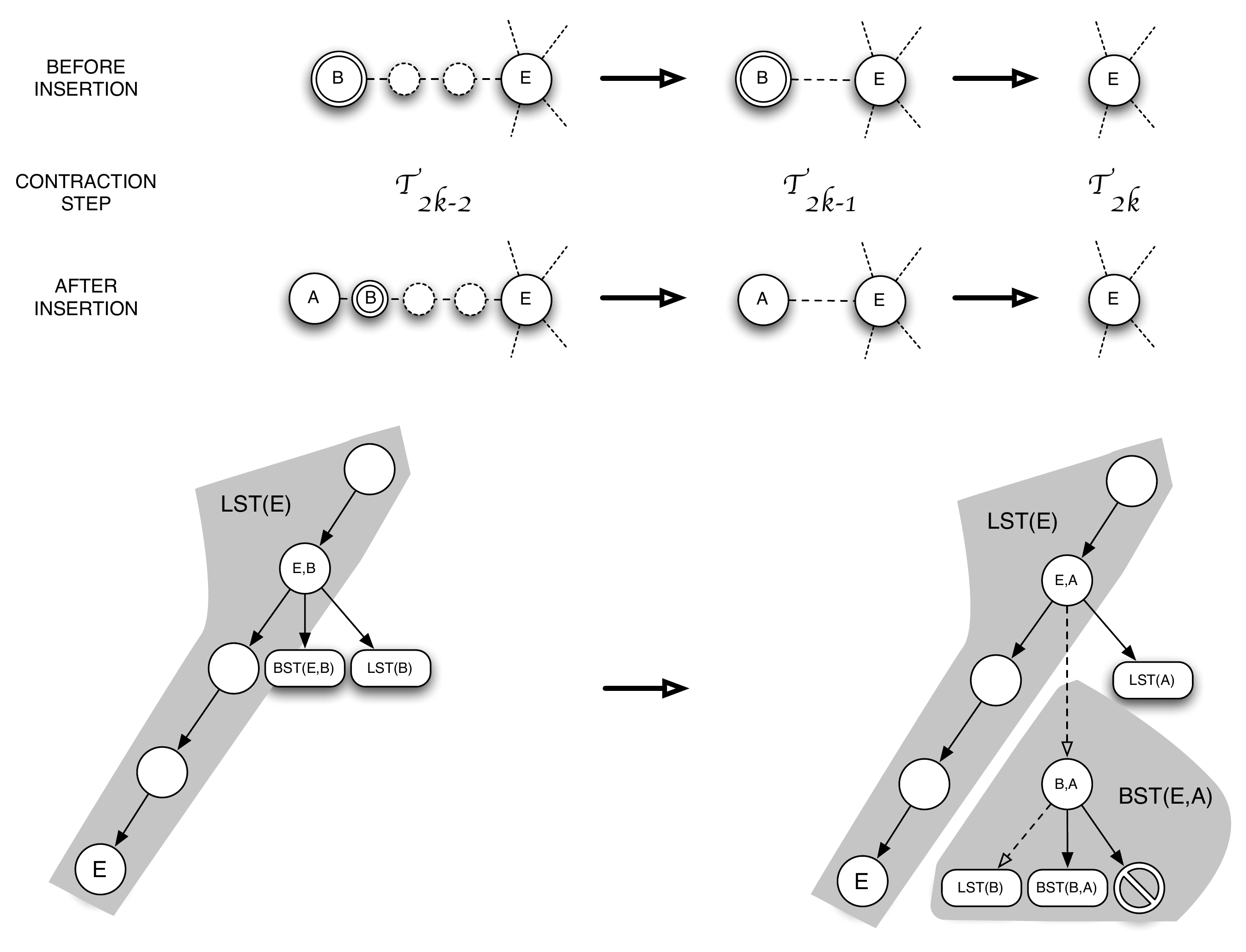}
\end{center}
\vspace{-0.2in}
\caption{Up Correction, Case 3.}
\label{fig:up3}
\end{figure}

\clearpage

\begin{figure}[tb]
\begin{center}
\includegraphics[width=\textwidth]{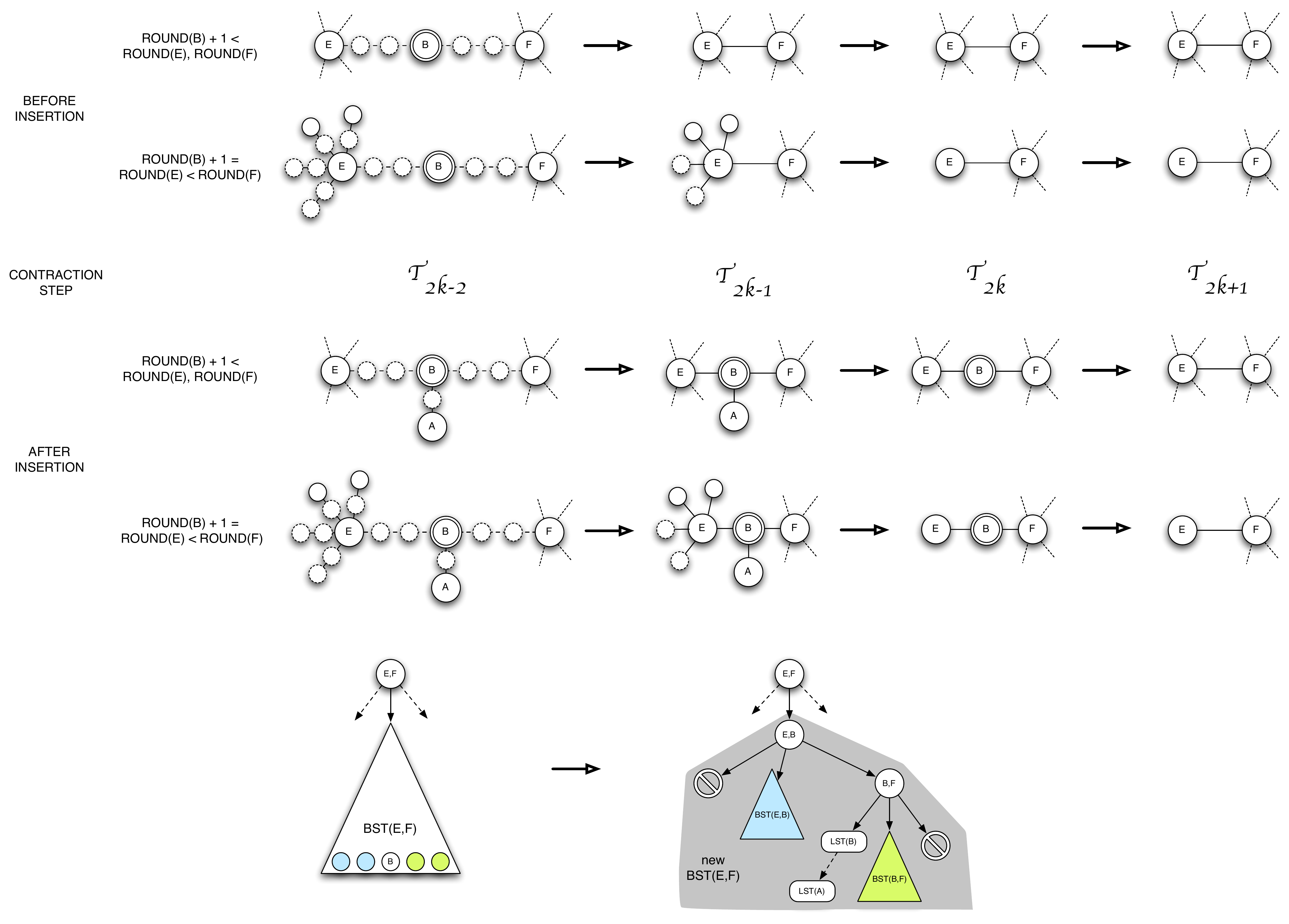}
\end{center}
\vspace{-0.2in}
\caption{Up Correction, Case 4.}
\label{fig:up4}
\end{figure}

\clearpage

\begin{figure}[tb]
\begin{center}
\includegraphics[width=\textwidth]{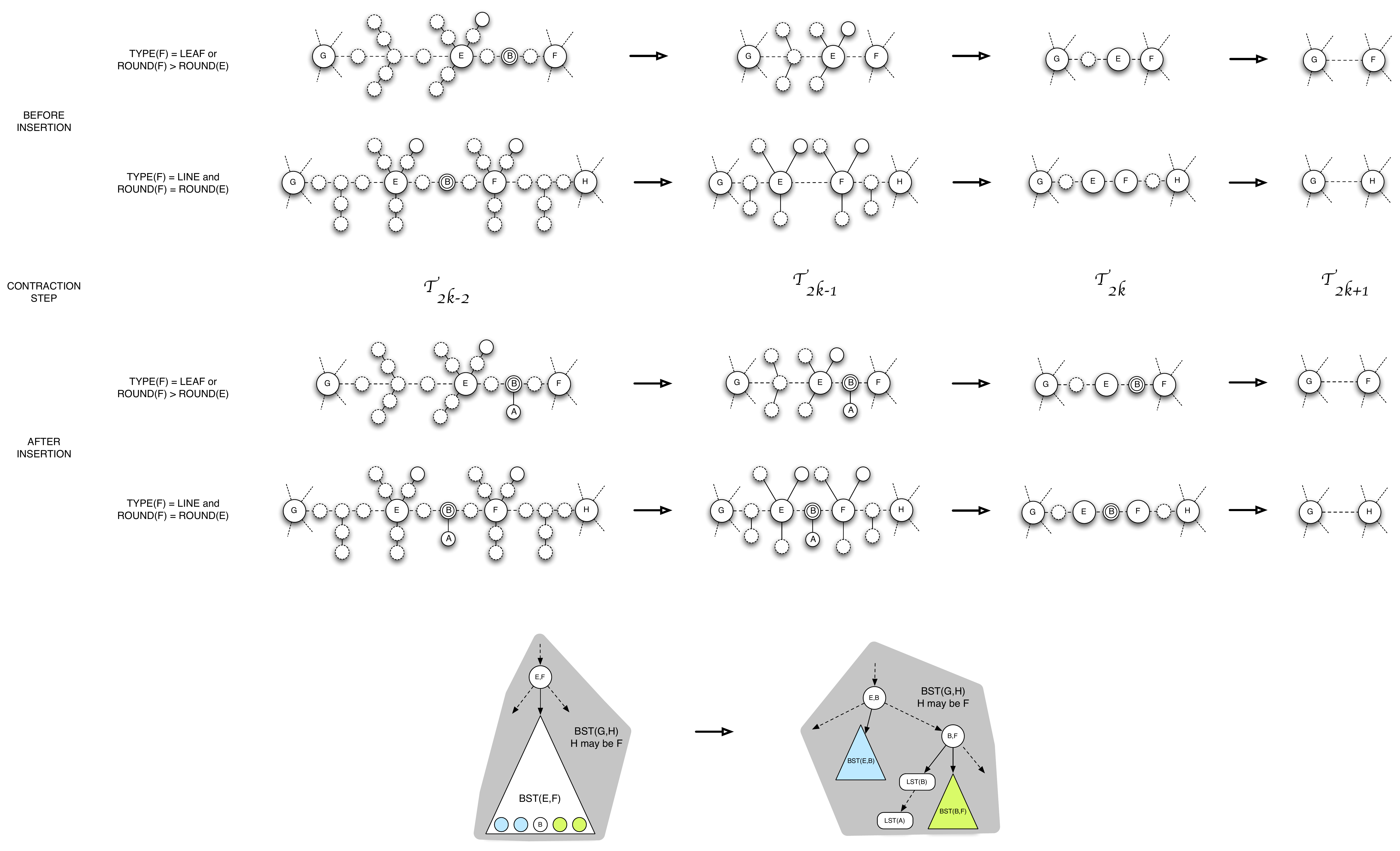}
\end{center}
\vspace{-0.2in}
\caption{Up Correction, Case 5.}
\label{fig:up5}
\end{figure}

\clearpage

\begin{figure}[htb]
\begin{center}
\includegraphics[width=\textwidth]{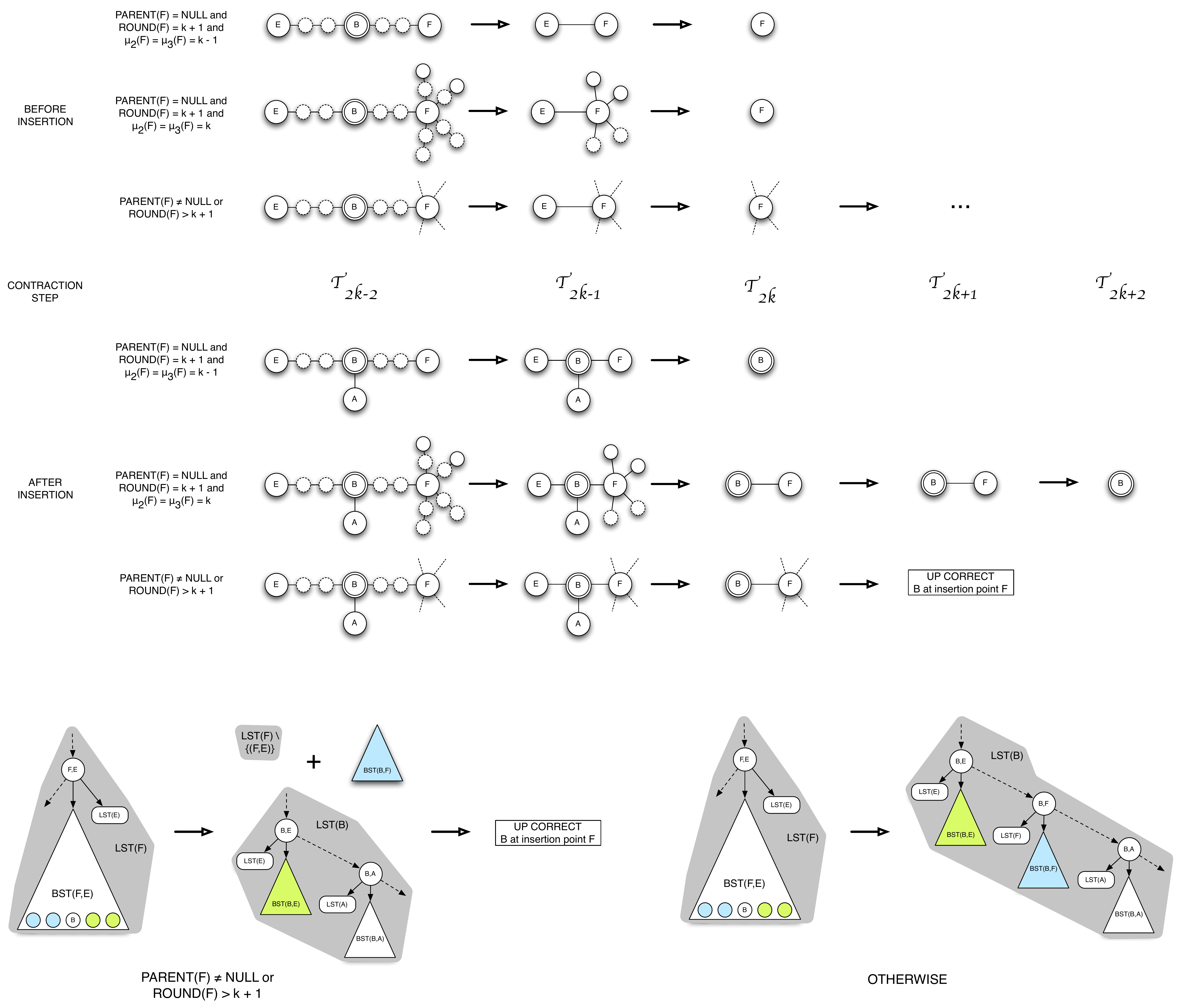}
\end{center}
\vspace{-0.2in}
\caption{Up Correction, Case 6.}
\label{fig:up6}
\end{figure}

\clearpage

\begin{figure}[tb]
\begin{center}
\includegraphics[width=\textwidth]{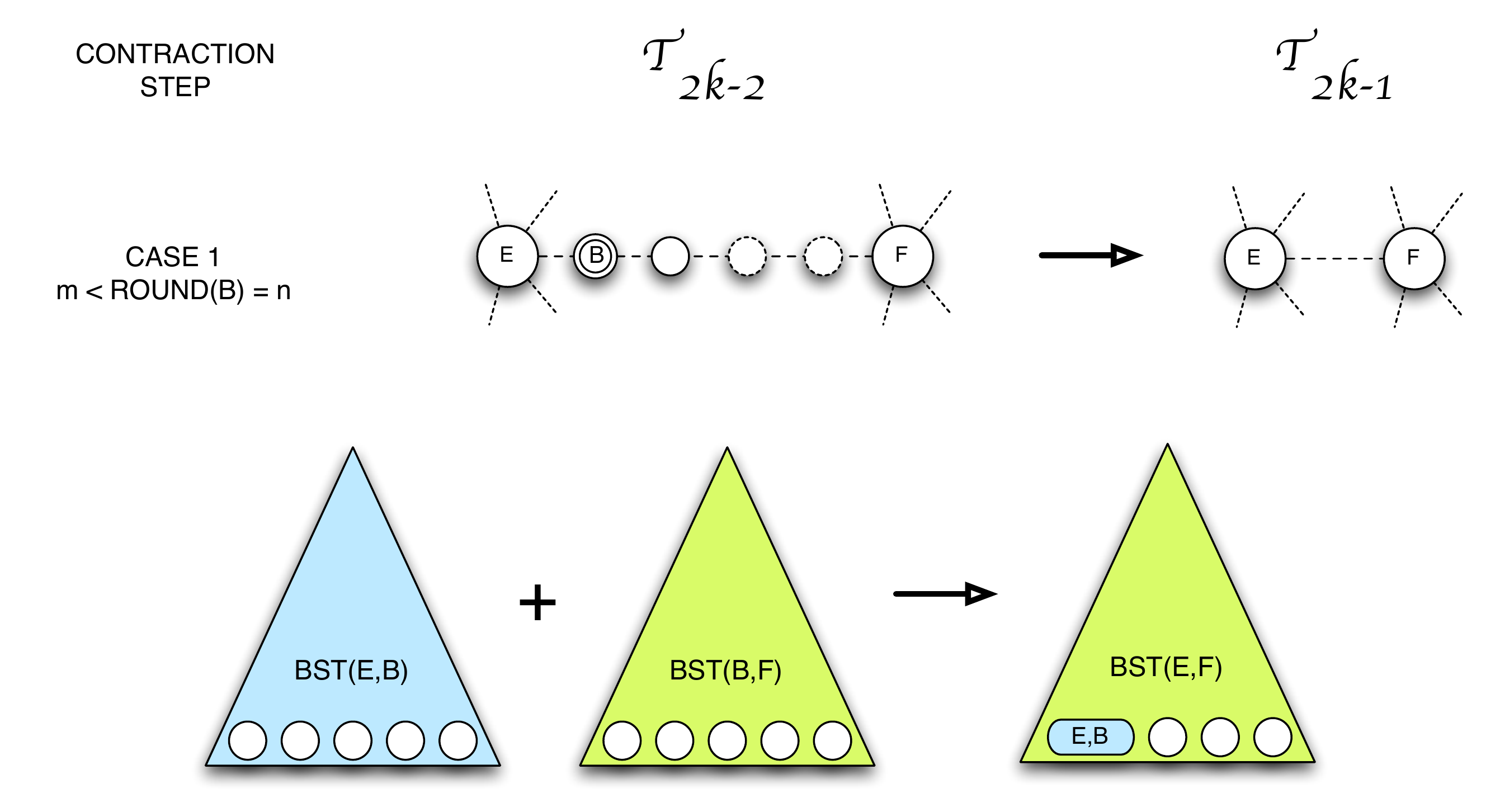}
\end{center}
\vspace{-0.2in}
\caption{Down Correction, Case 1.}
\label{fig:down1}
\end{figure}

\clearpage

\begin{figure}[tb]
\begin{center}
\includegraphics[width=\textwidth]{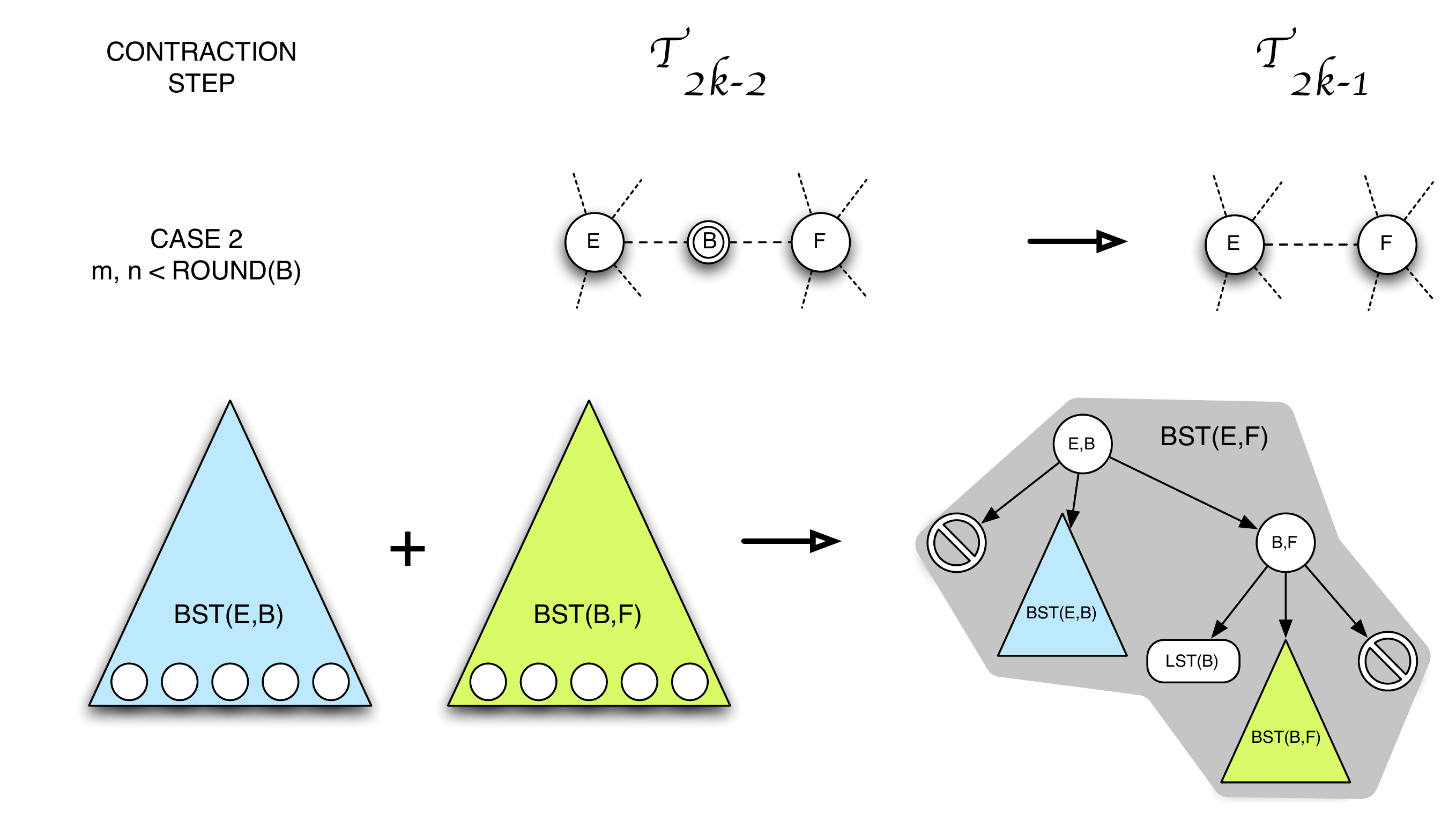}
\end{center}
\vspace{-0.2in}
\caption{Down Correction, Case 2.}
\label{fig:down2}
\end{figure}

\clearpage

\begin{figure}[tb]
\begin{center}
\includegraphics[width=\textwidth]{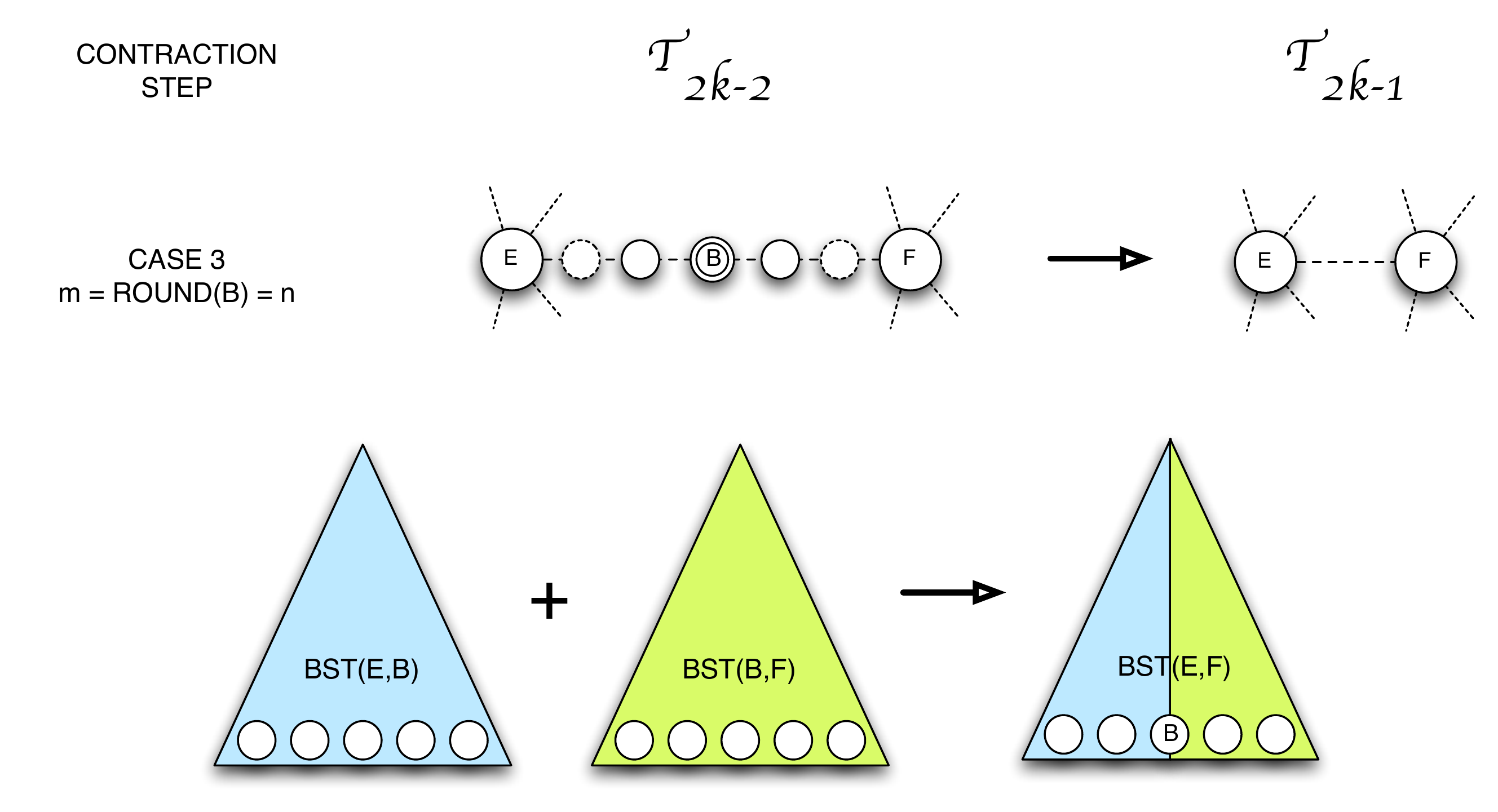}
\end{center}
\vspace{-0.2in}
\caption{Down Correction, Case 3.}
\label{fig:down3}
\end{figure}

\clearpage

\begin{figure}[tb]
\begin{center}
\includegraphics[width=\textwidth]{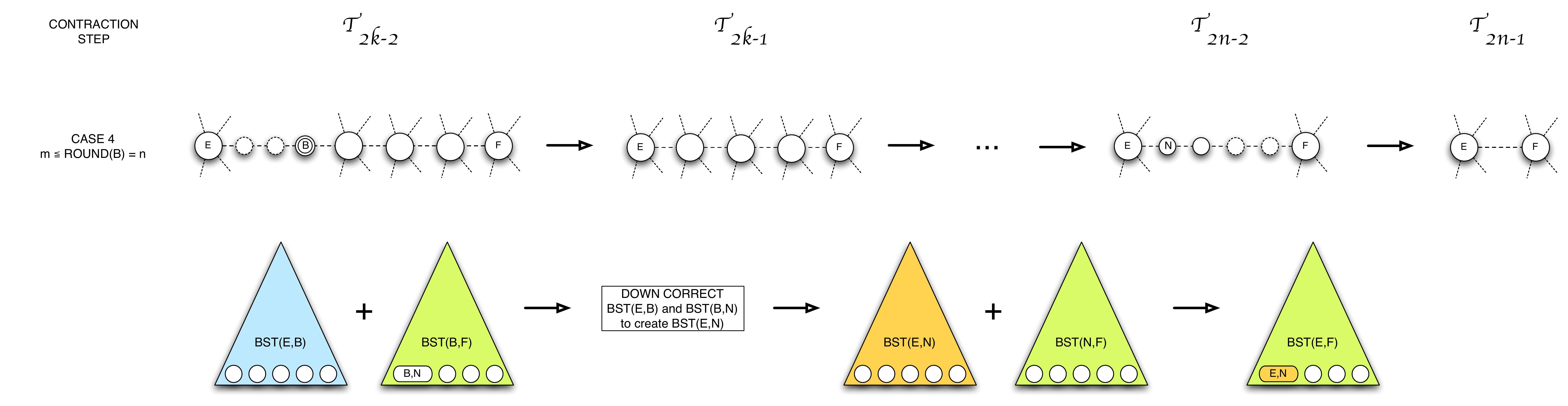}
\end{center}
\vspace{-0.2in}
\caption{Down Correction, Case 4.}
\label{fig:down4}
\end{figure}

\clearpage

\begin{figure}[tb]
\begin{center}
\includegraphics[width=\textwidth]{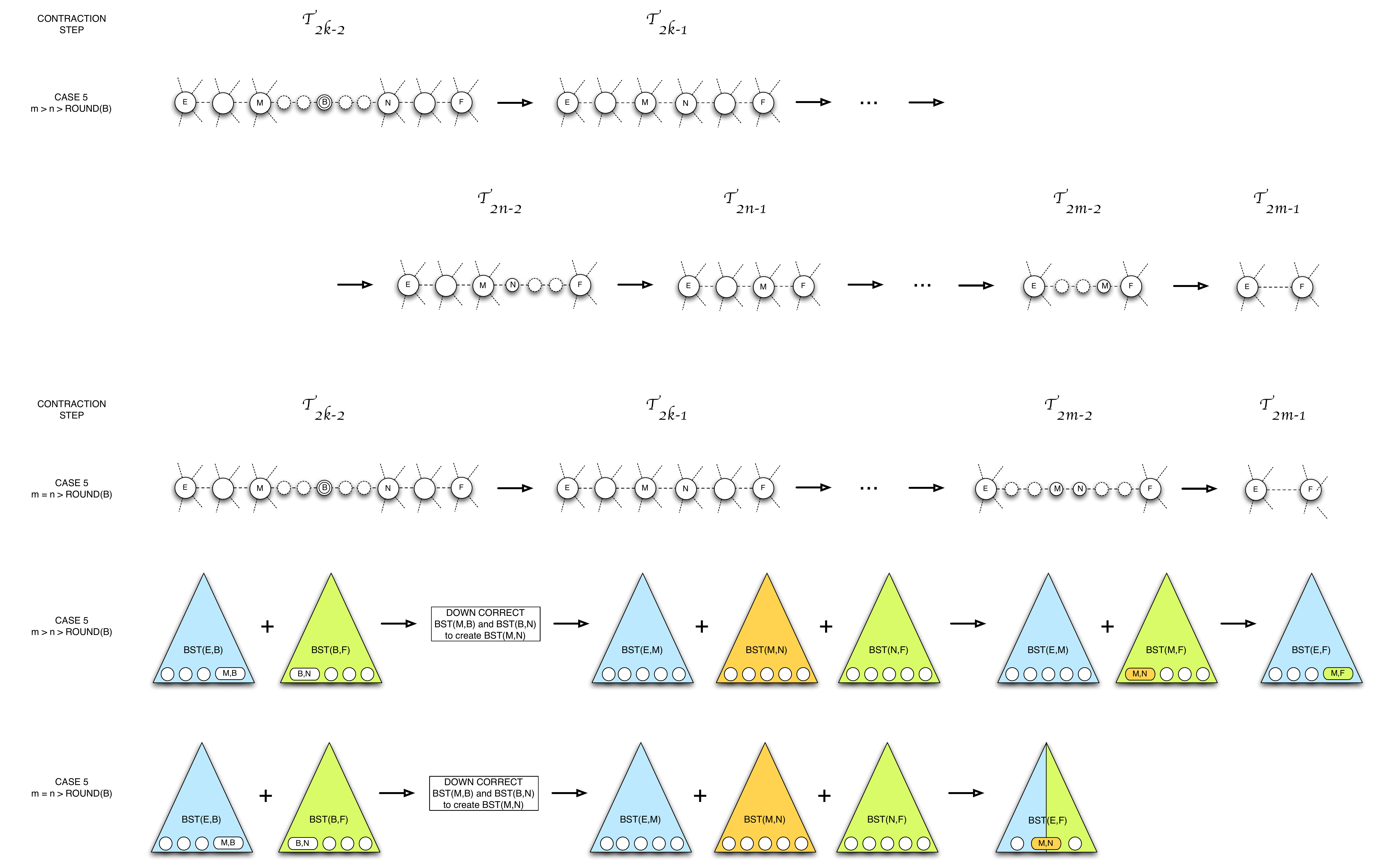}
\end{center}
\vspace{-0.2in}
\caption{Down Correction, Case 5.}
\label{fig:down5}
\end{figure}

\clearpage

\begin{figure}[tb]
\begin{center}
\includegraphics[width=\textwidth]{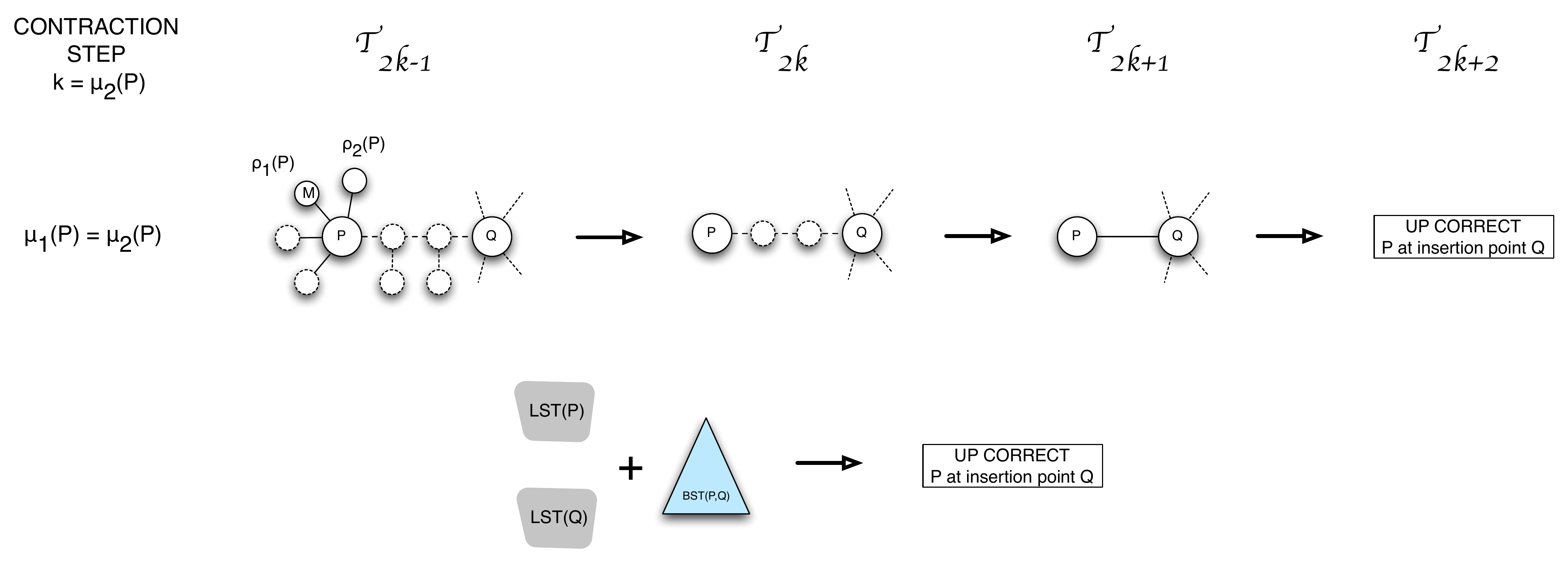}
\end{center}
\vspace{-0.2in}
\caption{Transition, Part 1.}
\label{fig:transition1}
\end{figure}

\clearpage

\begin{figure}[tb]
\begin{center}
\includegraphics[width=\textwidth]{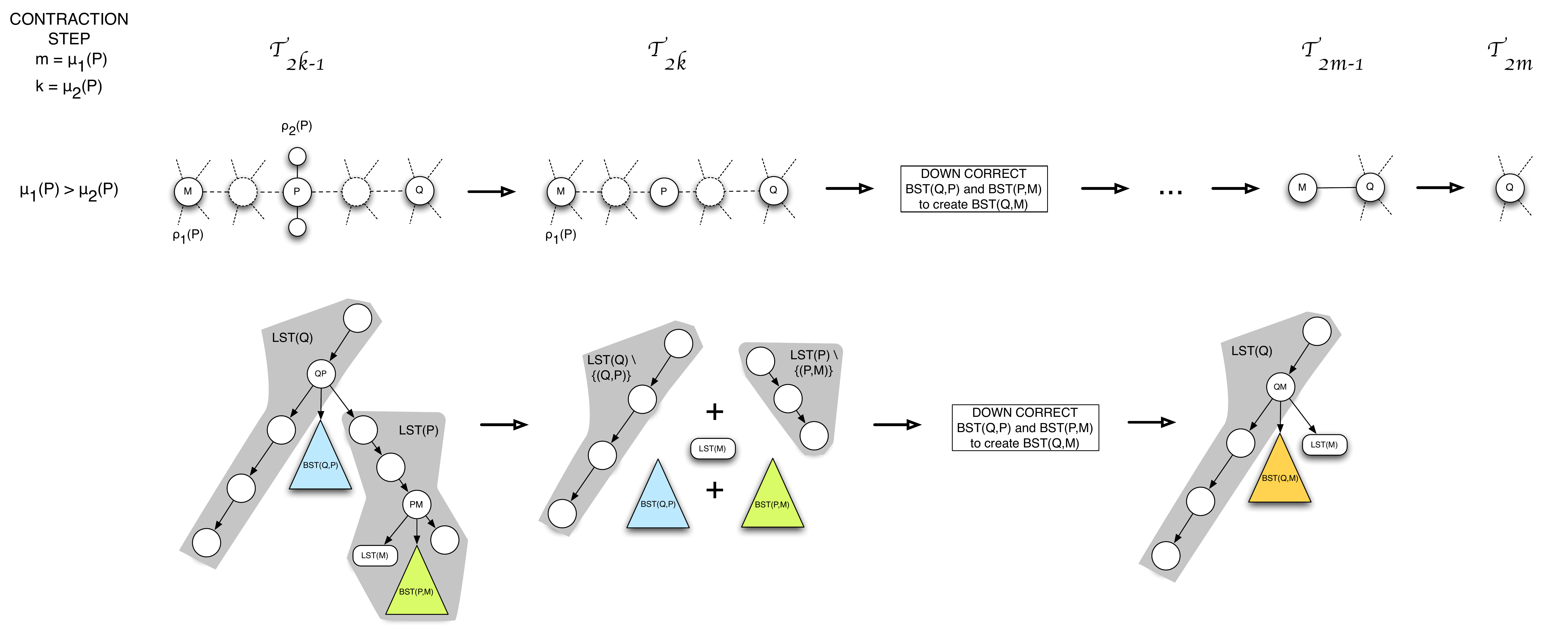}
\end{center}
\vspace{-0.2in}
\caption{Transition, Part 2.}
\label{fig:transition2}
\end{figure}

\clearpage

\begin{figure}[tb]
\begin{center}
\includegraphics[width=\textwidth]{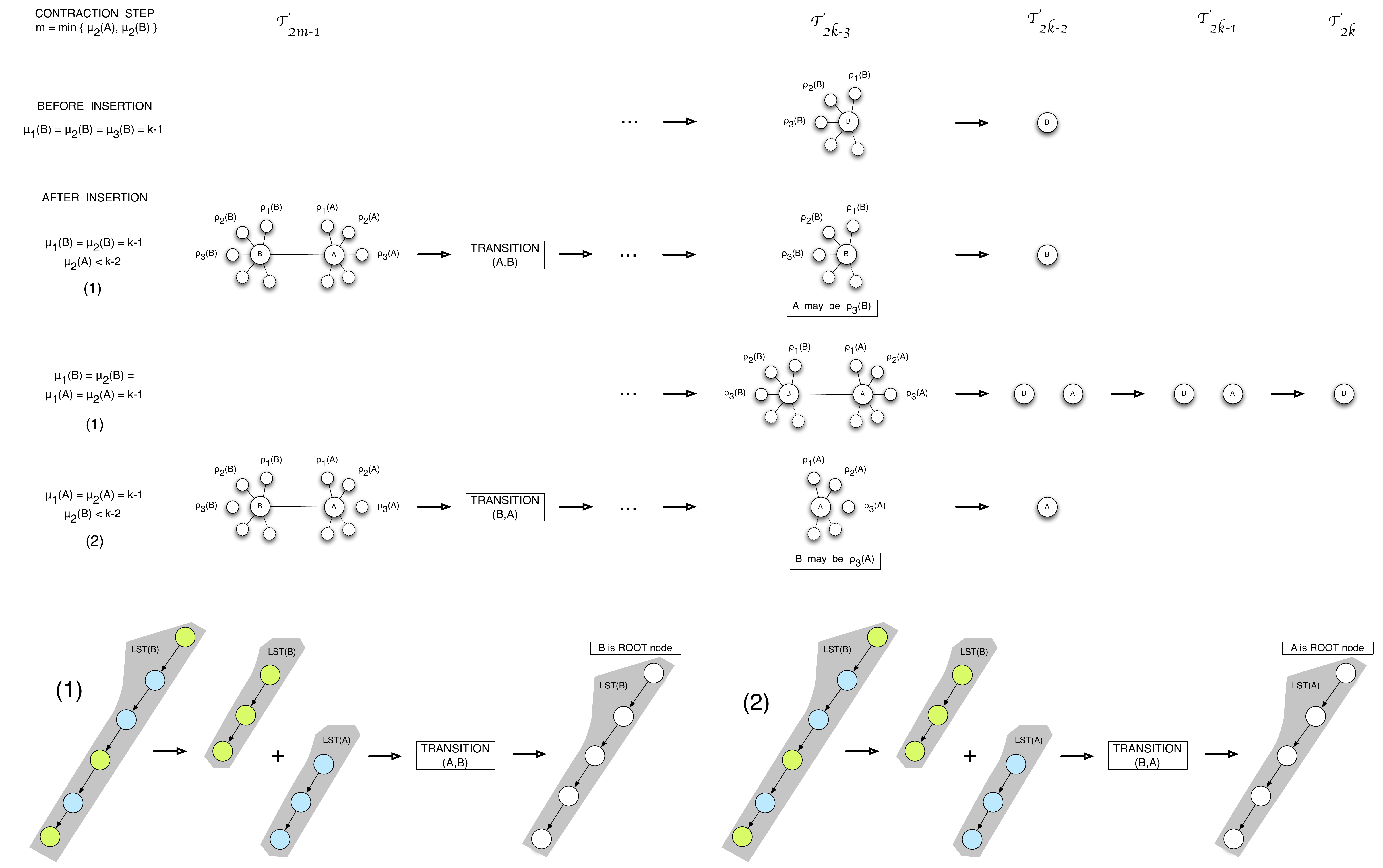}
\end{center}
\vspace{-0.2in}
\caption{Insert, Case 1, Part 1.}
\label{fig:insert1a}
\end{figure}

\clearpage

\begin{figure}[tb]
\begin{center}
\includegraphics[width=\textwidth]{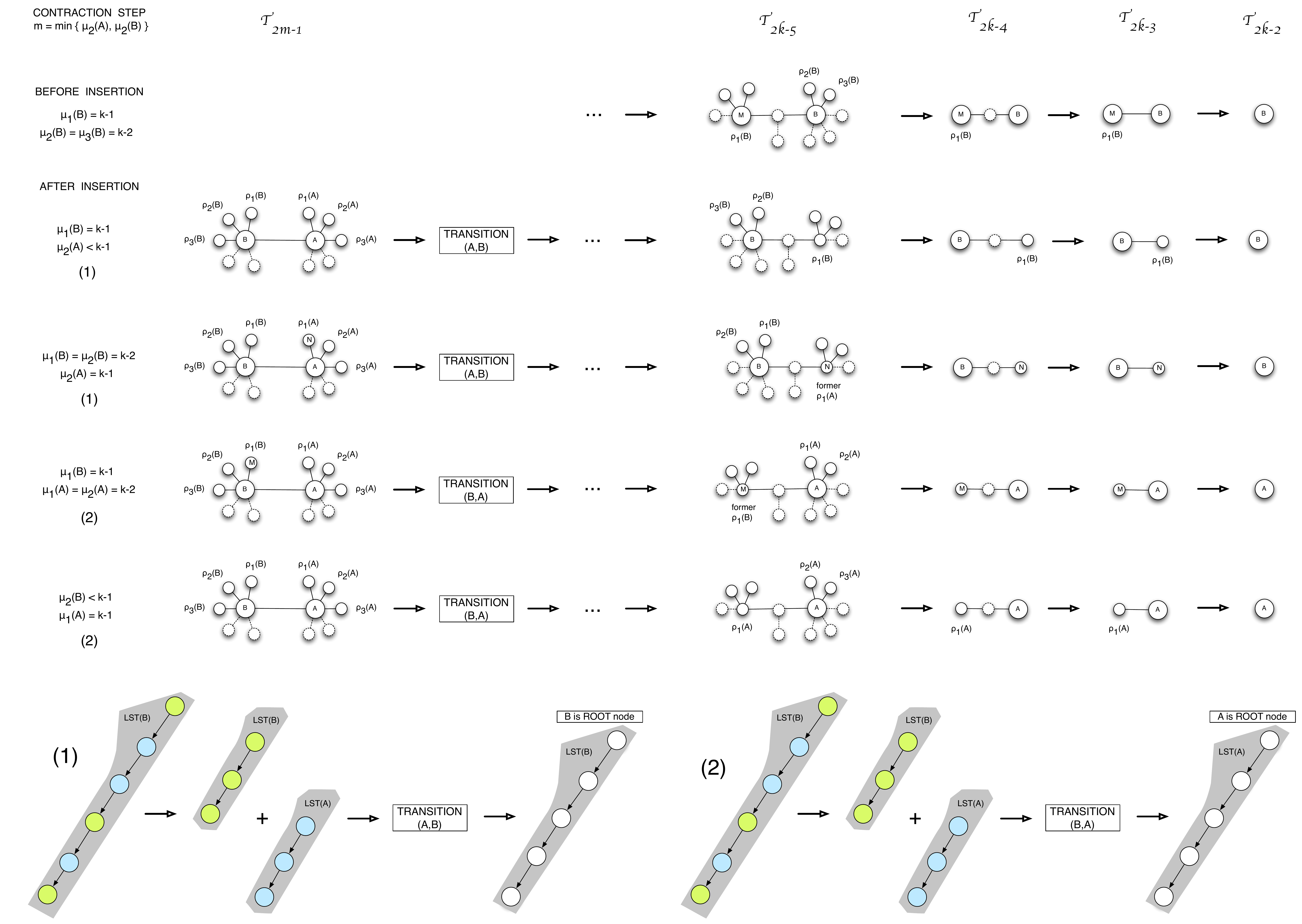}
\end{center}
\vspace{-0.2in}
\caption{Insert, Case 1, Part 2.}
\label{fig:insert1b}
\end{figure}

\clearpage

\begin{figure}[tb]
\begin{center}
\includegraphics[width=\textwidth]{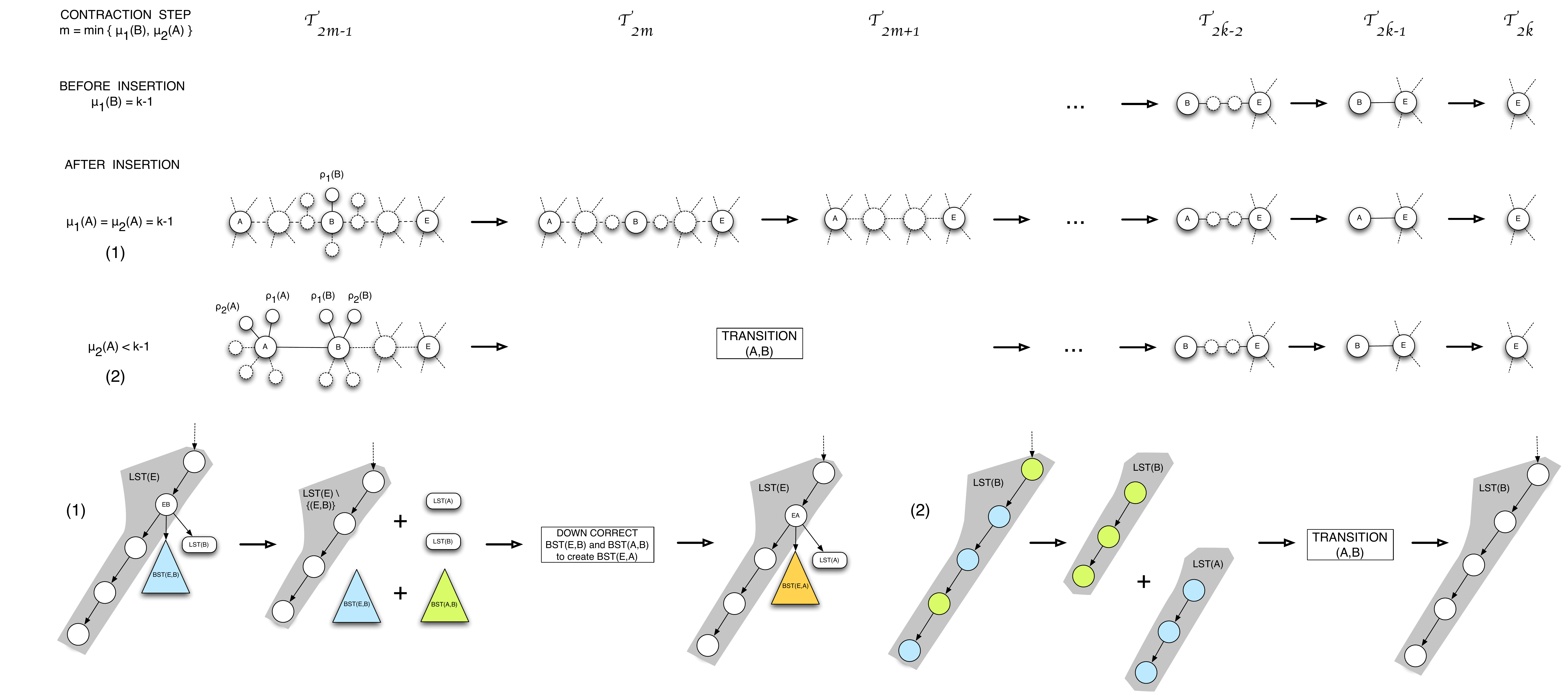}
\end{center}
\vspace{-0.2in}
\caption{Insert, Case 2.}
\label{fig:insert2}
\end{figure}

\clearpage

\begin{figure}[tb]
\begin{center}
\includegraphics[width=\textwidth]{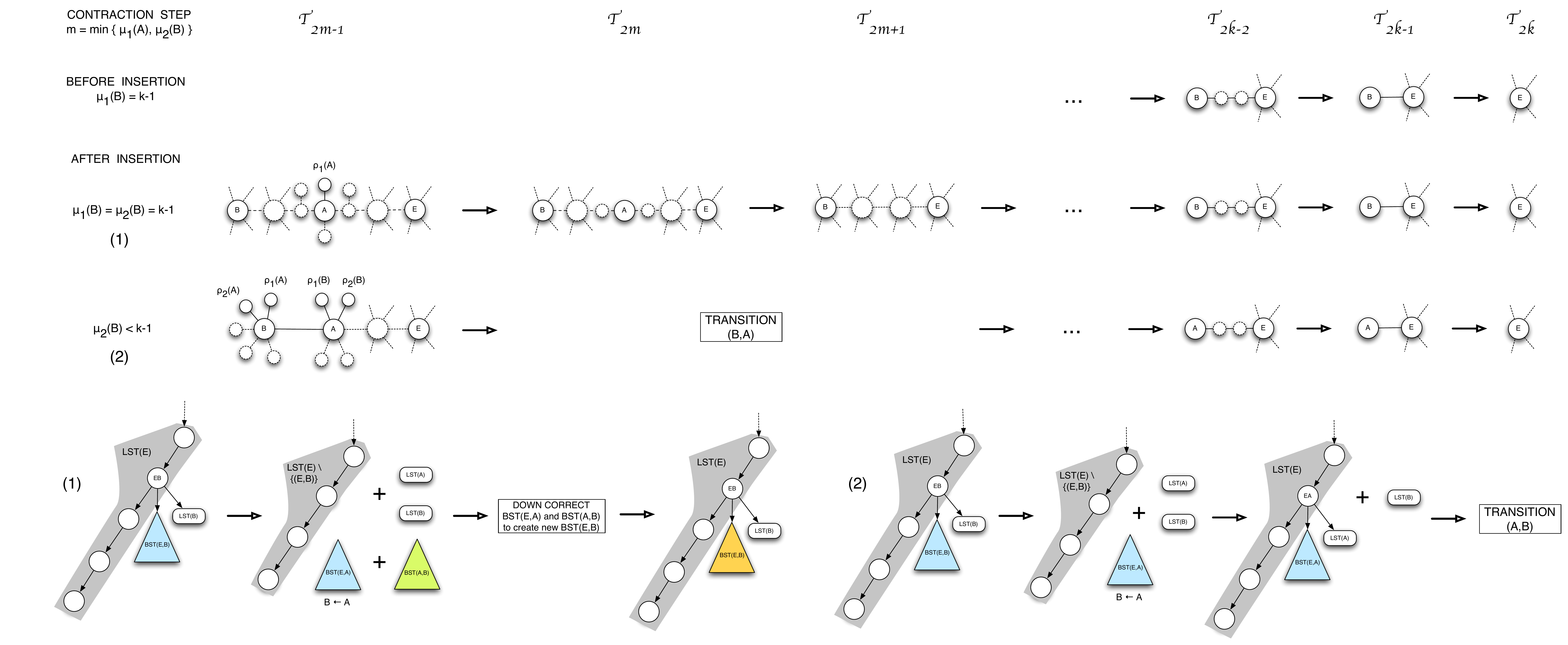}
\end{center}
\vspace{-0.2in}
\caption{Insert, Case 3.}
\label{fig:insert3}
\end{figure}

\clearpage

\begin{figure}[tb]
\begin{center}
\includegraphics[width=\textwidth]{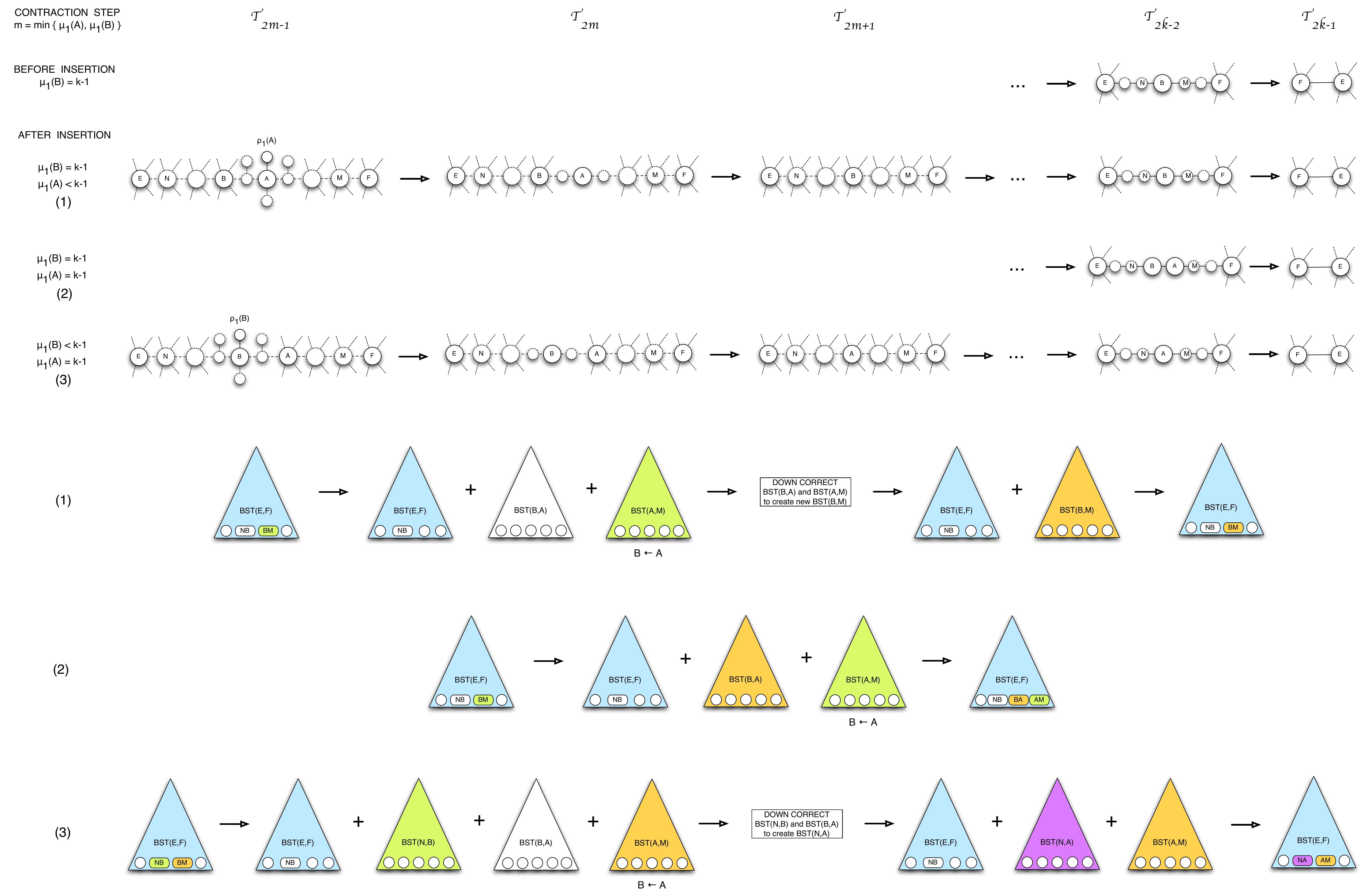}
\end{center}
\vspace{-0.2in}
\caption{Insert, Case 4.}
\label{fig:insert4}
\end{figure}

\clearpage

\begin{figure}[tb]
\begin{center}
\includegraphics[width=\textwidth]{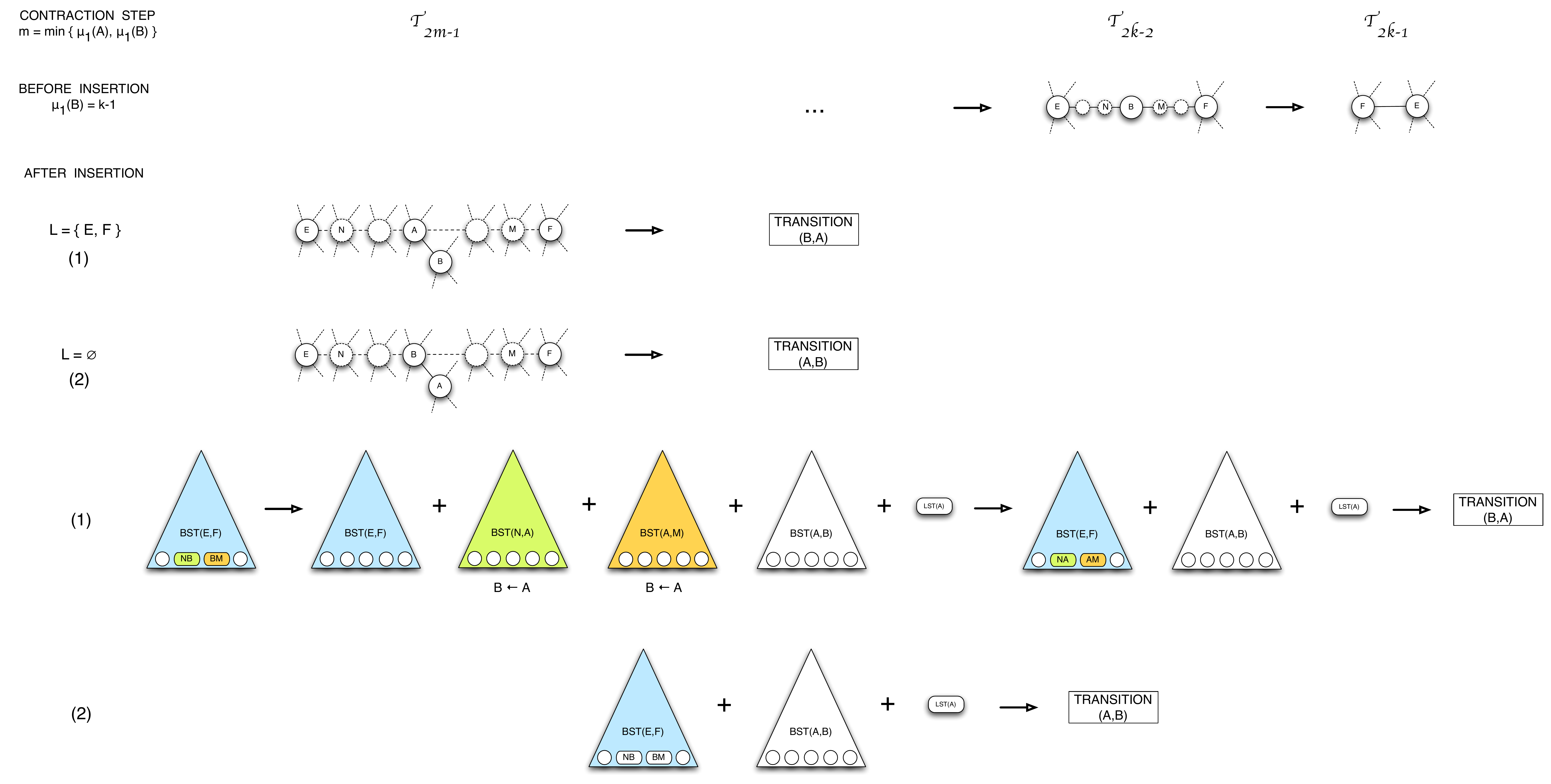}
\end{center}
\vspace{-0.2in}
\caption{Insert, Case 5.}
\label{fig:insert5}
\end{figure}

\clearpage

\begin{figure}[tb]
\begin{center}
\includegraphics[width=\textwidth]{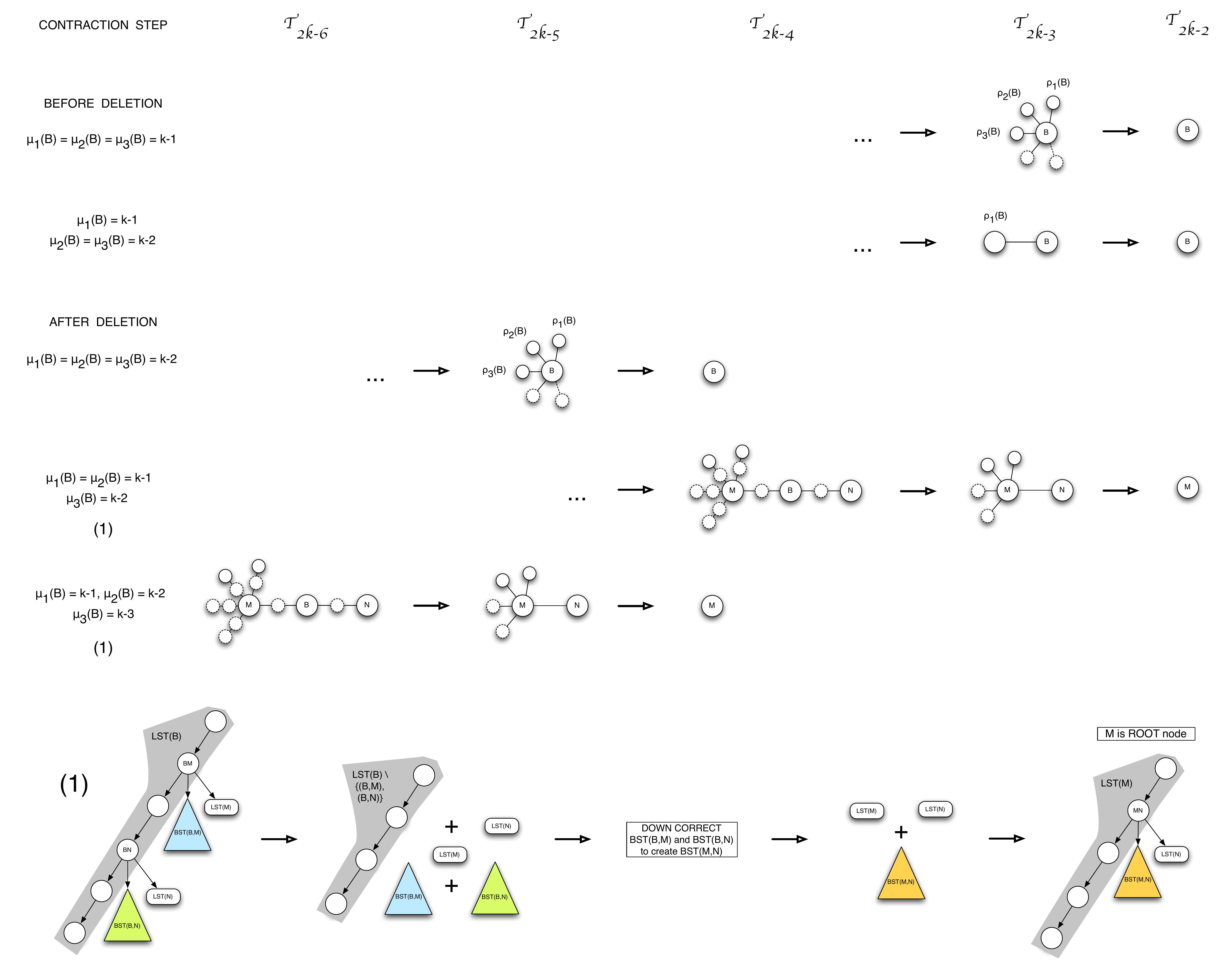}
\end{center}
\vspace{-0.2in}
\caption{Stabilize, Case 1.}
\label{fig:stabilize1}
\end{figure}

\clearpage

\begin{figure}[tb]
\begin{center}
\includegraphics[width=\textwidth]{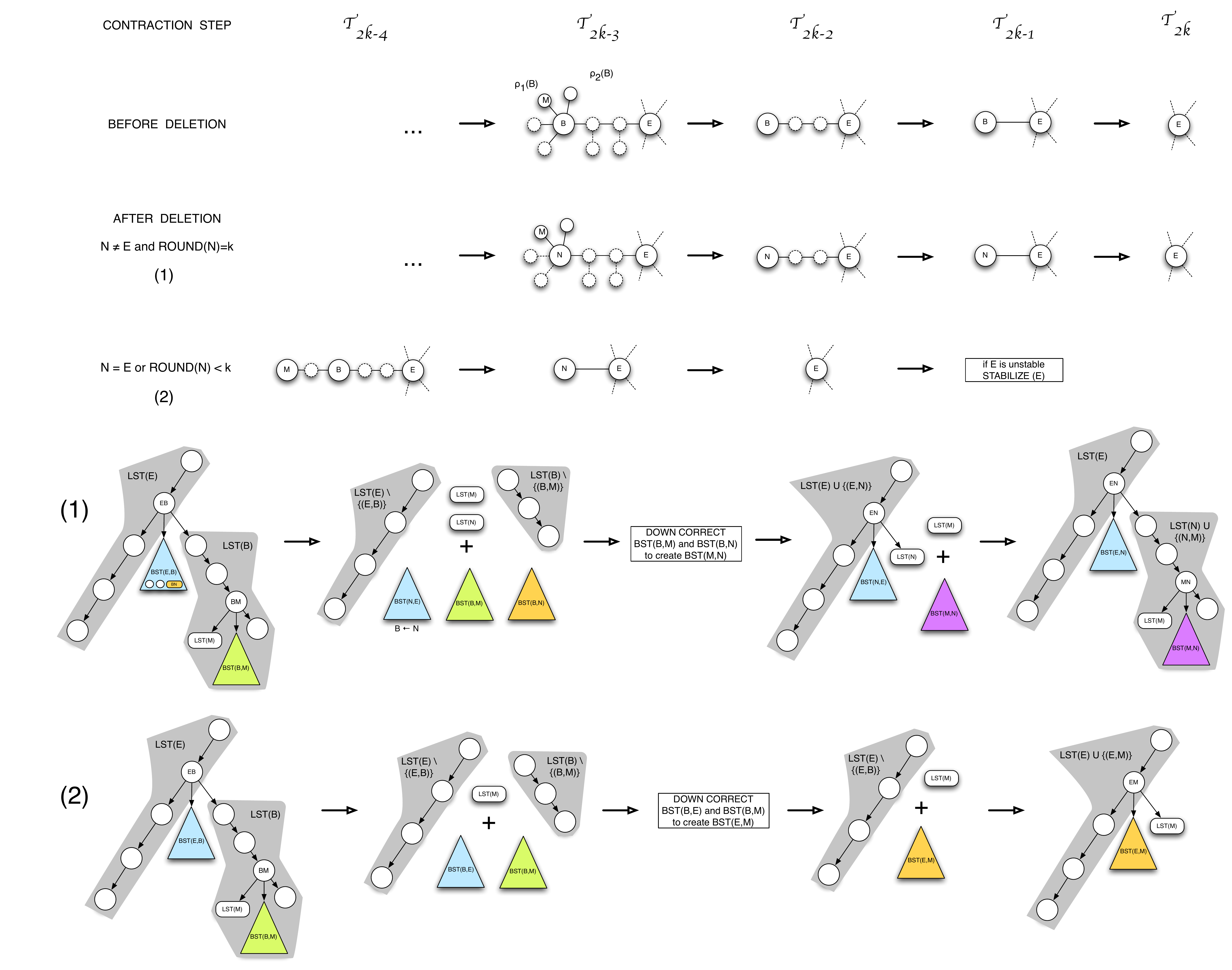}
\end{center}
\vspace{-0.2in}
\caption{Stabilize, Case 2.}
\label{fig:stabilize2}
\end{figure}

\clearpage

\begin{figure}[tb]
\begin{center}
\includegraphics[width=\textwidth]{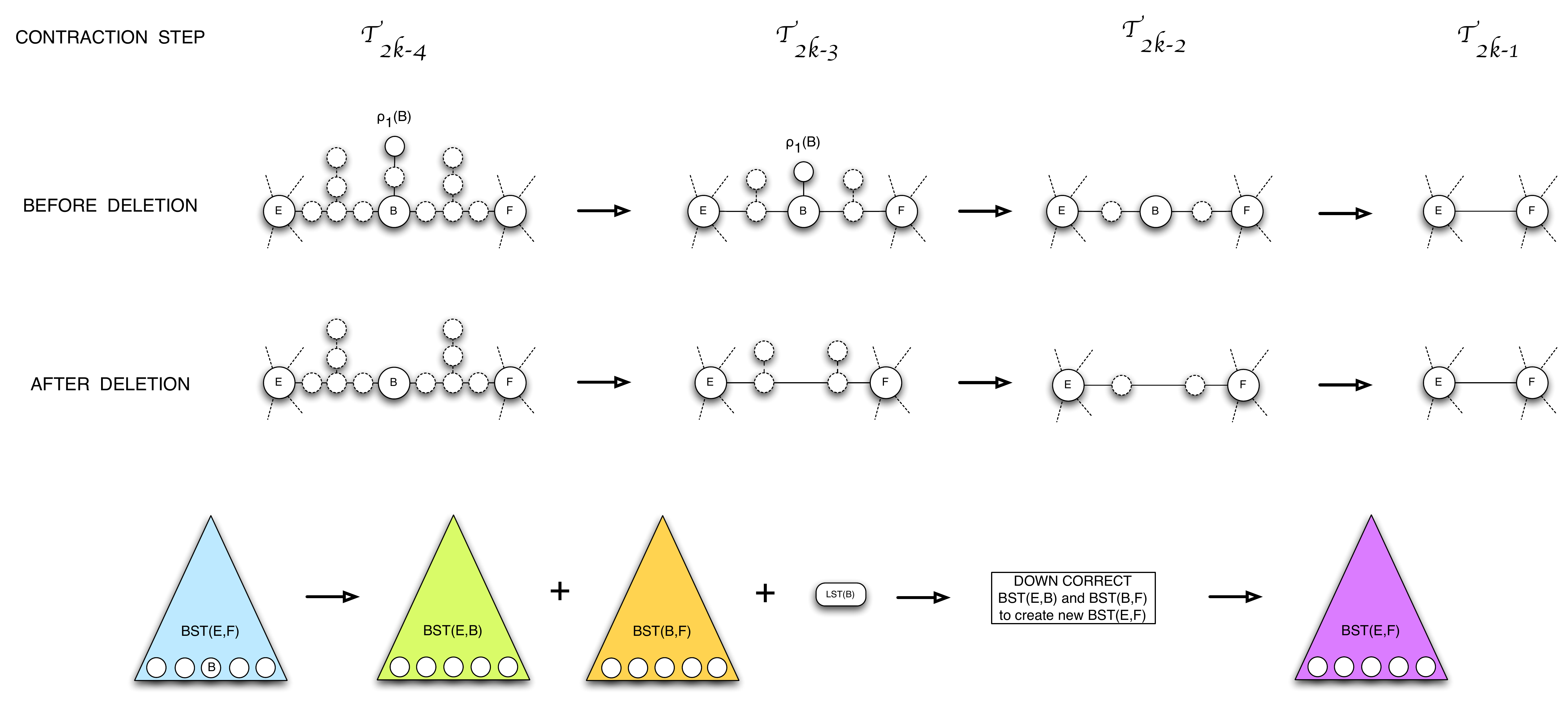}
\end{center}
\vspace{-0.2in}
\caption{Stabilize, Case 3.}
\label{fig:stabilize3}
\end{figure}

\clearpage

\begin{figure}[tb]
\begin{center}
\includegraphics[width=\textwidth]{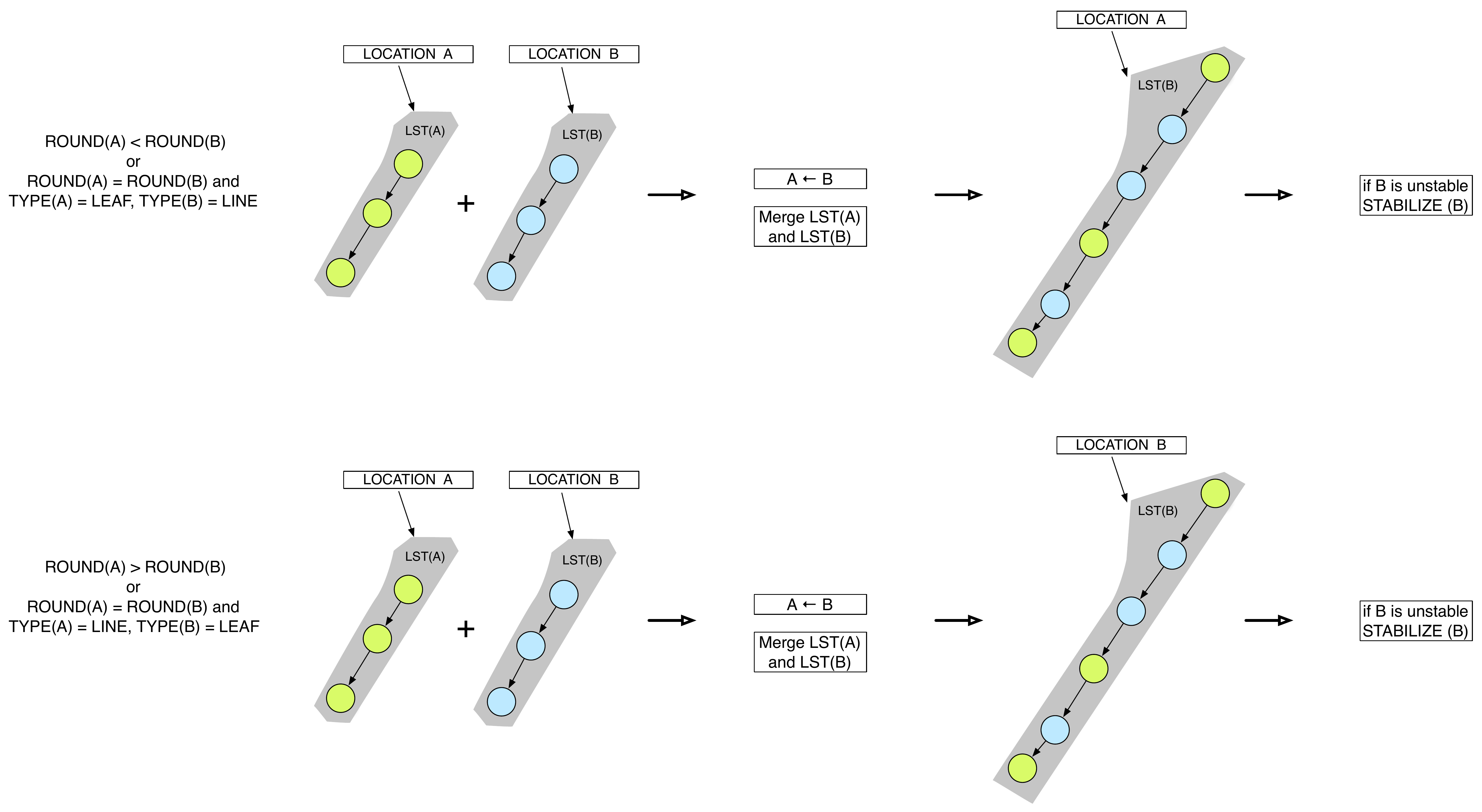}
\end{center}
\vspace{-0.2in}
\caption{Delete.}
\label{fig:delete1}
\end{figure}

\end{document}